\documentclass[3p,10pt]{elsarticle}

\makeatletter
\def\ps@pprintTitle{%
 \let\@oddhead\@empty
 \let\@evenhead\@empty
 \def\@oddfoot{}%
 \let\@evenfoot\@oddfoot}
\makeatother

\usepackage{hyperref}

\journal{}

\biboptions{authoryear}

\usepackage{psfrag}
\usepackage{amsmath}
\usepackage{amssymb}

\usepackage{amsfonts}
\usepackage{latexsym}
\usepackage{graphicx}

\usepackage{colortbl}
\usepackage{fancyhdr}
\usepackage{epstopdf}
\usepackage{float,bbm}
\usepackage{algorithm}
\usepackage{enumerate}
\usepackage{multicol}
\usepackage{graphicx}

\usepackage{cases}
\usepackage[usenames,dvipsnames,svgnames,table]{xcolor}

\newcommand{\PP }{G}
\newcommand{\W }{W}

\newcommand{\N}{\mathbb{N}}
\newcommand{\R}{\mathbb{R}}

\newcommand{\q}{z}

\usepackage{multirow}
\usepackage{rotating}

\newtheorem{theorem}{Theorem}
\newtheorem{definition}{Definition}
\newtheorem{assumption}{Assumption}
\newtheorem{proposition}{Proposition}
\newtheorem{lemma}{Lemma}
\newtheorem{corollary}{Corollary}
\newtheorem{remark}{Remark}
\newproof{proof}{Proof}

\usepackage{thmtools}
\declaretheorem[style=definition]{example}
\renewcommand\thmcontinues[1]{Continued}

\usepackage{soul}

\def \Z{\mathbb{Z}}

\usepackage{chngcntr}
\usepackage{apptools}
\AtAppendix{\counterwithin{lemma}{section}}
\AtAppendix{\counterwithin{theorem}{section}}
\AtAppendix{\counterwithin{assumption}{section}}
\AtAppendix{\counterwithin{example}{section}}
\AtAppendix{\counterwithin{definition}{section}}
\AtAppendix{\counterwithin{remark}{section}}
\begin{document}%

\begin{frontmatter}

\title{ A variational inequality framework for network games: Existence, uniqueness, convergence and sensitivity analysis.  \tnoteref{mytitlenote}}
\tnotetext[mytitlenote]{Research supported by the SNSF grant numbers P2EZP2$\_$168812, P300P2$\_$177805  and partially supported by ARO MURI W911NF-12-1-0509.}

\author{Francesca Parise}
\ead{parisef@mit.edu}
\author{Asuman Ozdaglar}
\ead{asuman@mit.edu}
\address{Laboratory for Information and Decision Systems, Massachusetts Institute of Technology, Cambridge, MA, USA.}

\begin{abstract}
We provide a unified variational inequality framework for the study of fundamental properties of the Nash equilibrium in network games. We identify several conditions on the underlying network (in terms of spectral norm, infinity norm and minimum eigenvalue of its adjacency matrix) that guarantee existence, uniqueness, convergence and continuity of equilibrium in general network games with multidimensional and possibly constrained strategy sets. We delineate the relations between these conditions and characterize classes of networks that satisfy each of these conditions.
\end{abstract}


\end{frontmatter}

\section{Introduction}

In many social and economic settings, decisions of individuals are affected by the actions of their friends, colleagues, and peers. Examples include adoption of new products and innovations, opinion formation and social learning, public good provision, financial exchanges and international trade. Network games have emerged as a powerful framework to study these settings with particular focus on how the underlying patterns of interactions, governed by a network, affect the economic outcome.\footnote{See for instance\  \cite{galeotti2010network,goyal2012connections,nagurney2013network,jackson2014games,bramoulle2015games,jackson2016networks}.  \cite{durlauf2004neighborhood} offers a survey of the theoretical and empirical literature on peer effects.} 
 For tractability reasons, many of the works in this area studied games with special structure (e.g., quadratic cost functions, scalar non-negative strategies) or special properties (e.g., games of strategic complements or substitutes for which the best response of each agent is increasing or decreasing in the other agents strategies).\footnote{See  \cite{bulow1985multimarket,milgrom1990rationalizability, vives1990nash}.}
These works relied on disparate techniques for their analysis which then lead to different focal network properties for existence and uniqueness of equilibria. For example, papers on network games of strategic complements typically relate  equilibrium properties  to
 the \textit{spectral radius} of the (weighted) adjacency matrix of the network (see e.g. { \cite{ballester2006s,jackson2014games,acemoglu2015networks}}) while  papers considering network games of strategic substitutes  highlight the role of the \textit{minimum eigenvalue} of the adjacency matrix  (see e.g. \cite{bramoulle2014strategic,allouch2015private}).
Moreover, there has been relatively little work on network games that feature neither strategic complements nor substitutes and  possibly involve multidimensional strategy sets  (e.g., models where agents decide their effort level on more than one activity, \cite{chen2017multiple}).

In this paper, we provide a unified framework based on a variational inequality reformulation of the Nash equilibrium to study equilibrium properties of network games including existence  and uniqueness, convergence of the best response dynamics and comparative statics. Our framework extends the literature in multiple dimensions. It applies to games of strategic complements, substitutes as well as games with  mixed strategic interactions. It provides a systematic understanding of which spectral properties of the network (spectral norm, minimum eigenvalue or infinity norm of the adjacency matrix) are relevant in establishing fundamental properties of the equilibrium. Moreover, it covers network games with multidimensional and constrained strategy sets. 

 Our work is built on the key observation that the analysis of network games can be performed in two steps.
In the first step we focus on the operator of the variational inequality associated with the game, which is typically referred to as the {\it game Jacobian} and we derive sufficient conditions on the network and on the cost functions to guarantee that  it possesses either one of three fundamental  properties: strong monotonicity (which is a stronger version than the strict diagonal concavity condition used in \cite{rosen1965existence}), uniform block P-function  and uniform P-function.\footnote{ While sufficient conditions for the former two properties in terms of the gradient of the game Jacobian have been studied in the literature (positive definiteness and the $P_\Upsilon$ condition discussed in \cite{scutari2014real}, respectively), we here suggest  a novel sufficient condition for the uniform P-function property (which we term uniform P-matrix condition). Our result extends the P-matrix condition used in  \cite[Proposition 3.5.9(a)]{facchinei2007finite} and is needed to  guarantee existence and uniqueness of the solution of a variational inequality without imposing boundedness assumptions on its closed convex set.}  Our sufficient conditions are formulated in terms of different network measures, i.e., spectral norm, infinity norm and minimum eigenvalue of the weighted adjacency matrix, as detailed in Assumptions \ref{ass:gen}a), \ref{ass:gen}b) and \ref{ass:gen}c). We highlight the relations between these conditions and  we show that for symmetric networks, the condition in terms of the minimum eigenvalue is the least restrictive and hence is satisfied by the largest set of networks, whereas for asymmetric networks conditions in terms of spectral norm and infinity norm cannot be compared and are satisfied by different sets of networks. While the conditions that involve spectral norm and minimum eigenvalue have appeared in the study of special network games, the condition in terms of the infinity norm is new and arises naturally for  games played over asymmetric networks where each agent is only influenced by a relatively small number of neighbors.

In the second step, we combine our sufficient conditions (connecting network properties to properties of the game Jacobian) with the theory of  variational inequalities to derive equilibrium properties.  A summary of our results is presented in Table \ref{fig:summary2}.  Specifically, we extend previous literature results on existence and uniqueness of the Nash equilibrium to games with multidimensional constrained strategy sets and  mixed strategic effects, we  provide sufficient conditions  to guarantee   convergence  of continuous and discrete  best response dynamics and finally we study how  parameter variations affect the Nash equilibrium.  Our first result in this context is to guarantee Lipschitz continuity of the Nash equilibrium under the block P-function property.\footnote{This result extends \cite[Theorem 2.1]{dafermos1988sensitivity} which holds under the more restrictive strong monotonicity condition.} We then build on a sensitivity analysis result for variational inequalities applied to general games given in  \cite{facchinei201012,friesz2016foundations} to establish conditions under which the Nash equilibrium is differentiable and we derive an explicit formula for its sensitivity.\footnote{
Contrary to previous works, our formula for the sensitivity of the Nash equilibrium depends  on  primal variables only. The only  result of such type that we are aware of is    \cite[Theorem 3.1]{dafermos1988sensitivity}.  The formula  therein is however obtained following geometric arguments and consequently depends on the orthogonal projection operator on the set of active constraints.} 

\begin{table}[h]
\begin{center}
\resizebox{\textwidth}{!}{
\begin{tabular}{|c|c|c|c|c|} \hline
\vphantom{\Large{{ \ding{51}}}}  &\textbf{ Assumption \ref{ass:gen}a} & \textbf{Assumption \ref{ass:gen}b } &\multicolumn{2}{c|}{\textbf{Assumption \ref{ass:gen}c}}  \\ 
\vphantom{\Large{{ \ding{51}}}}&& &$K^i(x)=\tilde K(x)\succeq 0$  & $n=1,K^i(x)>\nu> 0$ \\ 
& {\small ( Thm. \ref{thm:norm})}&{\small(Thm. \ref{thm:inf})} &{\small(Thm. \ref{thm:min})} &{\small (Thm. \ref{thm:scalar})} \\ \hline
\vphantom{\Large{{ \ding{51}}}}\textbf{Existence and} &\large{ \ding{51}}&\large{ \ding{51}} &\large{ \large{ \ding{51}}} & \large{ \large{ \ding{51}}}\\
\textbf{uniqueness}& {\small (Prop. \ref{prof:ex}a)}&{\small(Prop. \ref{prof:ex}b)} &{\small(Prop. \ref{prof:ex}a)} &{\small (Prop. \ref{prof:ex}c)} \\ \hline
\vphantom{\Large{{ \ding{51}}}}\textbf{BR dynamics} &\large{ \large{ \ding{51}}}&  \large{ \ding{51}}&\large{ \large{ \ding{51}}}   &\large{ \ding{51}} \\
\textbf{(continuous)} &{\small (Thm. \ref{cor:br})}&{\small(Thm. \ref{cor:br})} &{\small(see assumptions of Thm. \ref{thm:br_smon})} &{\small (see assumptions of Thm. \ref{lemma:br_sub}) } \\ \hline
\vphantom{\Large{{ \ding{51}}}} \textbf{BR dynamics}&\large{ \ding{51}} &\large{ \ding{51}}&\large{ \ding{55}} & \large{ \ding{55}}\\
\textbf{(discrete, simult.)}  & {\small(Thm. \ref{cor:br})}& {\small (Thm. \ref{cor:br})}& {\small(Ex. \ref{ex:br}A)  }& {\small(Ex. \ref{ex:br}A )} \\ \hline
\vphantom{\Large{{ \ding{51}}}}\textbf{Lipschitz} &\large{ \ding{51}}&\large{ \ding{51}}  & \large{ \large{ \ding{51}}}  & \large{ \large{ \ding{51}}} \\
\textbf{ continuity} & {\small (Thm. \ref{cor:Lipschitz})}& {\small(Thm. \ref{cor:Lipschitz})}&{\small(Thm. \ref{cor:Lipschitz}) }&{\small (Thm. \ref{cor:Lipschitz})} \\ \hline
\end{tabular}}
\end{center}
\caption{Summary of the relation between the considered network and cost conditions (i.e. Assumption \ref{ass:gen}a, \ref{ass:gen}b, \ref{ass:gen}c) and properties of the Nash equilibrium of the game.}
\label{fig:summary2}
\end{table}%

 To illustrate our theory we consider three running examples. First, we consider \textit{scalar linear quadratic network games}  for which the best response is a (truncated) linear function. These games have been extensively studied in the literature as illustrated in the seminal papers  \cite{ballester2006s}, where a model featuring local complements and global substitutes  is considered (mainly motivated by crime applications) and  \cite{bramoulle2014strategic}, where  games with  mixed strategic effects are considered. The results in \cite{bramoulle2014strategic} are derived by using the approach of  potential games, first introduced in \cite{monderer:shapley:96}. 
 Expanding on  \cite{bramoulle2014strategic}, we show that for scalar linear quadratic games the existence of a potential function is related to a condition on the symmetry of the network. When such condition is violated the  potential approach suggested in \cite{bramoulle2014strategic} cannot be applied. The variational inequality  approach presented in this work can be seen as an extension of the potential approach for cases when a potential function is not available. In fact,  potential games whose potential function is  strongly convex are a subclass of strongly monotone games. By leveraging on the theory of variational inequalities   (instead of convex optimization) we show how the results in \cite{ballester2006s} and \cite{bramoulle2014strategic} can be recovered and extended to  linear quadratic models that are not potential.   

A main feature of linear quadratic games, that significantly simplifies their analysis, is  the fact that the best response function is  linear. A few  works in the literature have extended the analysis of network games with scalar strategies  to nonlinear settings by focusing on cases where the  best response function is nonlinear but monotone. For example, \cite{belhaj2014network} considers a case where the best response function is increasing leading to a supermodular game. On the other hand, \cite{allouch2015private} focuses on a model of public good games where the best response function is decreasing. Conditions for existence and uniqueness in both these cases have been derived using  techniques tailored to the special structure of the problem. As  second motivating example,  we consider network games with scalar non-negative strategies where the \textit{best response is nonlinear and non-monotone}, representing e.g. races and tournaments where agents work hardest when network externalities are small  (neck and neck race) while they reduce their efforts if they fall behind (discouragement effect). 

Finally, all the games mentioned above feature scalar strategies. As third motivating example, we consider a network game where each agent decides on how much effort to exert in \textit{more than one activity} at the time. This model of multiple activities over networks was first suggested in \cite{chen2017multiple} and can be used to study agents engagement in activities that are interdependent such as crime and drug use (if activities are complements) or crime and education (if activities are substitutes).  \cite{chen2017multiple} focuses on the particular case when the network effects within the same activity are complements and derives results on existence and uniqueness of the Nash equilibrium by using a similar approach as in the single activity case studied in \cite{ballester2006s}.  Our framework enables the study of cases where  network effects within the same activity can  be substitutes and    interdependencies are present not only in the cost function but also in the strategy sets (e.g. because of budget constraints).

In addition to the papers cited above, our paper is most related to the seminal works of \cite{harker1990finite, facchinei2007finite, facchinei201012, scutari2014real,friesz2016foundations} and references therein. These works discuss the use of  variational inequalities to analyze equilibria of general games. Our novel contribution  is to  focus on network games and  unfold the effect of network properties on the game Jacobian and consequently on the equilibrium properties.

The only papers that we are aware of that study properties of the Nash equilibrium in network games by using variational inequalities are  \cite{ui2016bayesian}, \cite{melo2017variational} and \cite{naghizadehuniqueness}.  All these  works  consider network games with scalar non-negative strategies. 
 Specifically, \cite{ui2016bayesian} considers Bayesian network games with linear quadratic cost functions.  \cite{melo2017variational} considers existence, uniqueness and comparative statics for scalar network games with symmetric unweighted networks, by focusing on strong monotonicity of the game Jacobian.  For games with strategic substitutes, we show via a counterexample that strong monotonicity cannot be  guaranteed under the minimum eigenvalue condition, not even for scalar network games with  symmetric networks (see Example \ref{ex3}). In fact our analysis suggests that for scalar games of strategic substitutes, the natural property is the uniform P-matrix condition. In addition to this substantial difference, our paper is distinct from \cite{melo2017variational} in the following ways: i) we do not consider only strong monotonicity but  also uniform (block) P-functions, ii) we investigate how  properties of the game Jacobian relate to network properties considering not only the minimum eigenvalue but also  the spectral and  infinity norm, iii) we consider networks that might be asymmetric and agents with possibly multidimensional strategy sets, iv) we consider games with mixed strategic effects for which the best response  might not be monotone as a function of the neighbor aggregate, v) besides uniqueness and comparative statics, we also study convergence of best response dynamics.  The recent paper  \cite{naghizadehuniqueness} focuses on a special case of scalar network games and derives conditions in terms of the absolute value of the elements of the adjacency matrix for uniqueness of Nash equilibrium. This may be overly restrictive since taking the absolute value loses structural properties associated with games of strategic substitutes.  \cite{naghizadeh2017provision} studies public good network games by using  linear complementarity problems (which are a subclass of variational inequalities).
We  remark that variational inequalities have been used to study specific network economic models such as spatial price equilibrium models, traffic networks, migration models and  market equilibria, as reviewed for example  in \cite{nagurney2013network, facchinei2007finite}. In these settings however the network typically appears as part of the constraints and not in the cost function. Consequently, the  network effects studied in this paper do not appear.  Finally, \cite{xu2017nexus} uses variational inequality theory to study contest games whose cost function does not possess the aggregative structure considered in this paper (i.e. the cost depends on the strategies of the other agents individually and not on their aggregate).

It is also worth highlighting that  network games have  similarities with aggregative games. Equilibrium properties in aggregative games have been extensively studied in \cite{novshek1985existence,kukushkin1994fixed,jensen2005existence,acemoglu2013aggregate,jensen2010aggregative,kukushkin:04,cornes2012fully,dubey2006strategic}. 
Moreover, motivated by technological applications such as demand-response energy markets or communication networks, several papers have studied distributed dynamics for  convergence to the equilibria in aggregative games (see \cite{chen2014autonomous,koshal2012gossip,koshal2016distributed,grammatico:parise:colombino:lygeros:14, paccagnan2016aggregative}). The main difference   is that while in aggregative games each agent is affected by the same aggregate of the other agents strategies,  in network games this aggregate is agent and network dependent. 

 Our paper is organized as follows. In Section \ref{sec:net} we introduce the framework of network games and three motivational examples.
In Section \ref{sec:vi} we recall the connection between variational inequalities and game theory, we summarize properties that guarantee existence and uniqueness of the solution to a variational inequality and we present our technical results on the uniform P-matrix condition.
Section \ref{sec:overview} presents an overview of our results relating network and cost conditions to properties of the game Jacobian and illustrates these conditions for several networks of interest. 
Sections \ref{sec:ex},  \ref{sec:br_dynamics} and \ref{sec:comparative} exploit  the results of Section~\ref{sec:overview} to study existence and uniqueness, convergence of best response dynamics and comparative statics, respectively. Section \ref{sec:conc} concludes the paper. 
Some basic matrix properties and lemmas used in the main text are summarized in Appendix \ref{appendix:def}.
Appendix~\ref{sec:game_jacobian} provides the technical statements of the results anticipated in Section \ref{sec:overview} and Appendix \ref{appendix:proof} their proofs.
Appendix \ref{appendix:uniform} proves our technical result on the uniform P-matrix condition. Appendices \ref{appendix:br} and \ref{appendix:comp} expand on Sections \ref{sec:br_dynamics} and \ref{sec:comparative}, respectively. 
Definitions, examples and technical statements provided in the appendices are labeled with the corresponding letter.

\subsubsection*{Notation:}
 We denote the gradient of  a function $f(x):\R^n\rightarrow \R^d$  by $\nabla_x f(x) \in \R^{d \times n}.$ Given $a,b\in\N$, $\N[a,b]$ denotes the set of integer numbers in the interval $[a,b]$.
 $I_n$ denotes the $n$-dimensional identity matrix, $\mathbbm{1}_n$ the vector of unit entries and $e_i$ the $i$th canonical vector.
Given $A\in\mathbb{R}^{n\times n}$, $A\succ0$ ($\succeq0$) $\Leftrightarrow$ $x^\top A x>0~(\ge0),$ $\forall x\neq 0$,    $A_{(k,:)}$ denotes the $k$th row of $A$, $\rho(A)$ denotes the spectral radius of $A$ and $\Lambda(A)$ the  spectrum. $A\otimes B$ denotes the Kronecker product and $[A;B]:=[A^\top, B^\top ]^\top$. Given $N$ matrices $\{A^i\}_{i=1}^N$, $\mbox{blkd}(A^1,\ldots,A^N)=\mbox{blkd}[A^i]_{i=1}^N$ is the block diagonal matrix whose $i$th block is  $A^i$. Given $N$ vectors $x^i\in\mathbb{R}^{n}$, $x:=[x^1;\ldots;x^N]:=[x^i]_{i=1}^N:=[{x^1}^\top,\ldots ,{x^N}^\top]^\top\in\R^{Nn}$. Note that $x^i$ is the $i$th block component of $x$. We instead denote by $[x]_h$, for $h\in\N[1,Nn]$, the $h$-th scalar element of $x$.   The symbols $\ge_e$ and $\le_e$ denote element-wise ordering relations for vectors.  $\partial \mathcal{X}$ is the boundary of the set $\mathcal{X}$. $\Pi_{\mathcal{X}}^Q(x)$ denotes  the projection of the vector $x$ in the closed and convex set $\mathcal{X}$ according to the weighted norm $\|\cdot\|_Q$; $\Pi_{\mathcal{X}}(x):=\Pi_{\mathcal{X}}^I(x)$. For the definition of vector and matrix norms see Appendix A.

\section{Network games}
\label{sec:net}

Consider a network game $\mathcal{G}$ with $N$ players interacting over a weighted directed graph  described by the \textit{non-negative} adjacency matrix $\PP $. The element $\PP_{ij}$ represents the influence of agent $j$'s strategy on the cost function of agent $i$. We assume $\PP _{ii}=0$ for all $i\in\N[1,N]$. We say that $j$ is an in-neighbor of $i$ if $\PP _{ij}>0$. We denote by $\mathcal{N}^i$ the set of in-neighbors of agent $i$. From here on we   refer to in-neighbors simply as neighbors since we do not consider out-neighbors and  we use the terms network and graph interchangeably.
Each player $i\in\N[1,N]$ aims at selecting a vector strategy $x^i\in\R^n$ in its feasible set $\mathcal{X}^i\subseteq \R^n$  to \textit{minimize} a  cost function

\begin{equation}\label{gamey}
J^i(x^i,\q^i(x))
\end{equation}
which depends on its own strategy $x^i$ and on the aggregate of the neighbors strategies, $\q^i(x)$, obtained as the weighted
linear combination  

 $$ \q^i(x):=\sum_{j=1}^N \PP _{ij}x^j,$$
where  $x:=[x^i]_{i=1}^N\in\R^{Nn}$ is  a vector whose $i$-th block component is equal to the  strategy of agent~$i$.
The best response of agent $i$ to the neighbor aggregate $\q^i(x)$ is defined as

$$B^i(\q^i(x)):= \arg\min_{x^i\in\mathcal{X}^i} J^i(x^i,\q^i(x)),$$
where we recall that $\q^i(x)$ does not depend on $x^i$ since $\PP_{ii}=0$.
A set of strategies in which no agent has an incentive for unilateral deviations (i.e., each agent is playing a best response to  other agents strategies) is a Nash equilibrium,  that is,  $\{ x^{\star i} \}_{i=1}^N$, with $ x^{\star i}\in \mathcal{X}^i$ for all $i$, is a Nash equilibrium  if for all players $i\in\N[1,N]$ it holds $x^{\star i}\in B^i(\q^i(x^\star))$.

 Network games can be used to model a vast range of economic settings. Nonetheless, for tractability of analysis and for generating crisp insights  the literature adopted specific structural assumptions on best response functions. The next three  examples illustrate such assumptions  and show how our framework can unify previous literature results and enable consideration of richer economic interactions. 

\begin{example}[label=ex:lq] \textup{\textbf{(Scalar linear quadratic games)}}
Consider a  network game where each agent chooses a scalar  non-negative strategy $x^i\in \mathcal{X}^i=\R_{\ge 0}$ (representing for instance how much effort he exerts on a specific activity) to minimize the  linear quadratic cost function
\begin{equation}\label{eq:quad_cost}
J^i(x^i, z^i(x))=\frac12 (x^i)^2+ [K^i z^i(x)-a^i]x^i,
\end{equation}
with $K^i, a^i\in\R$. 
{Network games with payoffs of this form have been widely used in the literature to study various economic settings including private provision of public goods (e.g., innovation, R\&D, health investments in particular vaccinations; see \cite{bramoulle2007public,bramoulle2014strategic}) and games with local payoff complementarities but global substitutability  (e.g., \cite{ballester2006s} suggests that  the cost of engaging in a criminal activity is lower when friends participate in such behavior,  yet may be higher when overall  crime levels increase).
For these games, the best response of  agent~$i$ is given by a  (truncated) linear function of the neighbor aggregate $z^i(x)$:
\begin{equation}\label{eq:br_quad}B^i(\q^i(x)):=\max\{0, a^i- K^i z^i(x)]\}.\end{equation}
 From Eq. \eqref{eq:br_quad} it is evident that the payoff parameter $K^i$ captures how much the neighbor aggregate $z^i(x)$ affects the  equilibrium strategy of agent $i$. 
 When the $K^i$ are the same for all $i$, we say that the game has \textup{homogeneous weights}. We remark that even in this case, the  neighbor aggregate $z^i(x)$  may be heterogeneous across agents.   We also note that if for each agent $i$ $K^i\ge 0$ ($K^i\le 0$),  this is a game of strategic substitutes (complements), meaning that each agent's best response   decreases (increases) in  other agents actions. 

 Linear quadratic network games have been studied in \cite{bramoulle2014strategic} using the theory of potential games. This approach entails constructing a proxy maximization problem  whose stationary points provide the set of Nash equilibria. 
 By \cite{monderer:shapley:96} a necessary condition for the existence of an exact potential function  is  
\begin{equation}\label{pot}\frac{\partial^2{J^i(x^i,z^i(x))}}{\partial x^ix^j}= \frac{\partial^2{J^j(x^j, z^j(x))}}{\partial x^j x^i}.
\end{equation}
For linear quadratic network games, this amounts to the  restriction that
\begin{equation}\label{potQ}
\frac{\partial^2{J^i(x^i, z^i(x))}}{\partial x^ix^j}=K^i \PP_{ij} = K^j\PP_{ji} =\frac{\partial^2{J^j(x^j, z^j(x))}}{\partial x^j x^i}.\end{equation}
\cite{bramoulle2014strategic} focuses on symmetric networks $G=G^\top$ and (with the exception of the last section)\footnote{  Section 6 in \cite{bramoulle2014strategic} shows that for heterogeneous weights, there exists a linear change of variables such that the game \textit{in the new coordinates} is potential. More in general, for linear quadratic games this happens if  there exists a vector $\beta\in\R^N_{>0}$ such that 
 $\frac{K^i}{\beta^i} \PP_{ij} = \frac{ K^j}{\beta^j}\PP_{ji}$, which is a generalization of \eqref{pot}. See Lemma \ref{lemma:pot} for more details.} on homogeneous weights. In this case  condition \eqref{potQ} is met, enabling the use of the potential approach. Our subsequent analysis studies asymmetric networks and heterogeneous weights.\footnote{ Instead, we do not explicitly consider cases where the externalities may have different signs for the same agent but different neighbors. This case, corresponds to a matrix $G$ with both positive and negative entries. Most of our results hold also in this case, but we find that this additional generality obfuscates the overall presentation. } 
}
\end{example}

The linear quadratic model described above has been extensively studied in the literature, partially due to the fact that the linearity of its best response function allows for a simple yet informative analysis. As  the second main example, we describe a variation of the previous model, that features a nonlinear  best response.

\begin{example}[label=competition] \textup{\textbf{(Races and tournaments)}}
Consider a network game where each player has a scalar strategy $x^i\in[a^i, b^i]$ and a  nonlinear best response (as a function of the neighbor aggregate, denoted by $z^i$ for simplicity)
\begin{equation}\label{eq:non-mon}
B^i(z^i)=\min\{a^i+\phi_i(z^i),b^i\}.
\end{equation}
We assume $0<a^i<b^i$, $\phi_i(0)= 0$ and $\phi_i(z^i)\ge 0$. Special cases of this model have been considered in the literature for example in \cite{belhaj2014network}  with the additional assumption $\phi_i'\ge0$ (so that the best response is increasing in other agents actions, leading to a game of strategic complements) or in \cite{allouch2015private} with $\phi_i'\le 0$ (so that the best response is decreasing in other agents actions, leading to a game of strategic substitutes).\footnote{ More in detail, therein the author focuses on a public good game for which $a^i=0, b^i=\infty$ and  $\phi_i(z^i)=\gamma_i(w_i+z^i)-z^i$ where $\gamma_i$ is consumer'$i$ Engel curve and $w_i$ is agent's $i$ income. } 
We do not impose any monotonicity assumption on  $\phi_i$  and  focus on cases where the sign of $\phi_i'$ may change (either for different values of $z^i$ or across different agents). 

To illustrate the richer class of strategic interactions that can be modeled by using non-monotone best response functions, we consider the special case where  $a^i=a$, $b^i=b$ and $\phi_i(z^i)$ has the form
\begin{equation}\label{eq:non-mon-phi}
\phi_i(z^i)=\phi(z^i)=\gamma z^i(b-z^i) \quad \mbox{for all } i,
\end{equation}
with $\gamma>0$.  For simplicity, we additionally assume that the network $\PP$ is such that $\sum_j \PP_{ij}=1$ (so that the neighbor aggregate $z^i(x)$ is given by a convex combination of the neighbor strategies). 
{The corresponding best response function  is illustrated in  Figure \ref{fig:example2} and can be used to model  races and tournaments. In the initial phase, when $z_i\le \frac{b}{2}$, the players increasing effort motivates a neck and neck race which is then followed by a second phase,  when $z_i\ge \frac{b}{2}$, where agent's effort level declines capturing a discouragement effect.}
 \begin{figure}[h]
\begin{center}
\includegraphics[height=0.23\textwidth]{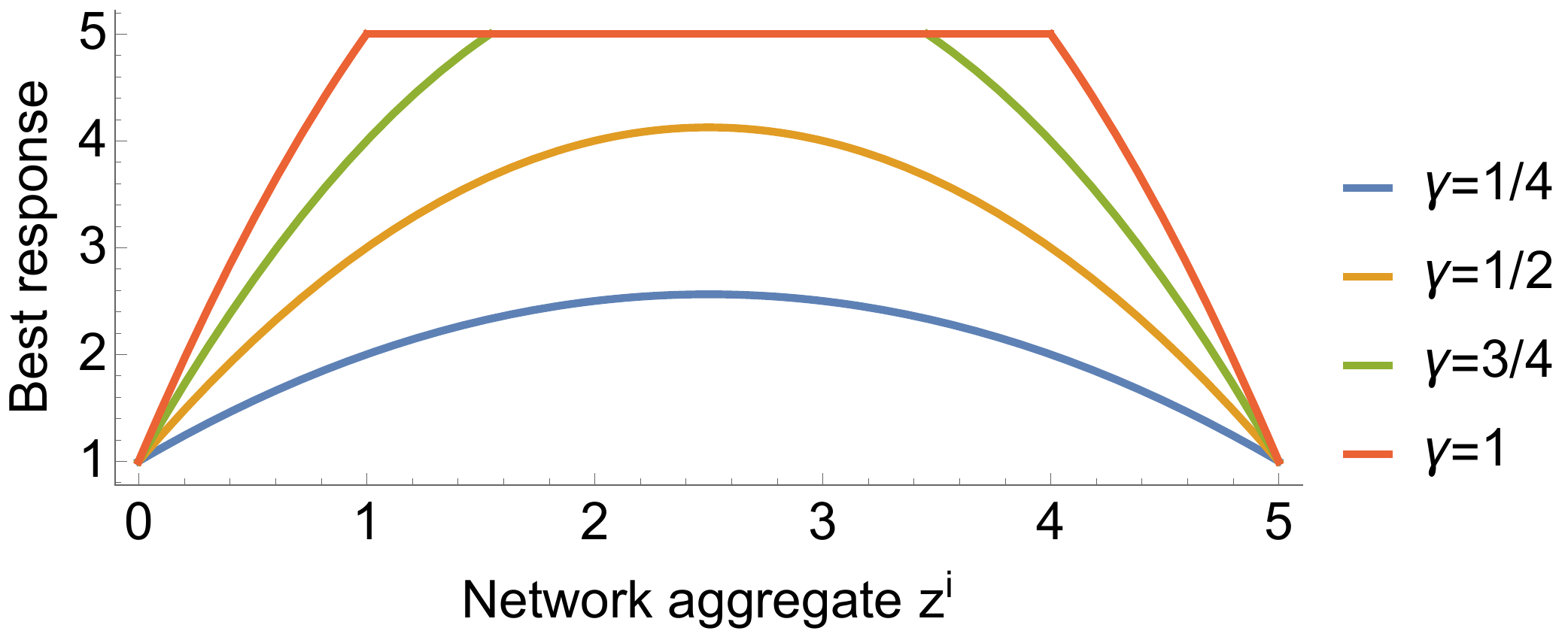}
\end{center}
\caption{ The best response function given in \eqref{eq:non-mon} for $\phi_i$ as in \eqref{eq:non-mon-phi} plotted for  $a=1$, $b=5$ and different values of $\gamma$.}
\label{fig:example2}
\end{figure}
 \end{example}
 
   In Examples \ref{ex:lq} and \ref{competition}  agents have a single decision variable. Our last motivating example features agents that engage in  multiple activities, as studied in \cite{chen2017multiple}.

 \begin{example}[label=ex:multiple] \textup{\textbf{(Multiple activities in networks)}}  Consider a network game where each player $i$ 
 has a strategy vector $x^i=[x^i_A, x^i_B]\in\R^2$ with $x^i_A, x^i_B$ representing his level of engagement in two interdependent activities $A$ and $B$, such as crime and education. We assume that each player $i$ has  lower and upper constraints for its engagement in each activity ($\underline{x}^i_A, \bar{x}^i_A,\underline{x}^i_B, \bar{x}^i_B$) and an overall lower and upper budget constraint ($\underline{x}^i_\textup{tot}, \bar{x}^i_\textup{tot}$) so that 
$$\mathcal{X}^i:=\{x^i=[x^i_A,x^i_B]\in\R^2 \mid \underline{x}^i_A\le x^i_A \le \bar{x}^i_A,\quad  \underline{x}^i_B\le x^i_B \le \bar{x}^i_B,\quad \underline{x}^i_\textup{tot}\le x^i_A+x^i_B\le \bar{x}^i_\textup{tot} \}.$$
This generalizes  \cite{chen2017multiple}, where no constraints are considered, and \cite{belhaj2014competing} where a strict budget constraint is imposed (i.e., $\underline{x}^i_\textup{tot}= \bar{x}^i_\textup{tot}=1$) so that effectively the strategy of each agent is a scalar (since $x^i_B=1-x^i_A$). 
Each agent selects its level of engagement in activities $A$ and $B$ to maximize the following quantity
$$\underbrace{a^i_Ax^i_A -\frac12(x^i_A)^2+\delta x^i_A z^i_A(x)}_{\textup{net proceeds from activity A}}+ \underbrace{a^i_Bx^i_B-\frac12(x^i_B)^2+\delta x^i_B z^i_B(x)}_{\textup{net proceeds from activity B}}  + \underbrace{\mu x^i_Az^i_B(x)+\mu x^i_Bz^i_A(x)- \beta^i x^i_Ax^i_B \vphantom{\frac12}}_{\textup{ interdependence  of activities}},$$
{where the parameter $\delta$ weights the effect of the neighbor aggregate within each activity,  $\mu$ weights the effect of the neighbor aggregate  across different activities and $\beta^i$ captures the interdependence of the two activities for each agent $i$.  \cite{chen2017multiple}  assumes  $\delta>0$ so that the effort of each agent and its neighbors are strategic complements within each activity, while $\beta^i\in(-1,1)$ can be   negative (modeling two complementary activities  such as crime and drug use) or  positive (modeling two substitutable activities such as crime and education).}\footnote{ \label{fot:sign} For the model to be consistent from an  economic perspective, the sign of $\mu$ should be consistent with the signs of $\delta$ and $\beta^i$ in the sense that  $\textup{sign}(\mu)=-\textup{sign}(\delta)\textup{sign}(\beta^i)$. To see this, consider the case  $\delta<0$ and  $\beta^i>0$. In this case an increase of  $z^i_A$ leads to a decrease in $x^i_A$ (since $\delta<0$) which leads to an increase in $x^i_B$ (since $\beta^i>0$). Overall, $z^i_A$ and $x^i_B$ are therefore complements implying that $\mu$ must be positive.} In our analysis, we will also consider the case $\delta<0$.
\end{example}

\section{Connection between game theory and variational inequalities}
\label{sec:vi}

The main goal of this paper is to provide a unified framework to study existence and uniqueness of the Nash equilibrium, convergence of the best response dynamics and  continuity of the Nash equilibrium with respect to parametric variations for general network games.  Throughout the paper we use the following assumption.

\begin{assumption}\label{cost}
The set $\mathcal{X}^i\subseteq \R^n$ is nonempty, closed and convex for all $i\in\N[1,N]$.
The function $ J^i(x^i,\q^i(x))$ is  continuously differentiable and  convex in $x^i$  for all $i\in\N[1,N]$  and for all $x^{j}\in\mathcal{X}^{j}$, $j\in\mathcal{N}^i$.  Moreover, $J^i(x^i, z^i)$ is  twice  differentiable in $[x^i; z^i]$ and $\nabla_{x^i} J^i(x^i, z^i)$ is Lipschitz in $[x^i; z^i]$. 
\end{assumption}

Our approach relies on the theory of variational inequalities as defined next.

\begin{definition}[Variational Inequality (VI)]\label{vi}
A vector $\bar x\in\R^{d}$ solves the variational inequality VI$(\mathcal{X},F)$ with set $\mathcal{X}\subseteq\R^{d}$ and operator 
$F:\mathcal{X} \rightarrow \R^{d}$ if and only if

\begin{equation}
F(\bar x)^\top (x-\bar x) \ge 0, \quad \hbox{for all } x\in\mathcal{X}.\label{eq:VI}
\end{equation}
\end{definition}

In the following we consider the VI with set 

\begin{equation}\label{eq:X}
\mathcal{X}:=\mathcal{X}^1\times\ldots\times\mathcal{X}^N\subseteq\R^{Nn}
\end{equation}
 obtained as the cartesian product of the local strategy sets $\mathcal{X}^i$ and  operator $F:\mathcal{X}\rightarrow \R^{Nn}$ whose $i$-th block component $F^i(x)$ is the gradient of the cost function of agent $i$ with respect to its own strategy, i.e.,

\begin{align}\label{eq:F}
F(x)&:=[F^i(x)]_{i=1}^N:=[\nabla_{x^i} J^i(x^i,\q^i(x))^\top]_{i=1}^N.
\end{align}
This operator is sometimes referred to as the \textit{game Jacobian}. The relevance of this VI in characterizing Nash equilibria of general games (i.e., not necessarily network games) comes from the following well-known relation, see e.g. \cite[Proposition 1.4.2]{facchinei2007finite}. 

\begin{proposition}[VI reformulation]\label{prop:kkt}
Suppose that Assumption \ref{cost} holds. A vector of strategies $x^\star$ is a Nash equilibrium for the game  $\mathcal{G}$ if and only if it solves the VI$(\mathcal{X},F)$ with $\mathcal{X}$ as in \eqref{eq:X} and $F$ as in \eqref{eq:F}. 
\end{proposition}

{ Proposition \ref{prop:kkt} can be seen as a generalization of  potential games as introduced in \cite{monderer:shapley:96}. In fact, it follows from \cite[Lemma 4.4]{monderer:shapley:96} that a game is potential with potential function $U(x)$ if and only if $\nabla_x U(x)=F(x)$.  In other words, a game has an exact potential if and only if the game Jacobian $F(x)$  is integrable. In such case, by the minimum principle, the  VI condition given in \eqref{eq:VI}  coincides with  the necessary optimality conditions for the optimization problem $\min_{x\in\mathcal{X}} U(x)$, whose  stationary points  are the Nash equilibria of the game, see e.g. \cite[Lemma 1]{bramoulle2014strategic}.   The VI approach enables the analysis of games that are not potential (i.e. games for which $F(x)$ is not integrable).   }
Specifically, Proposition \ref{prop:kkt} allows one to analyze the Nash equilibria in terms of properties of the game Jacobian $F$ and of the set $\mathcal{X}$. To this end, we recall the following definitions.  
\begin{definition}\label{def:mon}
An operator $F:\mathcal{X}\subseteq \R^{Nn}\rightarrow \R^{Nn}$ is
\begin{itemize}
\item[a)] \textup{Strongly monotone:} if there exists $\alpha>0$ such that $(F(x)-F(y))^\top(x-y)\ge \alpha \|x-y\|_2^2$ for all $x,y\in\mathcal{X}$. 
\item[b)] \textup{A uniform block P-function with respect to the partition $\mathcal{X}=\mathcal{X}^1\times\ldots\times\mathcal{X}^N$:} if there exists $\eta>0$ such that $\max_{i\in\N[1,N]} [F^i(x)-F^i(y)]^\top[x^i-y^i]\ge \eta \|x-y\|_2^2$
for all $x,y\in\mathcal{X}$.
\item[c)] \textup{A uniform P-function:} if there exists $\eta>0$ such that 
$\max_{h\in\N[1,Nn]} [F(x)-F(y)]_h[x-y]_h\ge \eta \|x-y\|_2^2$ for all $x, y\in\mathcal{X}$.
\end{itemize}
\end{definition}
The properties in the previous definition are stated in decreasing order of strength, as illustrated in the first line of Figure \ref{fig:rel}.\footnote{
 In the literature a block P-function is typically referred to as a P-function with respect to the partition $\mathcal{X}=\mathcal{X}^1\times\ldots\times \mathcal{X}^N$; we added the term ``block'' to avoid any confusion. We note that for $n=1$ (i.e., for scalar games)  the definitions of uniform block P-function and uniform  P-function coincide.} 
{ We note that in the case of potential games  \textit{strong monotonicity} of $F(x)$ is equivalent to strong convexity of the potential function $U(x)$.} For general games (i.e., not necessarily potential games),  such condition  is a slightly  stronger  version of the \textit{diagonal strict concavity} condition used in the seminal work of \cite{rosen1965existence} (see more details in Appendix \ref{appendix:def}). The reason why we focus on strong monotonicity in this work is that it allows us to prove existence of the Nash equilibrium without assuming compactness, as instead assumed in \cite{rosen1965existence}.
Our interest in the properties listed in Definition \ref{def:mon} stems from the following classical result for 
existence and uniqueness of the solution of a VI, see e.g.  \cite[Theorem 2.3.3(b) and Proposition 3.5.10(b)]{facchinei2007finite}.\footnote{Statement c) is a consequence of  
statement b). In fact if $\mathcal{X}$ is a rectangle then it can be partitioned as $\mathcal{X}=\mathcal{X}^1\times \hdots\times \mathcal{X}^{Nn}$ with $\mathcal{X}^h\subseteq\R$ for all $h\in\N[1,Nn]$ and if $F$ is a uniform P-function then it is a uniform block P-function with respect to a block  partition with blocks of dimension one. We here write this condition separately to avoid any confusion with the block partition $\mathcal{X}=\mathcal{X}^1\times \hdots\times \mathcal{X}^{N}$ induced by the $N$ players.}

\begin{proposition}[Existence and uniqueness for VIs]  \label{prof:ex}
Consider the VI$(\mathcal{X},F)$ where $F$ is  continuous and  $\mathcal{X}$ is nonempty, closed and convex. The  VI$(\mathcal{X},F)$  admits a unique solution under any of the following statements:
\begin{itemize}
\item[a)] $F$ is strongly monotone.
\item[b)] The set $\mathcal{X}$ is a cartesian product $\mathcal{X}^1\times\ldots\times \mathcal{X}^N$ and $F$ is a uniform block P-function with respect to the same partition.
\item[c)] The set $\mathcal{X}$ is a  rectangle and $F$ is a uniform P-function.
\end{itemize}
\end{proposition}

Note that the stronger  the condition on $F$ is, the weaker  the requirement on the set $\mathcal{X}$ is.
Proposition \ref{prof:ex} together with Proposition \ref{prop:kkt} allows one to derive sufficient conditions for existence and uniqueness of the Nash equilibrium in terms of properties of the game Jacobian and of the constraint sets. Our main contribution in Section \ref{sec:overview}  is to derive sufficient conditions that guarantee that the game Jacobian of a network game has one of the properties listed in  Definition \ref{def:mon}, by imposing conditions on the cost functions of the players and on the spectral properties of the network. To this end, we exploit  sufficient conditions for the properties in Definition \ref{def:mon}  to hold in terms of the gradient $\nabla_x F(x)$ as detailed in the next subsection.

\subsection{\textbf{Sufficient conditions in terms of $\boldsymbol{\nabla_x F(x)}$}}

 We start by introducing the definitions of \textit{P-matrices}, \textit{P$_\Upsilon$ condition} and \textit{uniform P-matrix condition}.
\begin{definition}[P-matrix] \cite[Theorem 3.3]{fiedler1962matrices} \label{def:P_matrix}
A matrix $A\in\R^{d\times d}$ is a P-matrix if all its principal minors have positive determinant. Equivalently,  $A$ is a  P-matrix if and only if for any $w\in\R^d$, $w\neq0$ there exists  a diagonal matrix $H_w\succ 0$ such that $w^\top H_w A w>0$.
\end{definition}
 A special class of P-matrices are positive definite matrices. In fact, if a matrix $A$ is positive definite then Definition \ref{def:P_matrix} holds with $H_w=I$. The opposite holds true if the matrix $A$ is symmetric.\footnote{ It is shown in \cite[Theorem 3.3]{fiedler1962matrices} that if $A$ is a  $P$-matrix then all its real eigenvalues are positive. If $A$ is symmetric all its eigenvalues are real. These two properties imply that $A$ is positive definite.}

 The sufficient conditions detailed in the following Proposition \ref{prop:suff} amount to ensuring that $\nabla_xF(x)$, which is  a matrix valued function,  possesses  suitable properties ``uniformly'' in $x$. For example, strong monotonicity of $F(x)$ can be guaranteed if the gradient $\nabla_xF(x)$ is uniformly positive definite. Similarly, to guarantee that $F(x)$ is a uniform P-function we show that it is sufficient to assume that $\nabla_xF(x)$ satisfies what we term a ``uniform P-matrix condition'' which is given in full generality in  Definition \ref{def:P_matrix_uniform} in  the appendix. We here report the corresponding definition for the affine case $F(x)=Ax+a$. Intuitively, this  corresponds to assuming that  $\nabla_xF(x)$ is a ``uniform'' P-matrix for all values of $x$. 
 \begin{definition}[Uniform P-matrix condition - affine case]
Consider an operator $F(x)=Ax+a$ where $ A\in\R^{d\times d}, a\in\R^d$. Then $\nabla_x F(x)$ satisfies the uniform P-matrix condition if and only if there exists $\eta>0$ such that for any $w\in\R^d$ there exists a diagonal matrix $H_w\succ 0$ such that $w^\top H_w A w \ge \eta \|w\|^2$ and $\max_w \lambda_{\textup{max}}(H_w)<\infty$. 
\end{definition}
 
The condition  above is formulated in terms of the full matrix $\nabla_x F(x)$ which, in the case of agents with multidimensional strategies, has dimension $Nn\times Nn$. The $P_\Upsilon$ condition introduced next is instead used to summarize the effect that each pair of agents has on one another with a scalar number.  In fact the condition in Definition \ref{def:P_up}  relates the properties of a  matrix $A(x)\in\R^{Nn\times Nn}$ with $N$ blocks each of dimension $n$ with the properties of a smaller matrix $\Upsilon\in\R^{N\times N}$, where the effect of each block is summarized with a scalar number $\kappa_{i,j}$ (independent of $x$).  Specifically, let $A_{i,j}(x)\in\R^{n\times n}$ be the block in position $(i,j)$. Suppose that for all $i\in\N[1,N]$, $A_{i,i}(x)=A_{i,i}(x)^\top \succ 0$ and let $\kappa_{i,i}=\inf_{x\in\mathcal{X}} \lambda_{\textup{min}}(A_{i,i}(x))$. Moreover, set $\kappa_{i,j}=\sup_{x\in\mathcal{X}} \|A_{i,j}(x)\|_2$ for $j\neq i$. Then $\Upsilon\in\R^{N\times N}$ is constructed as follows.
\begin{equation}\label{upsi} A(x)=\left[\begin{array}{cccc}A_{1,1}(x) & A_{1,2}(x) & \hdots & A_{1,N}(x) \\A_{2,1}(x) & A_{2,2}(x) & \hdots & A_{2,N}(x) \\\vdots & \vdots & \ddots & \vdots \\A_{N,1}(x) & A_{N,2}(x) & \hdots & A_{N,N}(x)\end{array}\right]\qquad \rightarrow \qquad \Upsilon= \left[\begin{array}{cccc}\kappa_{1,1} & -\kappa_{1,2} & \hdots & -\kappa_{1,N} \\-\kappa_{2,1} & \kappa_{2,2} & \hdots & -\kappa_{2,N} \\\vdots & \vdots & \ddots & \vdots \\-\kappa_{N,1} & -\kappa_{N,2} & \hdots & \kappa_{N,N}\end{array}\right] \end{equation}

\begin{definition}[$\boldsymbol{P_\Upsilon}$ condition] \cite{scutari2014real} \label{def:P_up}
A  matrix valued function $x\mapsto A(x)\in\R^{Nn\times Nn}$ satisfies the $P_\Upsilon$ condition over $\mathcal{X}$ if $A_{i,i}(x)=A_{i,i}(x)^\top \succ 0$, $\kappa_{i,j}<\infty$ for all $i,j$ and  $\Upsilon$ in \eqref{upsi}
is a P-matrix.
\end{definition}

The following proposition relates the previous conditions to the properties in Definition~\ref{def:mon}.

\begin{proposition}[Sufficient conditions for strong monotonicity and  (block) P-functions]\label{prop:suff}
If the operator $F$ is continuously differentiable and the set $\mathcal{X}$ is nonempty, closed and convex, the following  holds.
\begin{itemize}
\item[a)] $F$ is strongly monotone if and only if there exists $\alpha>0$ such that $\frac{\nabla_x F(x)+\nabla_x F(x)^\top}{2} \succeq \alpha I$.
\item[b)] If $\nabla_x F(x)$ satisfies the $P_\Upsilon$ condition (as by Definition \ref{def:P_up}), then $F$  is a uniform block P-function.
\item[c)] If $\nabla_x F(x)$ satisfies the uniform P-matrix condition (as by  Definition \ref{def:P_matrix_uniform}  in Appendix \ref{appendix:uniform}), then $F$ is a uniform P-function.
\end{itemize}
\end{proposition}
Statements a) and b) can be found in \cite[Proposition 2.3.2(c)]{facchinei2007finite} and \cite[Proposition 5(e)]{scutari2014real}.  The proof of statement c) is given in Appendix~\ref{appendix:uniform}.
The conditions in Proposition \ref{prop:suff} are illustrated in the second line of Figure \ref{fig:rel}. 

\begin{figure}[H]
\begin{center}
\begin{tabular}{ccccc}
$F$ is strongly monotone&  $\Rightarrow$&$F$ is a uniform block P-function& $\Rightarrow$& $F$ is a uniform P-function\\[0.2cm]
$\Updownarrow$ && $\Uparrow$ && $\Uparrow$ \\[0.2cm]
 $\exists \alpha>0 $ s.t. && $\nabla_x F(x)$ satisfies the  && $\nabla_x F(x)$ satisfies the uniform  \\
  $\frac{\nabla_x F(x)+\nabla_x F(x)^\top}{2}\succeq \alpha I$ && $P_\Upsilon$  condition as by Def.  \ref{def:P_up} && P-matrix condition as by Def. \ref{def:P_matrix_uniform}
\end{tabular}
\end{center}
\caption{Relation between properties of $F$ and sufficient conditions in terms of $\nabla_x F(x)$, see Proposition \ref{prop:suff}.}
\label{fig:rel}
\end{figure}

It is important to stress that while both strong monotonicity and the $P_\Upsilon$ condition are sufficient to guarantee that $F$ is a uniform block P-function (and hence uniqueness of the Nash equilibrium), there is in general no relation between the two. In fact there are strongly monotone functions whose gradient  does not satisfy the $P_\Upsilon$ condition (see Example \ref{ex4}) and there are functions whose gradient satisfies the $P_\Upsilon$ condition but  are not strongly monotone (see Example  \ref{ex:PnotS}). Besides existence and uniqueness, we  show in Section \ref{sec:br_dynamics} that the $P_\Upsilon$ condition guarantees convergence of the best response dynamics, while we show in Section \ref{sec:comparative} that strong monotonicity is useful for sensitivity  analysis. Finally, we  employ the uniform P-function condition for the analysis of scalar  games of strategic substitutes. Consequently, it is important to understand  conditions under which each of the properties in Figure~\ref{fig:rel} is satisfied.

\section{ Properties of the game Jacobian in network games}
\label{sec:overview}

 Our goal in this section is to find sufficient conditions  in terms  of the influence of agent $i$ on its marginal cost (i.e., $\nabla^2_{x^ix^i}J^i(x^i, z^i)$),  the influence of the neighbor aggregate on the cost of agent $i$ (i.e.,  $\nabla^2_{x^i z^i}J^i(x^i, z^i)$) and  the network $\PP$ to guarantee that  the game Jacobian $F(x)$ posses the properties detailed in Definition~\ref{def:mon}. 
To this end, we exploit the sufficient conditions on $\nabla_x F(x)$ detailed in Proposition \ref{prop:suff}. 
In the case of network games, $\nabla_x F(x)$ can be rewritten as  
\begin{align}\label{eq:grad_F}
\nabla_x F(x)= D(x)+ K(x) \W ,
\end{align}
where $\W :=\PP \otimes I_n$ and

\begin{equation}\label{eq:DK}
\begin{aligned}
D(x)&:=\mbox{blkd}[D^i(x)]_{i=1}^N :=\mbox{blkd}[\nabla^2_{x^ix^i}J^i(x^i,z^i)\mid_{z^i=z^i(x)}]_{i=1}^N \\
K(x)&:=\mbox{blkd}[ K^i(x) ]_{i=1}^N :=\mbox{blkd}[ \nabla^2_{x^iz^i}J^i(x^i,z^i)\mid_{z^i=z^i(x)} ]_{i=1}^N. 
\end{aligned}
\end{equation}
Note that by the properties of the Hessian, $D^i(x)=D^i(x)^\top$, hence $D(x)=D(x)^\top$. Let us define

\begin{equation}\label{k1}
\begin{aligned}
\kappa_1^i&:=\min_{x}   \ \lambda_{\textup{min}}( D^i(x) ),\quad  &\ \kappa_1:=\min_{i} \ \kappa_1^i, \\ 
\kappa_2^i&:=\max_{x}\  \|K^i(x)\|_2,\quad  & \kappa_2:=\max_{i} \ \kappa_2^i.
\end{aligned}
\end{equation}
With this notation it is clear that for network games the quantities  in Definition \ref{def:P_up} can be rewritten as  $\kappa_{i,i}=\kappa_1^i$ and $\kappa_{i,j}=\kappa_2^i \PP_{ij}$.  The next example provides some intuition.

\begin{example}[continues=ex:lq]
Consider  Example \ref{ex:lq} with $a^i=0$ for simplicity.
According to Proposition~\ref{prop:kkt} a vector $x^\star$ is a Nash equilibrium if and only if it solves the VI$(\mathcal{X},F)$ where $\mathcal{X}=\R_{\ge0}\times \ldots\times \R_{\ge0}=\R^N_{\ge0}$ and 
\begin{equation*} \textstyle F(x)=[\nabla_{x^i} J^i(x^i,z^i(x))]_{i=1}^N=[x^i+K^i z^i(x)]_{i=1}^N=[x^i+K^i \sum_{j=1}^N \PP_{ij} x^j]_{i=1}^N=[ I_N + K \PP] x,\end{equation*}
with $K:=\mbox{blkd}[K^i]_{i=1}^N$.
From $\nabla_x F(x)= I_N + K \PP$ we obtain $D^i(x)=1$ and $K^i(x)=K^i$ for all $i\in\N[1,N]$. Consequently, the quantities in \eqref{k1} are $\kappa_1^i=\kappa_1=1$ for all $i\in\N[1,N]$ and $\kappa_2^i=|K^i|$, $\kappa_2=\max_i |K^i|$.
\end{example}

To get an intuition on the role of the quantities introduced in \eqref{k1} in our subsequent analysis, let us consider the strongly monotone case. By Proposition \ref{prop:suff}a, to prove strong monotonicity of $F$ one needs to show that $\frac{\nabla_x F(x)+\nabla_x F(x)^\top}{2}$  is positive definite for all $x$.
By using the gradient structure highlighted  in~\eqref{eq:grad_F} it is immediate to see that for network games

\begin{equation}\label{eq:main}
\frac{\nabla_x F(x)+\nabla_x F(x)^\top}{2}= D(x)+ \frac{K(x) \W+ \W^\top K(x)^\top}{2} \succeq \lambda_{\textup{min}}(D(x)) I+ \frac{K(x) \W+ \W^\top K(x)^\top}{2}.
\end{equation}
 Consequently, the positive  definiteness of $\frac{\nabla_x F(x)+\nabla_x F(x)^\top}{2}$ can be guaranteed by bounding the minimum eigenvalue of the sum of the two matrices on the right hand side of \eqref{eq:main}. It is clear that such a bound should depend  on  $\kappa_1=\min_x \lambda_{\textup{min}}(D(x))$, $\kappa_2= \max_x \|K(x)\|_2$ as defined in \eqref{k1} and on the properties of $\PP$, since $W=\PP \otimes I_n$. This approach motivates us to study conditions of the form
 
\begin{equation}\alpha_w:=\kappa_1-\kappa_2 w(\PP)>0\label{conditionw}\end{equation}
where $w(\PP)$ is a scalar that captures the network effect.
 Condition \eqref{conditionw} can  be understood as a bound on how large can the effect of the neighbor aggregate (measured by $\kappa_2$) be on an agent payoff  with respect to the effect of its own action (measured by $\kappa_1$) where $w(G)$ is  a weight that scales the magnitude of the two effects.
In the following, we consider three different network measures, as detailed next.

\begin{assumption}[Sufficient conditions in terms of $\kappa_1,\kappa_2$ and $\PP$]\label{ass:gen}
Suppose that at least one of the following conditions holds:
\begin{equation}
\alpha_2:= \kappa_1 - \kappa_2 \|\PP\|_2 >0 ,\tag{Assumption 2a} 
\end{equation}
\begin{equation}
\alpha_\infty:= \kappa_1 - \kappa_2 \|\PP\|_\infty >0, \tag{Assumption 2b} 
\end{equation}
\begin{equation}
\alpha_{\textup{min}}:= \kappa_1 - \kappa_2 |\lambda_{\textup{min}} (\PP)|>0. \tag{Assumption 2c} 
\end{equation}
\end{assumption}

Assumption \ref{ass:gen}a and \ref{ass:gen}c allows us to recover and extend conditions derived in the  literature for special instances of network games by using a  different type of  analysis. Assumption \ref{ass:gen}b is new to our knowledge and provides a natural summary of the network effect since  $\|\PP\|_\infty$ corresponds to the maximum row sum of $\PP$ and is thus a measure of  the maximum aggregate  influence that  the neighbors  have on  each  agent. 

 Our main technical result in this paper is to show that each of the three conditions in   Assumption \ref{ass:gen} guarantees a different set of properties of the game Jacobian, as summarized  in Table \ref{fig:summary1}.

\begin{table}[H]
\begin{center}
\resizebox{\textwidth}{!}{
\begin{tabular}{|c|c|c|c|c|}\hline
  \vphantom{\Large{{ \ding{51}}}}& \textbf{Assumption \ref{ass:gen}a} & \textbf{Assumption \ref{ass:gen}b } &\multicolumn{2}{c|}{\textbf{Assumption \ref{ass:gen}c}}  \\ 
\vphantom{\Large{{ \ding{51}}}}&& &{\small $K^i(x)=\tilde K(x)\succeq 0\ \forall i$}  &{\small $n=1,K^i(x)\ge \nu> 0\ \forall i$} \\ \hline 
\vphantom{\Large{{ \ding{51}}}}\textbf{Strong mon.} &\large{ \large{ \ding{51}}}&  \large{ \ding{55}}&\large{ \large{ \ding{51}}}   &\large{ \ding{55}} \\
&{\small (Thm. \ref{thm:norm})}&{\small(Ex. \ref{ex:PnotS})} &{\small(Thm. \ref{thm:min})} &{\small (Ex. \ref{ex3} ) } \\ \hline
\vphantom{\Large{{ \ding{51}}}}$\boldsymbol{P_\Upsilon}$ \textbf{condition }&\large{ \ding{51}} &\large{ \ding{51}}&\large{ \ding{55}} & \large{ \ding{55}}\\
& {\small(Thm. \ref{thm:norm})}& {\small (Thm. \ref{thm:inf})}& {\small(Ex. \ref{ex4})  }& {\small(Ex. \ref{ex3} )} \\ \hline
\vphantom{\Large{{ \ding{51}}}}\textbf{Unif. P-function } &\large{ \ding{51}}&\large{ \large{ \ding{51}}} &\large{ \large{ \ding{51}}} & \large{ \large{ \ding{51}}} \\
&{\small (Fig. \ref{fig:rel})}& {\small(Fig. \ref{fig:rel})}  &{\small (Fig. \ref{fig:rel}) }& {\small(Thm. \ref{thm:scalar}) }\\ \hline
\end{tabular}}
\end{center}
\caption{Summary of the relation between Assumptions \ref{ass:gen}a, \ref{ass:gen}b, \ref{ass:gen}c and properties of the game Jacobian. { The technical statements are provided in Appendix \ref{sec:game_jacobian}.} { In all cases we suppose that Assumption \ref{cost} holds}. For negative cases we provide a counter-example  in Appendix  \ref{sec:game_jacobian}. We list the $P_\Upsilon$ condition instead of the block P-function property because we show in Section~\ref{sec:br_dynamics} that the former is required for convergence of the discrete best response dynamics.
 Note that Assumptions \ref{ass:gen}c alone does not guarantee any property, in fact multiple Nash equilibria may arise (see Ex. \ref{ex2}). }
\label{fig:summary1}
\end{table}%

  Before providing some intuitions on the results of Table \ref{fig:summary1}, we outline the relation between the conditions in   Assumption \ref{ass:gen} and delineate which classes of networks satisfy each one of them.

\subsection{\textbf{Some comments on Assumption \ref{ass:gen}}}

Table \ref{fig:summary1} shows that Assumption \ref{ass:gen}a  guarantees both strong monotonicity  and the $P_\Upsilon$ condition. However, depending on the network, this assumption might be  restrictive. To see this, note that the larger the value of $\|\PP\|_2$ is, the more restrictive Assumption~\ref{ass:gen}a becomes. Figure~\ref{fig:net} shows that for some representative networks such as the complete network (Figure \ref{fig:net}a) or the asymmetric  star (Figure~\ref{fig:net}d), $\|\PP\|_2$ might actually grow unbounded in the number of players $N$.  This means that, for such networks, Assumption \ref{ass:gen}a  may not hold as the number of agents grows. We argue in the next subsections that Assumption \ref{ass:gen}b and \ref{ass:gen}c may be used as alternative conditions in these cases  by showing that for symmetric networks Assumption \ref{ass:gen}c holds for a broader range of networks, while for asymmetric networks  Assumption \ref{ass:gen}a and \ref{ass:gen}b allow addressing different set of networks.

\begin{figure}[h]
\begin{center}
\hspace{2.2cm}  \includegraphics[width=0.7\textwidth]{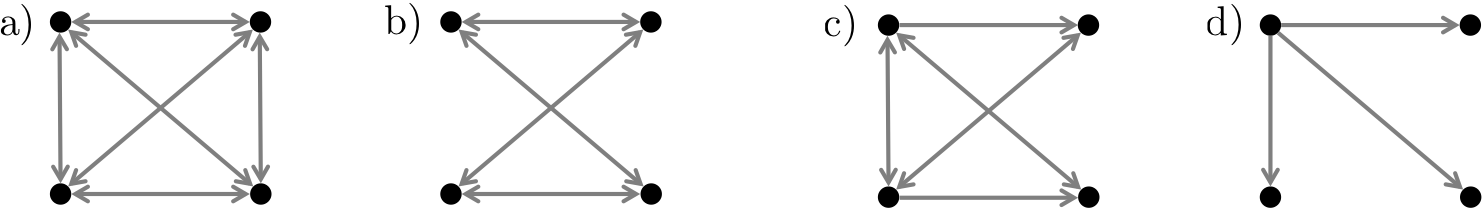}
\end{center}

\hspace{0.2cm}  \begin{tabular}{|p{0.7cm}p{1.5cm}|p{2.6cm}p{2.9cm}|p{2.7cm}p{2.7cm}|}\hline
 \multicolumn{2}{|c|}{} & \multicolumn{2}{c|}{\small Symmetric networks}&  \multicolumn{2}{c|}{\small Asymmetric networks}\\\hline
  \multicolumn{2}{|c|}{} & \hspace{1.1cm} a) &\hspace{1cm}  b) & \hspace{1cm} c) &\hspace{1cm}  d)\\\hline
\textbf{(\ref{ass:gen}a)}:&$\|\PP\|_2$ & \hspace{0.4cm} $N-1=$ 3 & \hspace{1cm} 2  &\hspace{0.6cm}  2.2882 &\hspace{0.cm}  $\sqrt{N-1}\approx$ 1.7\\  \hline
\textbf{(\ref{ass:gen}b)}:&$\|\PP\|_\infty$&\hspace{0.4cm}  $N-1= $ 3 &\hspace{1cm}  2 &\hspace{1cm}  2 &\hspace{1cm}  1\\ \hline
\textbf{(\ref{ass:gen}c)}:&$|\lambda_{\textup{min}}(\PP)|$ &\hspace{1cm}  1 &\hspace{1cm}  2 & \hspace{1cm} - &\hspace{1cm}  -\\ \hline
\end{tabular}
\caption{Spectral properties of some representative networks with $N=4$. a) Complete network: $\|\PP\|_2=\|\PP\|_\infty=N-1$, $|\lambda_{\textup{min}}(\PP)|=1$; b) Bipartite network; c) $2$-regular network; d) Asymmetric star: $\|\PP\|_2=\sqrt{N-1}$, $\|\PP\|_\infty=1$. We use the convention that $\PP_{ij}$ corresponds to an arrow from $j$ to $i$, since $\PP_{ij}>0$ means that agent $i$ is affected by the strategy of agent $j$. For example in the network d), all the agents are affected by the strategy of the agent in the top left corner and he is not affected by any other agent. In all the networks above, we use unitary edge weight. }
\label{fig:net}
\end{figure}

\subsubsection{\textbf{Relation between $\|\PP\|_2$ and $\|\PP\|_\infty$}}
\label{sec:overview1}

Since $\PP_{ii}=0$ it follows from Gershgorin theorem that $\rho(\PP)\le \|\PP\|_\infty$.  If $\PP$ is symmetric, $\|\PP\|_2=\rho(\PP)$
(see Lemma \ref{lem:s} in the Appendix). Consequently, we get $\|\PP\|_2=\rho(\PP)\le \|\PP\|_\infty$. This shows that  whenever $\PP$ is symmetric Assumption \ref{ass:gen}b is more restrictive than Assumption \ref{ass:gen}a. Hence Assumption \ref{ass:gen}b is a viable alternative to Assumption~\ref{ass:gen}a only for asymmetric networks. 
For asymmetric networks, there is no relation between $\|\PP\|_2$ and $\|\PP\|_\infty$, i.e.,  there are networks for which $\|\PP\|_2<\|\PP\|_\infty$ and networks for which  $ \|\PP\|_\infty<\|\PP\|_2$.
Figures \ref{fig:net}c) and  \ref{fig:net}d) show some examples of the latter case, thus justifying our interest for Assumption~\ref{ass:gen}b. 
Note for example that for the asymmetric star network, we have $\|\PP\|_2=\sqrt{N-1}$ while $\|\PP\|_\infty=1$ independently of $N$, thus Assumption~\ref{ass:gen}b can hold for large number of players while Assumption~\ref{ass:gen}a does not.
More generally, we note that, for the case of regular networks where each agent has the same number $d$ of in-neighbors, one gets immediately that $\|\PP\|_\infty=d$. Lemma \ref{lem:d_reg} in Appendix \ref{appendix:def} shows that  $\rho(\PP)= d$. Since for asymmetric matrices $\rho(\PP)\le \|\PP\|_2 $ we obtain $\|\PP\|_\infty=d=\rho(\PP)\le \|\PP\|_2 $. Hence for asymmetric regular networks Assumption~\ref{ass:gen}b is a relaxation of Assumption~\ref{ass:gen}a. Figure \ref{fig:net}c)  shows an example where the inequality is strict.

\begin{figure}[h]
\begin{center}
\hspace{1.5cm}  \includegraphics[width=0.6\textwidth]{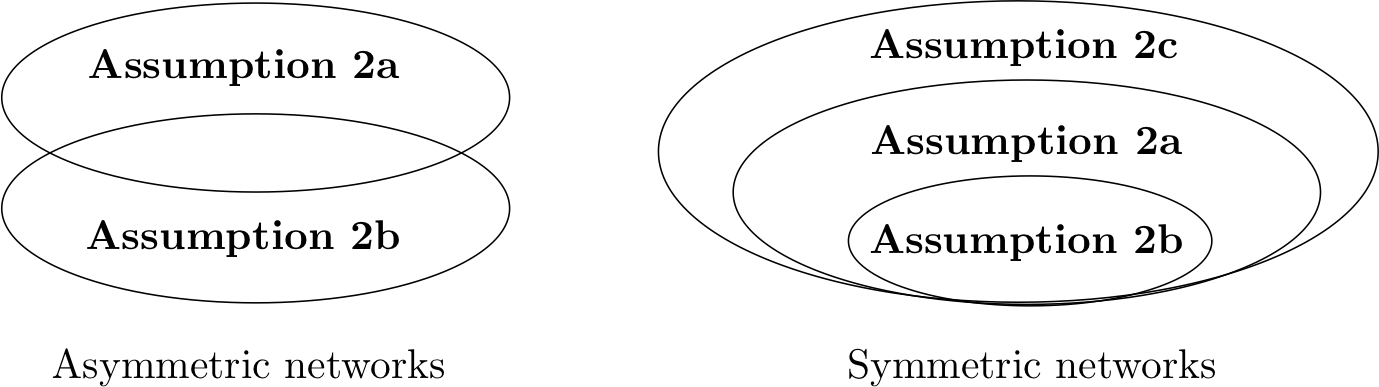}
\end{center}
\caption{Summary of the relations between Assumptions \ref{ass:gen}a, \ref{ass:gen}b and \ref{ass:gen}c, for asymmetric and symmetric networks. }
\label{fig:sets}
\end{figure}

\subsubsection{\textbf{Relation between $\|\PP\|_2$ and $|\lambda_{\textup{min}}(\PP)|$}}
\label{sec:overview2}
If the network is symmetric, then all the eigenvalues of $\PP$ are real and $\lambda_{\textup{min}}(\PP)$ is well defined. In this case, we have already mentioned that $\|\PP\|_2\le \|\PP\|_\infty$. The relation between $\|\PP\|_2$ and $\lambda_{\textup{min}}(\PP)$ is  detailed in the next lemma.

\begin{lemma}\label{lemma:P}
Suppose that $\PP =\PP ^\top$ is a non-zero and non-negative matrix entry-wise with zero diagonal. Then $\lambda_{\textup{min}}(\PP )<0$ and $\lambda_{\textup{max}}(\PP )=\rho(\PP )=\|\PP\|_2$. Consequently, $|\lambda_{\textup{min}}(\PP)|\le \|\PP\|_2$.
\end{lemma}
\begin{proof}
Since $\PP $ is non-negative, it follows from Perron-Frobenius theorem that $\lambda_{\textup{max}}(\PP )=\rho(\PP )>0$. Since $\mbox{Tr}(\PP )=\sum_i \lambda_i(\PP )=0$ and $\lambda_{\textup{max}}(\PP )>0$, it must be $\lambda_{\textup{min}}(\PP )<0$.
\end{proof}

In other words, Lemma \ref{lemma:P} proves that if $\PP$ is symmetric Assumption \ref{ass:gen}c is  less restrictive than Assumption~\ref{ass:gen}a. Figure \ref{fig:net} shows networks where this is strictly the case. Note for example that for the complete network in Figure \ref{fig:net}a) we have $\|\PP\|_2=N-1$ while $\lambda_{\textup{min}}(\PP)=-1$ independently of $N$, thus Assumption~\ref{ass:gen}c can hold for large number of players while Assumption~\ref{ass:gen}a does not.  We stress that there are graphs of interest for which  $|\lambda_{\textup{min}}(\PP )|=\|\PP\|_2$. Some examples are the undirected ring with an even number of nodes or any bipartite graph (see Figure \ref{fig:net}b). For these cases,  Assumption~\ref{ass:gen}c and  Assumption~\ref{ass:gen}a  coincide.
 For general $d$-regular graphs, it is known that $|\lambda_{\textup{min}}(\PP )|=\|\PP\|_2$ if and only if $\PP$ is bipartite. The gap between $|\lambda_{\textup{min}}(\PP )|$  and $\|\PP\|_2$ when $\PP$ is not  bipartite is studied e.g. in  \cite{desai1994characterization}.
 Figure \ref{fig:sets} shows a summary of the relations between the three assumptions for asymmetric and symmetric networks.

 \subsection{\textbf{Some intuitions on the results in Table \ref{fig:summary1}}}

 The technical statements relative to the results  in Table \ref{fig:summary1} and their proofs are provided in  Appendices~B and C. We here provide an overview and some intuition on these  results.

\subsubsection{ \textbf{Strong monotonicity}}\label{diss_smon}

According to Proposition \ref{prop:suff}a, $F(x)$ is strongly monotone if and only if $\nabla_x F(x)$ is uniformly positive definite. We have already mentioned at the beginning of this section - see Eq. \eqref{eq:main} - that such a condition can be guaranteed by lower bounding the smallest eigenvalue of the Jordan product 
$\frac{K(x)W+W^\top K(x)^\top}{2}.$ 

 In Theorem \ref{thm:norm}a) we show that for any matrix $K(x)$ and $W$ it holds
\begin{equation}\label{main2}
\textstyle \lambda_{\textup{min}}\left(\frac{K(x)W+W^\top K(x)^\top}{2}\right)\ge - \|K(x)\|_2\|G\|_2=- \kappa_2\|G\|_2.
\end{equation}
This immediately implies that strong monotonicity holds under Assumption~\ref{ass:gen}a. 

 We show in Example \ref{ex:PnotS} that instead Assumption~\ref{ass:gen}b does not guarantee strong monotonicity. A loose intuition can be provided by noting that when $n=1$ by Gershgorin  theorem, 
 $$\textstyle \lambda_{\textup{min}}\left(\frac{K(x)G+G^\top K(x)}{2}\right) \ge -\frac12 \| K(x)G+G^\top K(x) \|_\infty   \ge -\frac{\kappa_2}{2} (\| G \|_\infty+\|G \|_1).$$
 When $G$ is asymmetric, Assumption~\ref{ass:gen}b provides a bound on $\|G\|_\infty$ but  it does not impose any assumption on $\|G\|_1$.  When $G$ is symmetric, we  recall that Assumption~\ref{ass:gen}a is less strict than Assumption~\ref{ass:gen}b.

 Finally, Assumption~\ref{ass:gen}c can be used to guarantee strong monotonicity if additionally
  $G=G^\top$ and $K^i(x)=\tilde K(x)\succ 0$ for all agents $i$. In this case,  we  show in Theorem \ref{thm:min}  that the lower bound in \eqref{main2} can be further improved to
 $$\textstyle \lambda_{\textup{min}}\left(\frac{K(x)W+W^\top K(x)^\top}{2}\right)\ge - \|\tilde K(x)\|_2 |\lambda_{\textup{min}}(G)|=- \kappa_2|\lambda_{\textup{min}}(G)|.$$
 To illustrate this case, consider the scalar linear quadratic game in Example~\ref{ex:lq}, so that $K(x)=K\in\R^{N\times N}$. If the weights $K^i$ are homogeneous and positive, then $K(x)$ is a scalar positive matrix, that is, $K(x)=\kappa_2 I_N$. Using this fact and the symmetry of the network  it is immediate to see that 
 $$\textstyle \lambda_{\textup{min}}\left(\frac{KG+G^\top K}{2}\right) =\lambda_{\textup{min}}\left(\frac{\kappa_2G+\kappa_2G^\top }{2}\right)=\lambda_{\textup{min}}\left(\kappa_2G\right) = -\kappa_2 | \lambda_{\textup{min}}(G)|.$$
The two assumptions that the game features strategic substitutes (i.e. $K^i>0$) and the weights $K^i$ are homogeneous (i.e. $K^i=\kappa_2\ \forall i$) are critical in this argument. In the case of complements one would instead get  $K(x)=-\kappa_2 I_N$ and the negative sign would imply $\lambda_{\textup{min}}\left(-\kappa_2G\right) = -\kappa_2  \lambda_{\textup{max}}(G)= -\kappa_2  \|G\|_2$ leading back to the more restrictive Assumption~\ref{ass:gen}a.  Instead, when the weights are heterogeneous,  $K$ is a positive diagonal matrix with different diagonal elements and   it might happen that
 $$\textstyle \lambda_{\textup{min}}\left(\frac{KG+G^\top K}{2}\right) <  -\kappa_2 |\lambda_{\textup{min}}(G)|,$$
as illustrated in Example \ref{ex3}. If this happens, Assumption~\ref{ass:gen}c is not sufficient to guarantee strong monotonicity,  even in the case of scalar games of strategic substitutes. We show instead that in this case Assumption~\ref{ass:gen}c guarantees that the game Jacobian satisfies the weaker uniform P-function property.

\subsubsection{ \textbf{$\boldsymbol{P_\Upsilon}$ condition}}
The main advantage of the ${P_\Upsilon}$ condition is that it is expressed on the reduced matrix $\Upsilon$, which has size $N\times N$, instead of the gradient $\nabla_x F(x)$, which has size $Nn\times Nn$. This dimension reduction is achieved by summarizing the matrix of cross agents interactions with a single scalar number
$$\textstyle \frac{\partial^2 J^i(x)}{\partial x^i \partial x^j}= K^i(x) G_{ij} \quad \rightarrow \quad \kappa_{i,j}=\max_x \|K^i(x)\|_2 G_{ij},\quad  \forall i\neq j.$$
This operation  loses the sign properties of  $K^i(x)$. Consequently, any distinction between game of strategic complements or substitutes is lost in the $\Upsilon$ matrix. To clarify this, consider again the linear quadratic game in Example~\ref{ex:lq} with $K^i=-\delta$ for all $i$ (recall that this is a game of strategic complements if $\delta>0$ and substitutes if $\delta<0$). In both cases $\Upsilon= I -|\delta| G$.  This is the fundamental reason why Assumption ~\ref{ass:gen}c, which is typical of games of strategic substitutes, cannot be used to ensure the ${P_\Upsilon}$ condition. The fact that instead Assumption ~\ref{ass:gen}a and  ~\ref{ass:gen}b are both sufficient is proven in Theorems \ref{thm:norm} and \ref{thm:inf} by showing that, under these assumptions, the  matrix $M:=\kappa_1 I -\kappa_2 G$ is  a P-matrix. The conclusion then follows from the fact that $\Upsilon$ is a Z-matrix and is element-wise greater than $M$.

 \subsubsection{ \textbf{Uniform P-function}}
The previous two subsections highlight that under Assumption ~\ref{ass:gen}a and  ~\ref{ass:gen}b, the game Jacobian is either strongly monotone or it satisfies   the $P_\Upsilon$ condition (which are stronger properties than the uniform P-function property discussed here). Assumption~\ref{ass:gen}c is instead  sufficient to prove strong monotonicity  if   the network weights $K^i(x)$ are homogeneous and positive definite. 
With heterogeneous weights, we focus on the special class of scalar games of strategic substitutes (i.e. $n=1$, $K^i(x)\ge \nu>0$). In this case, even though the game Jacobian might not be strongly monotone (see Example \ref{ex3}), we show in Theorem \ref{thm:scalar} that it is a uniform P-function. 
 Understanding whether any of the properties  in Definition \ref{def:mon} holds under Assumption \ref{ass:gen}c  when $n>1$ and the weights are heterogeneous is an open problem.

\section{Existence and uniqueness}
\label{sec:ex}
 This section uses the properties established in Table \ref{fig:summary1}, together with Propositions \ref{prop:kkt} and \ref{prof:ex}, to present existence and uniqueness results for the Nash equilibrium of special classes of network games. We organize our discussion by first focusing on linear quadratic (multidimensional) network games, whereby we also explain how our theory recovers and extends  existing results, and then  we discuss  the less explored class of nonlinear network games.

\subsection{\textbf{Linear quadratic network games}}\label{sec:a_linear}
We start by considering linear quadratic games with general convex closed strategy sets $\mathcal{X}^i\subseteq \R^n$ and cost function

\begin{equation}\label{cost_qg}
J^i(x^i,z^i(x)):= \frac{1}{2} (x^i)^\top Q^i x^i + ( K^i z^i(x) +  a^i)^\top x^i,
\end{equation}
where $Q^i,K^i\in\R^{n\times n}$, $Q^i=(Q^i)^\top \succ 0$ and $a^i\in\R^{n}$. Note that  for such games $D^i(x)=Q^i, K^i(x)=K^i$. Consequently,
$
D(x)=\textup{blkd}[Q^i]_{i=1}^N$ and 
$K(x)=\textup{blkd}[K^i]_{i=1}^N
$ are both independent of $x$.   
The results of Table~\ref{fig:summary1} immediately yield the following two corollaries; the first one for general linear quadratic games (using Theorem \ref{thm:norm} and \ref{thm:inf}) and the second for linear quadratic games of strategic substitutes (using Theorems \ref{thm:min} and \ref{thm:scalar}).

\begin{corollary}[{Linear  quadratic games}]\label{LQ_gen}
Consider a linear  quadratic game with cost functions as in \eqref{cost_qg}. Suppose that at least one of the following  conditions holds:

\begin{equation}\label{eq:QG} \tag{LQ$_2$}
 \min_i \left[\lambda_{\textup{min}}(Q^i) \right]- \max_i\left[\| K^i\|_2\right] \|\PP\|_2>0, 
 \end{equation}
 \begin{equation}\label{eq:QG2} \tag{LQ$_\infty$}
 \min_i \left[\lambda_{\textup{min}}(Q^i) \right]- \max_i\left[\| K^i\|_2\right] \|\PP\|_\infty>0. 
 \end{equation}
Then there exists a unique  Nash equilibrium. 
\end{corollary}
\begin{proof}
Condition \eqref{eq:QG} guarantees that Assumption \ref{ass:gen}a is met. By Theorem \ref{thm:norm} the game Jacobian is strongly monotone and the statement follows from Propositions \ref{prop:kkt} and \ref{prof:ex}a).
Condition \eqref{eq:QG2} guarantees that Assumption \ref{ass:gen}b is met. By Theorem \ref{thm:inf} the game Jacobian satisfies the $P_\Upsilon$ condition and the statement follows from Propositions \ref{prop:kkt} and \ref{prof:ex}b).
\end{proof}

 We note that Condition  \eqref{eq:QG}  guarantees that the sufficient condition derived  in \cite{ui2016bayesian} holds.\footnote{ For general games, Proposition 1 in \cite{ui2016bayesian} states that if $F(x)$ is strictly monotone then there exists at most one Nash equilibrium. For linear quadratic network games  Proposition 4 therein states that if $Q+KG\succ 0$ then there exists a unique Nash equilibrium. Using this condition \cite{ui2016bayesian} shows that one can recover  the results in \cite{ballester2006s} and \cite{bramoulle2014strategic}, among others. We note that Condition \eqref{eq:QG} in this work is a sufficient condition for $Q+KG\succ 0$. On the other hand, Condition \eqref{eq:QG2} does not imply $Q+KG\succ 0$ and thus allows the study of linear quadratic network games  that do not satisfy the condition in \cite{ui2016bayesian}.}
The second corollary focuses on  $\PP =\PP ^\top$  and provides  alternative sufficient conditions in terms of $|\lambda_{\textup{min}}(\PP )|$ for games of strategic substitutes (recall that  for symmetric networks it holds $|\lambda_{\textup{min}}(\PP )| \le \|\PP\|_2\le \|\PP\|_\infty$). 

\begin{corollary}[{Linear  quadratic games of strategic substitutes}]\label{LQ_sub}
Consider a linear  quadratic game with cost functions as in \eqref{cost_qg}. Suppose that $\PP =\PP ^\top$ and that one of the following conditions holds:
\begin{enumerate}
\item Scalar games of strategic substitutes: $n=1$, $ K^i>0$  $\forall i$.
\item  Homogeneous weights: $ K^i=\tilde K$ $\forall i$, $\tilde K+\tilde K^\top \succeq 0$.
\end{enumerate}
Then the condition
\begin{equation}\label{eq:Q-sub} \tag{LQ$_\textup{min}$}
 \min_i \left[\lambda_{\textup{min}}(Q^i) \right]- \max_i\left[\| K^i\|_2\right] |\lambda_{\textup{min}}(\PP )|>0 
 \end{equation}
guarantees existence and uniqueness of the Nash equilibrium. 
\end{corollary}
\begin{proof}
Condition \eqref{eq:Q-sub} guarantees that Assumption \ref{ass:gen}c holds.
Under condition 1), $\mathcal{X}$ is rectangular since $n=1$ and Theorem \ref{thm:scalar} guarantees that $F$ is a uniform P-function (note that $\nu:=\min_i K^i>0$). Under condition 2),  Theorem \ref{thm:min} guarantees that $F$ is strongly monotone.  In both cases existence and uniqueness then follow by Propositions \ref{prop:kkt} and \ref{prof:ex}. 
\end{proof}

 To place our results in the context of the existing   literature, we revise the  network game with multiple activities  described in Example \ref{ex:multiple}, which  subsumes the models studied in \cite{ballester2006s,bramoulle2007public,bramoulle2014strategic,chen2017multiple} and introduces new features such as budget constraints that have not been considered before.   

\begin{example}[continues=ex:multiple] 
Consider the network game described in Example \ref{ex:multiple} and note that  the payoff of each agent corresponds to a cost function as in \eqref{cost_qg} with 
$$Q^i:=\left[\begin{array}{cc}1 & \beta^i \\\beta^i & 1\end{array}\right], \quad K^i:=-\left[\begin{array}{cc}\delta & \mu \\\mu & \delta\end{array}\right], \quad a^i:=-\left[\begin{array}{c}a^i_A \\a^i_B\end{array}\right].$$
Consequently, $\lambda_{\textup{min}}(Q^i)=1-|\beta^i|>0$ and $\|K^i\|_2=|\delta|+|\mu|$.  Corollary  \ref{LQ_gen} provides the following sufficient conditions for existence and uniqueness of the Nash equilibrium 
\begin{align}
\min_i(1-|\beta^i|) -  (|\delta|+|\mu|) \|G\|_2>0 \label{cond2_multiple} \\
\min_i(1-|\beta^i|) -  (|\delta|+|\mu|) \|G\|_\infty>0 \label{condinf_multiple}
\end{align}
 Condition \eqref{cond2_multiple} reduces to  Assumptions 1 and 4 in \cite{chen2017multiple}  for the sub-cases considered therein where $\mathcal{X}^i=\R^2$, $G$ is symmetric, $\delta\ge 0$, $\mu=0$. For the case $\delta\ge 0$ but $\mu\neq 0$, Assumption 3 in \cite{chen2017multiple} is equivalent to  \eqref{cond2_multiple} if the signs of $\mu$ and $\beta^i$ follow the relation detailed  in Footnote \ref{fot:sign}.  We also note that both Conditions \eqref{cond2_multiple} and \eqref{condinf_multiple} do  not require $\delta$ to be positive, apply to possibly asymmetric networks and can be used for any closed and convex set $\mathcal{X}^i$ (e.g., representing budget constraints).

For the case when $\delta$ is negative and $G$ is symmetric a more expansive condition can be obtained by  using Corollary  \ref{LQ_sub}-2. In fact, in this model 
$$ K^i:=-\left[\begin{array}{cc}\delta & \mu \\\mu & \delta\end{array}\right]=\left[\begin{array}{cc}-\delta & -\mu \\ -\mu & -\delta\end{array}\right]=:\tilde K.$$
If $\delta<0$  the elements on the diagonal of $\tilde K$ are positive and if $|\mu|<|\delta|$ then $\textup{det}(\tilde K)=\delta^2-\mu^2>0$ implying that $\tilde K\succ 0$. From Corollary  \ref{LQ_sub}-2 a sufficient condition for uniqueness  is thus
\begin{align}
\label{multiple_min}
\PP=\PP^\top,\quad  \delta<0, \quad   |\mu|<|\delta|, \quad  \min_i(1-|\beta^i|) -  (|\delta|+|\mu|) |\lambda_\textup{min}(G)|>0, 
\end{align}
which can be seen as a generalization of the results in \cite{bramoulle2014strategic}  to network games with multiple activities where network effects within the same activity are substitutes. 

To better illustrate the relation between our conditions and the ones derived in previous works we consider the case $\beta^i=\beta$ for all agents $i$, $\mu=0$, $G=G^\top$ and no budget constraint so that $\mathcal{X}^i=\R^2_{\ge 0}$. If $\beta=0$ then there is no coupling between the activities and hence the equilibria of this model are the same as the equilibria of the single activity model described in Example \ref{ex:lq}. 
Table \ref{table:lq} shows that conditions \eqref{cond2_multiple} and \eqref{multiple_min} recover to the well-known conditions in  \cite{ballester2006s}, \cite{bramoulle2014strategic} and \cite{chen2017multiple} for the cases studied therein. A novel condition is obtained for  games that feature  strategic substitutability within the same activity.  

  \begin{table}[h]
\begin{center}
{\small
\begin{tabular}{c|c|c}
& Complements & Substitutes \\ 
& (${\delta>0}$) & (${\delta<0}$) \\ \hline && \\
Single activity  & $\|G\|_2 < \frac{1}{\delta}$ & $ |\lambda_{\textup{min}}(G)|< \frac{1}{|\delta|}$ \\ 
($\beta=0$) & \cite{ballester2006s} & \cite{bramoulle2014strategic} \\ && \\ \hline && \\ 
Multiple activities & $\|G\|_2 < \frac{1-|\beta|}{\delta}$  & $ |\lambda_{\textup{min}}(G)|< \frac{1-|\beta|}{|\delta|}$  \\ 
($\beta\neq 0$)& \cite{chen2017multiple} & \\  
\end{tabular}}
\end{center}
\caption{ Sufficient conditions for uniqueness of the Nash equilibrium for different linear quadratic games, seen as special cases of Example \ref{ex:multiple} for $\mu=0$, $\beta^i=\beta$, $G=G^\top$ and $\mathcal{X}^i=\R^2_{\ge 0}$.}
\label{table:lq}
\end{table}

Condition \eqref{condinf_multiple} is  new. To illustrate its utility, recall from the analysis in Section \ref{sec:overview1} that this condition may handle cases not covered by  Condition \eqref{cond2_multiple} when the network is asymmetric. An  example that illustrates  this additional flexibility is a game with single activity ($\beta=0$) and strategic complements ($\delta>0$) where there is an agent which is a \textit{trend setter} in the sense that its decision influences the rest of the agents much more than any other agent, as  in the  network illustrated in Figure \ref{fig:trend}.
\begin{figure}[h]
\begin{center}
\begin{minipage}{0.4\textwidth}
\begin{center}
\includegraphics[width=0.3\textwidth]{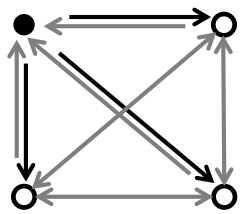}
\end{center}
\end{minipage}
$\Rightarrow$
\begin{minipage}{0.45\textwidth}
$$G=\left[\begin{array}{cccc}0 & 0.1 & 0.1 & 0.1 \\\boldsymbol{1} & 0 & 0.1 & 0.1 \\\boldsymbol{1} & 0.1 & 0 & 0.1 \\\boldsymbol{1} & 0.1 & 0.1 & 0\end{array}\right]$$
\end{minipage}
\end{center}
\caption{ Black arrows have higher weight than gray arrows, implying that  the agent in the top left corner is a ``trend setter''.}
\label{fig:trend}
\end{figure}
For this network $\|G\|_2=1.7437 $ while $\|G\|_\infty=1.2$. Hence by using the condition of \cite{ballester2006s} one can guarantee uniqueness for $\delta <0.5735 $ while by using the new Condition \eqref{condinf_multiple}, uniqueness is guaranteed for the larger interval $\delta <0.8333$.
\end{example}

\subsection{\textbf{Nonlinear  network games}}
\label{nonlinear}
We next consider network games with  nonlinear cost functions of the form
\begin{equation}\label{cost_ng}
J^i(x^i,z^i(x)):= q^i(x^i)+ f^i (z^i(x))^\top x^i,
\end{equation}
where $f^i,q^i$ are  differentiable functions and  $\{q^i\}_{i=1}^N$ are  uniformly strongly convex with constant $\kappa_1>0$. Note that for games with this structure $\kappa_2=\max_i\max_{z^i} \| \nabla_{z^i}f^i(z^i)\|$.   We use the results in Table~\ref{fig:summary1} to obtain the following corollaries for this class of games.

\begin{corollary}[{Nonlinear network games}]\label{NL_gen}
Consider a  nonlinear network game with cost functions as in \eqref{cost_ng}. Suppose that  at least one of the following conditions holds 

\begin{equation}\label{eq:NG}
 \kappa_1 - \kappa_2 \|\PP\|_2>0 \tag{NL$_2$},
 \end{equation}
 
 \begin{equation}\label{eq:NG2}
 \kappa_1 - \kappa_2 \|\PP\|_\infty>0 \tag{NL$_\infty$}.
 \end{equation}
Then there exists a unique Nash equilibrium. 
\end{corollary}
\begin{proof}
Under \eqref{eq:NG} Assumption \ref{ass:gen}a is met. The result follows from Theorem \ref{thm:norm} and Propositions \ref{prop:kkt} and \ref{prof:ex}a).
Under \eqref{eq:NG2} Assumption \ref{ass:gen}b is met. The result follows from Theorem \ref{thm:inf} and Propositions \ref{prop:kkt} and \ref{prof:ex}b).
\end{proof}
 Corollary \ref{NL_gen} under \eqref{eq:NG2} generalizes the result in \cite{acemoglu2015networks} (obtained for the case $n=1$, $q^i(x^i)=\frac{1}{2}(x^i)^2$, $f^i=f$ for all $i$,  where $f$ is a contraction mapping, and the network $G$ is such that  $\sum_{j\neq i} G_{ij}=1$) to games with generic heterogeneous interaction functions $f^i$ and multidimensional strategies.
If we consider games of strategic substitutes we instead obtain the following.

\begin{corollary}[{Nonlinear network  games of strategic substitutes}]\label{NL_sub}
Consider a  network  game with cost function as in \eqref{cost_ng}. 
Suppose that $\PP =\PP ^\top$, $n=1$ and  that at least one of the following holds
\begin{enumerate}
\item there exists $\nu>0$ such that $(f^i)'(z^i) \ge \nu$ for all~$i$ and all $z^i$;
\item  $(f^i)'(z^i)>0$ for all~$i$ and all $z^i$ and $\mathcal{X}^i$ is lower bounded  for all~$i$.
\end{enumerate}
Then the condition
\begin{equation}\label{eq:NG_sub}
 \kappa_1 - \kappa_2| \lambda_{\textup{min}}(\PP )|>0  \tag{NL$_\textup{min}$}
 \end{equation}
guarantees existence and uniqueness of the Nash equilibrium. 
\end{corollary}
\begin{proof}
Under condition 1. the conclusion  follows from Theorem \ref{thm:scalar} and Propositions \ref{prop:kkt} and \ref{prof:ex}c).   Under condition 2. Theorem \ref{thm:scalar} guarantees that $F(x)$ is a P-function (see Definition \ref{def:mon2}). Since this is a scalar game of strategic substitutes the best response function of each agent is non-increasing in $z^i$. Let $a^i$ be the lower bound on $\mathcal{X}^i$ and $z^i_a:=\sum_{j\neq i} G_{ij} a^j$. Then $z^i(x)\ge z^i_a$ for all $x\in\mathcal{X}$ and $B^i(z^i(x))\le B^i(z^i_a)=:b^i$. Equivalently, the best response  takes values in the compact interval $[a^i,b^i]$. The conclusion then follows from the fact that a Nash equilibrium exists by Brower fixed point theorem and  there is at most one equilibrium  since $F(x)$ is a P-function, see  \cite[Proposition 3.5.10 (a)]{facchinei2007finite}.
\end{proof}

Note that Corollary \ref{NL_sub} applies only to scalar games. Nonetheless, it   slightly extends the result in \cite{allouch2015private} by allowing for more general interaction functions.\footnote{ The setup of \cite{allouch2015private}  coincides with the choice $q^i(x^i)=(x^i)^2$ and $f^i(z^i)=-[\gamma^i(w^i+z^i)-z^i]$, where $\gamma^i$ is consumers $i$'s Engel curve and $w^i>0$ is his income. The network normality condition therein (i.e., $1+\frac{1}{\lambda_{\textup{min}}(\PP)}<(\gamma^i)'(\cdot)<1$) can  be rewritten as  $-\frac{1}{\lambda_{\textup{min}}(\PP)}>(f^i)'(\cdot)>0$, which implies that this is a game of strategic substitutes (i.e., $(f^i)'(\cdot)>0$) and that Assumption \ref{ass:gen}c is satisfied (i.e., $1-\max_{z^i} |(f^i)'(z^i)||\lambda_{\textup{min}}(\PP)|>0$).} 
To illustrate  the results in Corollary \ref{NL_gen} and \ref{NL_sub} we consider  Example \ref{competition} which nests the models discussed in \cite{belhaj2014network} and \cite{allouch2015private}.

\begin{example}[continues=competition] 
 The best response discussed in Example \ref{competition}  corresponds to a network game where each agent has cost function
$J^i(x^i,z^i)=\frac12(x^i)^2-(a^i+\phi_i(z^i))x^i$ 
and strategy set $\mathcal{X}^i=[a^i,b^i]$. Consequently,  $\kappa_1=1$, $\kappa_2=\max_i |\phi_i'(z^i)|$. By  Corollary \ref{NL_gen} and  Corollary \ref{NL_sub} existence and uniqueness can thus be guaranteed if
$
1 - \max_i |\phi_i'(z^i)| \min\{\|\PP\|_2,\|\PP\|_\infty\}>0$
 or, if additionally $G=G^\top$ and $\phi_i'(z^i)<0$, if
$
1 - \max_i |\phi_i'(z^i)|| \lambda_{\textup{min}}(\PP)|>0.$
Let us consider again the  special case with $a^i=a$, $b^i=b$, 
 $\phi_i(z^i)=\phi(z^i)=\gamma z^i(b-z^i),$
for all $i$ and network $\PP$ such that $\sum_j \PP_{ij}=1$. 
Recall that the parameter $\gamma>0$ models the level of network effects. For this particular choice of $\phi_i(z^i)$ and $\PP$, it turns out that $\kappa_2=\gamma \max_{z\in[a,b]}|b-2z|=\gamma b$ and $\|\PP\|_2\ge\|\PP\|_\infty=1$. Hence \eqref{eq:NG2} guarantees uniqueness of the equilibrium  for 
$\gamma < \frac{1}{b}.$

 To further understand the equilibrium structure we start by considering symmetric equilibria. Suppose $x^i=\bar x_\gamma$ for each $i$. Then it holds $z^i=\sum_{j} \PP_{ij} \bar x_\gamma=\bar x_\gamma$ for each $i$. Hence $\bar x(\gamma)=\bar x_\gamma \mathbbm{1}_N$  is an interior symmetric equilibrium  if and only if 
  $\phi(\bar x_\gamma)=\bar x_\gamma -  a.$
The  plot on the left of Figure \ref{fig:example2_eq} shows that such equation   admits a solution $\bar x_\gamma \in(a,b)$ for any $\gamma>0$. Hence for any level of network effects there exists a symmetric Nash equilibrium, where all the agents exert the same effort. From the uniqueness result above we know that for small  levels of network effects (i.e. $\gamma < \frac{1}{b}$), this is the unique Nash equilibrium. 

{  \cite{belhaj2014network}  shows that if $\phi_i'>0$ (i.e. the games features strategic complements) and satisfies some curvature conditions (e.g. $\phi_i$ is concave) then the Nash equilibrium is unique, no matter the level of network effects.} We next show that when the best response is instead non-monotone, for large values of $\gamma$, multiple equilibria may arise. To exemplify this point we consider a case with $N=2$ players, $G=G_c$ and we set $a=1,b=5$. The equilibria of this game are shown in Figure \ref{fig:example2_eq} (right). For $\gamma<0.4866$ there is a unique equilibrium which is symmetric ($x_S^1=x_S^2, $ black line).  This is consistent with our theory that guarantees a unique equilibrium for $\gamma< \frac{1}{b}=0.2$. The gap between the thresholds $\gamma=0.4866$ and $\gamma=0.2$ is due to the fact  that our condition  is sufficient but not necessary. For $\gamma>0.4866$ two additional asymmetric equilibria arise where agent 1 is highly engaged ($x^1_A$ blue dashed line) and agent 2 is scarcely engaged ($x^2_A$ blue dashed-dotted line) or viceversa. Note that for  $\gamma\ge\frac{1}{a}=1$ these two equilibria converge to fully and minimally engaged, that is $( x^1_A, x^2_A)=(a,b)$ or viceversa. Overall, we can conclude that when network effects are small the result of competition is to have two agents that are both equally engaged, if network effects are high there is instead the possibility of specialization with one agent clearly dominating over the other. { This result is qualitatively similar to Proposition 3 in \cite{bramoulle2014strategic}, where it was shown that, for the game of strategic substitutes considered therein, for large network effects there exists at least one Nash equilibrium with inactive agents.}
\begin{figure}[H]
\begin{center}
 \includegraphics[height=0.2\textwidth]{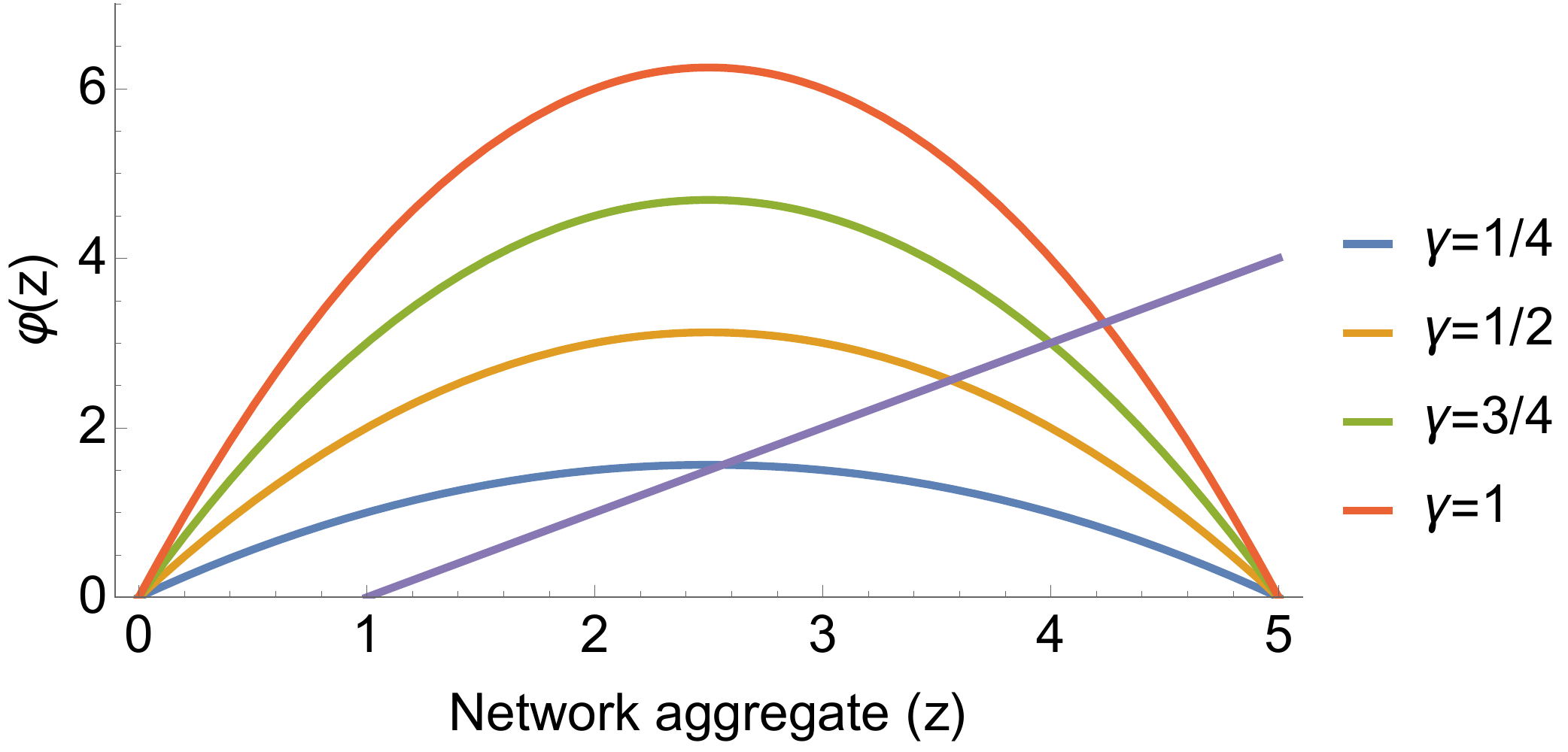}  \qquad
\includegraphics[height=0.2\textwidth]{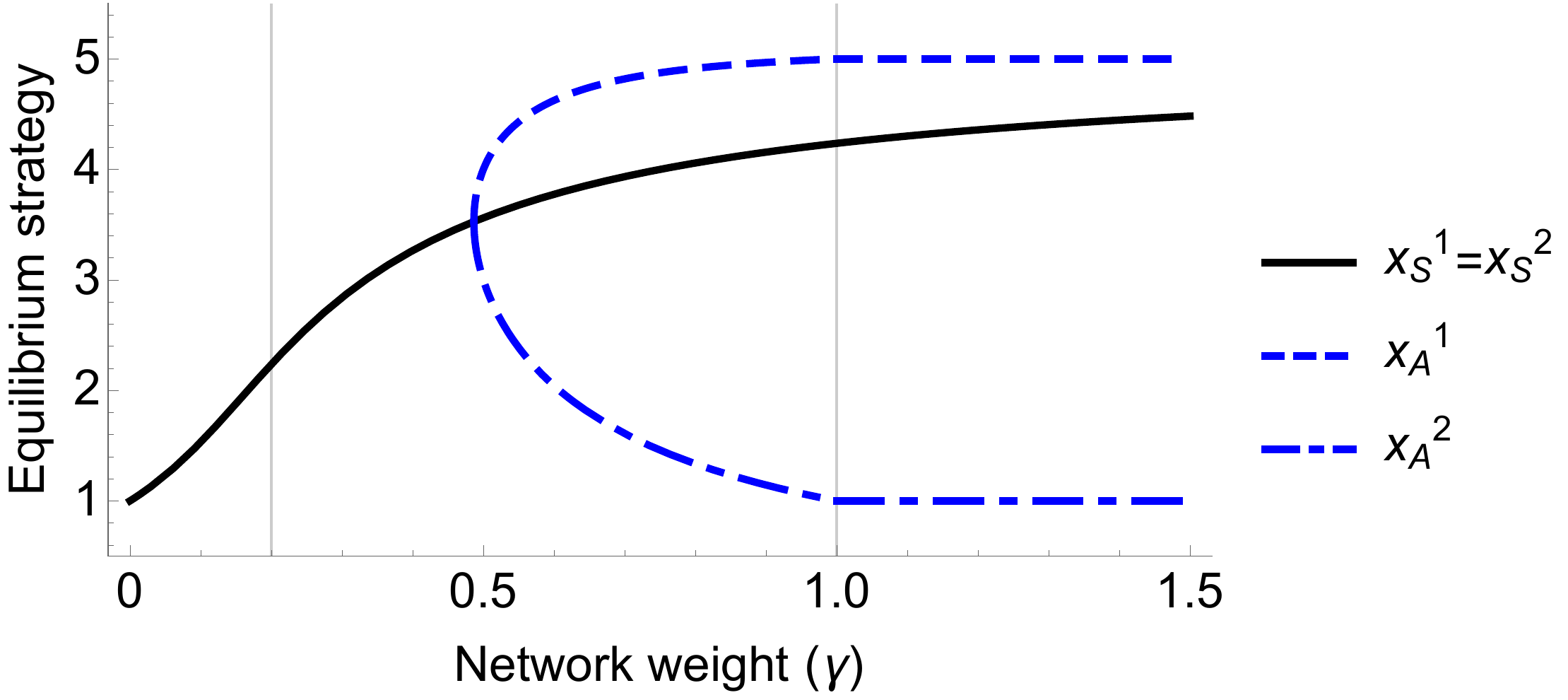}
\end{center}
\caption{ Left: $\phi(z)$ for different values of $\gamma$ and $z-a$. Right: Equilibria  of Example \ref{competition} with $N=2$, $G=G_c$, $a=1$, $b=5$. }
\label{fig:example2_eq}
\end{figure}
 \end{example}

\section{Best response dynamics}

\label{sec:br_dynamics}

 After characterizing the set of Nash equilibria, we here discuss how  agents can  reach such a configuration by iteratively playing their best responses.  In accordance with the previous literature we consider both continuous and discrete best response dynamics (BR dynamics). For continuous time we consider the scheme introduced in \cite{bramoulle2014strategic} and reported in Algorithm \ref{br_c}.\footnote{ We note that under Assumption 1, it follows by  the fact that $J^i(x^i,z^i)$ is strongly monotone in $x^i$, $\nabla_{x^i} J^i(x^i,z^i)$ is Lipschitz continuous in $z^i$ and \cite[Theorem 2.1]{dafermos1988sensitivity} that $B^i(z^i)=\arg\min_{x^i\in\mathcal{X}^i} J^i(x^i,z^i)$ is Lipschitz continuous in $z^i$. Hence, $B^i(z^i(x))-x^i$ is Lipschitz continuous in $x$ and Algorithm \ref{br_c} admits a unique solution $x(t)$ for all initial conditions and all $t\ge0$ by \cite[Theorem 2.3]{khalil1996noninear}.}

\begin{algorithm}
\caption{Continuous best response dynamics}\label{br_c}
\textbf{Set:}  $x^i(0)=x^i_0\in \mathcal{X}^i$, $z^i(0)=\sum_{j=1}^N \PP_{ij} x^j(0)$ \\
\textbf{Dynamics:}
\begin{align*}
\dot x^i(t)&= B^i(z^i(t))-x^i(t)=\left[\arg\min_{x^i\in\mathcal{X}^i} J^i(x^i,z^i(t))\right] -x^i(t) \qquad \forall i\in\N[1,N].
\end{align*}
\end{algorithm}

 For the case of discrete dynamics, we index by $k\in\N$ the time instants at which at least one agent is  updating its strategy. For each $i\in\N[1,N]$, we denote by $\mathcal{T}_i\subseteq \N$  the subset of  time instants at which agent $i$ updates its strategy.
In Algorithm \ref{br} we  consider  two variants of the discrete BR dynamics depending on whether the agents update their strategies simultaneously or sequentially, which correspond to $\mathcal{T}_i=\N$  and $\mathcal{T}_i=N (\N-1) +i$, respectively.

 \begin{algorithm}
\caption{Discrete best response dynamics}\label{br}
\textbf{Set:} $k=0$, $x^i_0\in \mathcal{X}^i$, $z^i_0=\sum_{j=1}^N \PP_{ij} x^j_0$. Set either  $\mathcal{T}_i=\N$ for all $i$ (simultaneous BR dynamics) or $\mathcal{T}_i=N (\N-1) +i$  for all $i$ (sequential BR dynamics).\\
\textbf{Iterate:}
\begin{align*}
x^i_{k+1}&= \begin{cases} B^i(z^i_k)= \arg\min_{x^i\in\mathcal{X}^i} J^i(x^i,z^i_k), & \mbox{if} \ k\in\mathcal{T}_i \\
x^i_{k}, & \mbox{otherwise} 
\end{cases}  \qquad \forall i\in\N[1,N]. 
\end{align*}
\end{algorithm}

 This type of dynamics have been used in the literature to refine the set of Nash equilibria by focusing on those that are asymptotically stable  for  the continuous/discrete BR dynamics.\footnote{ We refer to an equilibrium $x^\star$ as  asymptotically stable for some given dynamics  if there exists an open neighborhood $\mathcal{O}$ of  $x^\star$ such that  the  dynamics   asymptotically converge to $x^\star$ starting from  any  point $x_0\in\mathcal{O}$ and as globally asymptotically stable if this  holds for $\mathcal{O}=\mathcal{X}$. } Intuitively, this means that following a small change of the  strategies the agents can converge back to equilibrium by simply following their continuous/discrete BR dynamics. 

 It is well known that if the game has a strongly convex potential  then the continuous BR dynamics (as well as the sequential discrete BR dynamics) globally converge to the unique Nash equilibrium,   see Lemma \ref{lemma:potential}. The intuition behind this result is simple: since each strategy update  leads to a decrease of the potential function, the BR dynamics must converge to the unique minimum of such potential which is the unique Nash equilibrium. From \cite[Lemma 4.4]{monderer:shapley:96} it follows that a game is potential if and only if  there exists $U(x)$ such that  $\nabla_xU(x)=F(x)$. It is then easy to see that such potential $U(x)$ is strongly convex if and only if $F(x)$ is strongly monotone. In other words, strongly convex potential games are a subclass of strongly monotone games. This motivates the question of whether the preceding convergence results can be generalized to strongly monotone games that are not potential. 

The answer to this question depends on the type of dynamics considered (continuous vs discrete). For continuous BR dynamics note that, in a potential game, $U(x)$ acts as a Lyapunov function \cite[Section 7.1.1]{sandholm2010population}. In fact let $B(x)$ be the vector of best responses to the strategy vector $x$ (i.e. $[B(x)]_i=B^i(z^i(x))$) and define $d(x):= B(x)-x$, then
$$\dot{U}(x(t))=\nabla_x U(x(t))^\top \dot x(t)= F(x(t))^\top d(x(t))$$  and it follows by the definition of best response (see Lemma \ref{lemma:descent}) that  for any $x\in\mathcal{X}$
\begin{equation}\label{descent}
F(x)^\top d(x) \le - \kappa_1 \|d(x)\|_2^2.
\end{equation}
In other words, $d(x)$ is a descent direction for the potential function $U(x)$.

If the game is not potential then  $F(x)$ cannot be seen as a gradient of any function.  Hence while \eqref{descent} still holds, the interpretation of $d(x)$ as a descent direction  is lost. In the following theorem we show that, however, if $F(x)$ is strongly monotone and the cost function has the following form
\begin{equation}\label{br:cost}
J^i(x^i,z^i):=\frac12 \|x^i\|^2_{Q^i}+f^i(z^i)^\top x^i
\end{equation}
for $Q^i=(Q^i)^\top\succ 0$ and $f^i:\R^n\rightarrow \R^n$, 
then a similar argument can be made by using the Lyapunov  function
\begin{equation}\label{br:pot} 
\tilde U(x):=F(x)^\top (x-B(x))-\frac12 \|x-B(x)\|^2_Q \quad
\end{equation}
where $Q:=\textup{blkd}[Q^i]_{i=1}^N$.
The reason why we consider cost functions as in \eqref{br:cost} is that, in this case,    the best response mapping solves a quadratic program in $x^i$  and can thus be rewritten as a projection step, that is, 
\begin{equation}\label{eq:br:closed}
B(x)=\Pi_{\mathcal{X}}^Q[x-Q^{-1}F(x)]
\end{equation}
as shown in Lemma \ref{br_closed}. We note that the cost structure in \eqref{br:cost} is general enough to include all our motivating Examples \ref{ex:lq} to \ref{ex:multiple}.  Formula \eqref{eq:br:closed} allows us to show that $d(x)$ is a descent direction for $ \tilde U(x) $.
\begin{lemma}\label{lemma:descent_smon}
Consider a network game satisfying Assumption~\ref{cost} and with agent cost functions as in \eqref{br:cost}. Suppose that $F(x)$ is strongly monotone with constant $\alpha$ then 
$$\nabla_x \tilde U(x)^\top d(x)\le -\alpha \|d(x)\|^2_2.$$
\end{lemma} 

An immediate consequence of  Lemma \ref{lemma:descent_smon} is the convergence of the continuous BR dynamics under strong monotonicity.
 \begin{theorem}\label{thm:br_smon}
Consider a network game satisfying Assumption~\ref{cost}  and with agent cost functions as in \eqref{br:cost}. Suppose that $F$ is strongly monotone. Then for any $x_0\in\mathcal{X}$ the sequence $\{x(t)\}_{t\ge0}$ generated by Algorithm~\ref{br_c} converges to the unique Nash equilibrium of $\mathcal{G}$. 
\end{theorem} 
It is important to remark that $\tilde U(x)$ as defined in \eqref{br:pot} is a Lyapunov function since $\tilde U(x)$  decreases along the continuous trajectory $x(t)$, but it is not a potential for the game as defined in \cite{monderer:shapley:96}. For this reason, strong monotonicity is not sufficient to guarantee convergence of the discrete BR dynamics. To see this, note that, when the cost function is as in \eqref{br:cost},  the discrete BR dynamics coincide, by \eqref{eq:br:closed}, with the projection algorithm
$$x_{k+1}=\Pi_{\mathcal{X}}^Q[x_k-\tau Q^{-1}F(x_k)]$$ 
for the step choice $\tau=1$. If  $F$ is strongly monotone  the projection algorithm converges when $\tau$ is small  but not necessarily for $\tau=1$, see e.g.  \cite{facchinei2007finite}.
A counter-example, illustrating that  strong monotonicity alone (without the potential structure) is not sufficient to guarantee convergence of the discrete BR dynamics  (either simultaneous or sequential) is given in Example \ref{ex:br}.

Interestingly, Example \ref{ex:br} shows that strong monotonicity is not enough to guarantee convergence of the discrete \textit{simultaneous} BR dynamics  even when the game is potential. The reason is that while in sequential BR dynamics one can guarantee that each step decreases the potential, in the case of simultaneous moves this is no longer true. 
In \cite[Theorem 10]{scutari2014real} it is shown that convergence of the discrete simultaneous BR dynamics can instead be guaranteed  by using the $P_\Upsilon$ property. The reason why this result holds is fundamentally different from the argument used in potential games. In fact it is not based on the existence of a Lyapunov function, but convergence is instead guaranteed by showing that the best response mapping is a   block-contraction, as recalled in the following lemma.
 \begin{lemma} \cite[Proposition 42]{scutari2014real} \label{lemma:br_block}
 Consider a  game for which $\nabla_x F(x)$ satisfies the $P_\Upsilon$ condition, then the best response mapping $B(x)$ is a block-contraction, that is, there exist  $c\in\R^N_{>0}$  and  $0<\delta_c<1$ such that 
 $$\max_i \frac{\|B^i(x)-B^i(y)\|_2}{c_i} \le \delta_c \max_i \frac{\|x^i-y^i\|_2}{c_i} \qquad \mbox{for all} \ x,y \in\mathcal{X}.$$
 \end{lemma}
{ The $P_\Upsilon$ condition  can be used to prove  convergence of any of the algorithms above, as detailed next. }
\begin{theorem}\label{cor:br}
Consider a network game satisfying Assumption~\ref{cost}  and suppose that $\nabla_x F(x)$ satisfies the $P_\Upsilon$ condition. Then for any $x_0\in\mathcal{X}$  the sequence $\{x(t)\}_{t\ge0}$ generated by Algorithm~\ref{br_c} and the sequence $\{x_k\}_{k=0}^\infty$ generated by Algorithm~\ref{br} converge to the unique Nash equilibrium of $\mathcal{G}$.
\end{theorem}
We note that convergence of the discrete BR dynamics follows immediately from \cite[Theorem 10]{scutari2014real} where  convergence is proven for general games (i.e. not necessarily network games) and for both simultaneous and sequential updates (in fact, also random  updates with delays can be considered, see Algorithm 1 therein). The result on convergence of the continuous BR dynamics under the $P_\Upsilon$ condition  is instead new.

It is important to stress that the block contraction property used in Theorem \ref{cor:br} to prove convergence of the discrete BR dyanmics does not necessarily hold under strong monotonicity (see Example \ref{ex:br2}). A number of alternative discrete time algorithms have been suggested in the literature to guarantee convergence to the Nash equilibrium when the game Jacobian is strongly monotone. These schemes typically depend on some tuning parameters and are  meant for applications where there is flexibility in tuning the  algorithm or the update of the agents. For this reason, we do not discuss these algorithms here and we instead refer the interested reader to \cite{chen2014autonomous,koshal2012gossip,koshal2016distributed} where the use of distributed projection algorithms is discussed, to \cite{scutari2014real} where a regularized version of the best response dynamics is presented and to \cite{parise2015network,parise2015networkA} where convergence of the best response dynamics for strongly monotone linear quadratic network games is achieved by assuming that the agents respond to a filtered version of the neighbor aggregate. Note that Theorems \ref{thm:norm} and \ref{thm:min} in this work give sufficient conditions in terms of the network to  guarantee strong monotonicity and thus convergence of such schemes.

 A summary of the results discussed so far in this section is provided in the following table.

\begin{table}[H]
\begin{center}
\resizebox{0.95\textwidth}{!}{
\begin{tabular}{|c|c|c|c|}\hline
  \vphantom{\Large{{ \ding{51}}}}& \textbf{Potential game} & \textbf{Strongly monotone } & \hspace{1cm}  $\boldsymbol{P_\Upsilon}$ \textbf{condition} \hspace{1cm}  \\ \hline
\vphantom{\Large{{ \ding{51}}}}\textbf{Continuous BR dynamics} &\large{ \large{ \ding{51}}}&  \large{ \ding{51}}&\large{ \large{ \ding{51}}}    \\
{Algorithm \ref{br_c} } &{ \small (\cite{sandholm2010population})}&{\small( if cost as in \eqref{br:cost} - Thm. \ref{thm:br_smon})} &{\small(Thm. \ref{cor:br})}  \\ \hline
\vphantom{\Large{{ \ding{51}}}} \textbf{Discrete sequential }&\large{ \ding{51}} &\large{ \ding{55}}&\large{ \ding{51}} \\
{Alg. \ref{br} with $\mathcal{T}_i=N (\N-1) +i$ }  & { \small (\cite{monderer:shapley:96})}& {\small(Ex. \ref{ex:br}B)}& { \small(\cite{scutari2014real} - see Thm. \ref{cor:br})  } \\ \hline
\vphantom{\Large{{ \ding{55}}}} \textbf{Discrete simultaneous }&\large{ \ding{55}} &\large{ \ding{55}}&\large{ \ding{51}}\\
{Alg. \ref{br} with $\mathcal{T}_i=\N$}  & {\small(Ex. \ref{ex:br}A)}& {\small(Ex. \ref{ex:br}B)}& {\small(\cite{scutari2014real} - see Thm. \ref{cor:br})  } \\ \hline
\end{tabular}}
\end{center}
\caption{ Summary of the relation between properties of $F$ and convergence of the BR dynamics. }
\label{fig:summary4}
\end{table}%

By combining  Table \ref{fig:summary4} with the results in Table \ref{fig:summary1} one immediately gets that if either Assumption \ref{ass:gen}a) or Assumption \ref{ass:gen}b) hold then the ${P_\Upsilon}$ {condition} is met (by Theorems \ref{thm:norm} and \ref{thm:inf}) and both Algorithms \ref{br_c} and \ref{br} converge by Theorem \ref{cor:br}. 
For games of strategic substitutes with symmetric networks, one can consider also Assumption \ref{ass:gen}c). Theorems \ref{thm:br_smon} and  \ref{thm:min} guarantee that, if the cost function  is as in \eqref{br:cost} and $K^i(x)=\tilde K(x)\succeq 0\ \forall i$, then the continuous BR dynamics (Algorithm \ref{br_c}) converge but Example \ref{ex:br}A) shows that the discrete BR dynamics (Algorithm \ref{br}) may not.  The only case of Table \ref{fig:summary1} which is not covered by the previous theory is the case when  $n=1$ and $K^i(x) > 0$. In this case $F$ is a P-function hence none of the previous conditions apply. We next show that nonetheless it is possible to prove convergence of the continuous BR dynamics. Example \ref{ex:br}A) instead shows that the discrete BR dynamics might  not converge.

 \begin{theorem}\label{lemma:br_sub}
Consider a scalar network game satisfying Assumption~\ref{cost}  and with cost function as in \eqref{br:cost} with $Q^i=1$ for all $i\in\N[1,N]$. Assume that $G=G^\top$ and  that $\nabla_x F(x^\star)$ is a P-matrix. 
Moreover, suppose that  the equilibrium  is non degenerate, that is, if $x^{\star,i}\in\partial \mathcal{X}^i$ then $F^i( x^\star)\neq 0$.\footnote{ Note that the assumption $Q^i=1$ is without loss of generality because in the scalar case $Q^i$ is a positive number, hence one can always consider an equivalent game with cost functions $\tilde J^i(x)=\frac{J^i(x)}{Q^i}$. The assumption that $\nabla_x F(x^\star)$ is a P-matrix is weaker than the assumption that $F$ satisfies the uniform P-matrix condition and is guaranteed under the assumptions of Theorem \ref{thm:scalar}. Finally,  the assumption that the equilibrium  is non degenerate corresponds to the assumption that the strict complementarity condition holds, see Assumption \ref{ass:diff}c) and Appendix \ref{appendix:comp} for a related discussion.}
 Then $x^\star$ is  locally asymptotically stable for the continuous BR dynamics given in Algorithm \ref{br_c}.
\end{theorem}

 Theorem \ref{lemma:br_sub} generalizes \cite[Theorem 2]{allouch2015private}. A summary of the results on BR dynamics for network games under Assumption \ref{ass:gen}a), \ref{ass:gen}b) and \ref{ass:gen}c) is given in Table \ref{fig:summary2} in the introduction.
 To illustrate our  findings we apply the previous theory to the multiple activities model considered in Example \ref{ex:multiple}.

\begin{example}[continues=ex:multiple] We divide our analysis in two cases. 

If either condition \eqref{cond2_multiple} or \eqref{condinf_multiple} hold, then  $\nabla_x F(x)$ satisfies the $P_\Upsilon$ condition and Theorem \ref{cor:br}  allows us to conclude that the unique Nash equilibrium is globally asymptotically stable for both Algorithm \ref{br_c} and \ref{br} (i.e. both continuous and discrete BR dynamics globally converge). Recall that this model includes as special cases the games studied in  \cite{chen2017multiple} (for $\delta>0$) and \cite{ballester2006s} (for $\beta=0$). Our result proves convergence of the continuous and discrete BR dynamics in both these models.

If $\delta<0$, then condition \eqref{multiple_min} guarantees that $F(x)$ is strongly monotone.  Theorem \ref{thm:br_smon} allows us to conclude global convergence of the  continuous BR dynamics. { This result is in line with the findings of \cite[Corollary 2]{bramoulle2014strategic} obtained for  games with scalar strategies.}
\end{example} 

\section{Comparative statics}
\label{sec:comparative}

We finally study how the Nash equilibrium changes when the cost functions of the agents or the network changes. To this end, we introduce the parametric cost functions

\begin{equation}\label{eq:cost_param}
J^i(x^i,z^i(x),y^i)
\end{equation} 
where $y^i\in\mathcal{Y}^i\subseteq \R^{d_i}$ is a local parameter and we set $\mathcal{Y}:=\mathcal{Y}^1\times \ldots \times \mathcal{Y}^N\subseteq \R^d$ to be the set of all possible parameters. We denote by $\mathcal{G}(y,\PP)$ the game with parameter $y\in \mathcal{Y}$ and network $\PP$. The corresponding parametric game Jacobian is 

\begin{equation}\label{eq:F_param}
F(x,y,\PP):=[\nabla_{x^i} J^i(x^i,z^i(x),y^i)]_{i=1}^N.
\end{equation} 
We generalize Assumption \ref{cost}  as follows.

\begin{assumption}[Parametric games]\label{cost_y}
The set $\mathcal{X}^i\subseteq \R^n$ is nonempty closed and convex for all $i\in\N[1,N]$.
The function $J^i(x^i,\q^i(x),y^i)$ is  continuously differentiable and  strongly convex in $x^i$  for all $i\in\N[1,N]$, for all $x^{j}\in\mathcal{X}^{j}$, $j\in\mathcal{N}^i$ and for all $y^i\in\mathcal{Y}^i$. Moreover, $J^i(x^i,z^i,y^i)$ is  continuously differentiable in $[x^i;z^i;y^i]$ and 
$ \nabla_{x^i} J^i(x^i,z^i,y^i)$ is differentiable and Lipschitz in $[x^i;z^i;y^i]$. Let $L$ be the maximum of such Lipschitz constants over $i\in\N[1,N]$. 
\end{assumption}

\subsection{\textbf{Lipschitz continuity}}

\begin{theorem}[Lipschitz continuity]\label{cor:Lipschitz}
Suppose that Assumption~\ref{cost_y} holds and  that for a given  $\bar y\in\mathcal{Y}$ and network $\bar \PP$, the operator $F(x,\bar y,\bar \PP)$ is a uniform block P-function  with constant $\bar \eta$,  so that the Nash equilibrium $x^\star(\bar y,\bar \PP)$ of the game $\mathcal{G}(\bar y,\bar \PP)$ is unique.\footnote{See the sufficient conditions for the  uniform block P-function  property derived in Theorems \ref{thm:norm} to \ref{thm:scalar}.} Then for any $ y\in\mathcal{Y}$ and any Nash equilibrium $x^\star(y,\bar \PP)$ of the perturbed game $\mathcal{G}(y,\bar \PP)$ it holds

\begin{align*}
\|x^\star(y,\bar \PP)-x^\star(\bar y, \bar \PP)\|_2 \le \frac{L}{\bar \eta}   \|y-\bar  y\|_2.
\end{align*}
If additionally $\mathcal{X}$ is bounded then for any $ y\in\mathcal{Y}$, any network $\PP$ and  any Nash equilibrium $x^\star(y, \PP)$ of $\mathcal{G}(y, \PP)$ it holds

\begin{align*}
\|x^\star(y,\PP)-x^\star(\bar y,\bar \PP)\|_2^2 \le \frac{L^2}{\bar \eta^2}  ( \|\PP-\bar \PP\|_2^2 \Delta^2 + \|y-\bar  y\|_2^2),
\end{align*}
where $\Delta:=\max_{x\in\mathcal{X}}\|x\|_2$.
\end{theorem}\

A similar Lipschitz continuity result for the solution of general VIs was presented in  \cite[Theorem 2.1]{dafermos1988sensitivity} under the assumption of strong monotonicity.   \cite{melo2017variational}  studies Lipschitz continuity of the Nash equilibrium in  strongly monotone network games. With respect to both of these works, Theorem~\ref{cor:Lipschitz}  proves Lipschitz continuity under the weaker uniform block P-function property  allowing the analysis of games whose Jacobian is not strongly monotone, e.g. under Assumption \ref{ass:gen}b) or \ref{ass:gen}c). 

The assumptions needed in Theorem \ref{cor:Lipschitz} to obtain Lipschitz continuity of the Nash equilibrium are fairly general. In the next section we study its differentiability by considering more structured constraint sets and adopting a constraint qualification assumption.

\subsection{\textbf{Differentiability}}

In this section we consider a fixed network and hence omit $\PP$ in $\mathcal{G}(y,\PP)$, $F(x,y,\PP)$ and $x^\star(y,\PP)$.  We also assume that the sets $\mathcal{X}^i$ are polyhedra and satisfy the {\it Slater constraint qualification}, as presented next.  This allows us to use a KKT reformulation of the VI$(\mathcal{X},F)$, as recalled in Proposition \ref{prop:kkt2}.  
\begin{assumption}[Polyhedral constraints]\label{constraints}
The  constraint sets can be expressed as $\mathcal{X}^i:=\{x^i\in\R^n\mid B^i x^i\le b^i, H^i x^i=h^i\}$, for some $B^i\in\R^{m_i \times n},  b^i\in\R^{m_i}, H^i\in\R^{p_i \times n},  h^i\in\R^{p_i}$. 
There exists $x^i\in\mathcal{X}^i$ such that $B^ix^i<b^i$. 
\end{assumption}

For simplicity, we  define $B:=\mbox{blkd}(B^1,\ldots, B^N) \in\R^{m\times Nn}, b=[b^i]_{i=1}^N \in\R^{m}, H:=\mbox{blkd}(H^1,\ldots, H^N) \in\R^{p \times Nn}, h=[h^i]_{i=1}^N \in\R^{p}$, so that $\mathcal{X}=\{x\in\R^{Nn}\mid Bx\le b, Hx=h\}$.

\begin{proposition}[KKT system for VIs]\label{prop:kkt2}
Suppose Assumptions~\ref{cost_y} and \ref{constraints} hold. The following statements are equivalent:
\begin{enumerate}
\item $x^\star(y)$ solves the VI$(\mathcal{X},F(\cdot,y))$;
\item  there exists $\lambda(y)\in \R^{m} $ and $\mu(y) \in \R^{p} $
such that
\begin{subequations}\label{KKT}
\begin{align}
&F(x^\star(y),y) + B^\top\lambda(y) + H^\top\mu(y)=0 \label{KKT1} \\
& H x^\star(y)=h, \quad \lambda(y)^\top (B x^\star(y)-b)=0, \label{KKT2}\\
& B x^\star(y)\le b,  \quad \lambda(y) \ge 0.
\end{align}\end{subequations}
\end{enumerate}
\end{proposition}
The previous proposition is proven in \cite[Proposition 1.3.4 and 3.2.1]{facchinei2007finite}.

\begin{assumption}[Constraint qualification]\label{ass:diff} Given a parameter $\bar y\in\mathcal{Y}$ and the corresponding Nash equilibrium $x^\star(\bar y)$, let $B^0, b^0$ be the matrices obtained by deleting from $B, b$ all the rows corresponding to constraints that are not active at $x^\star(\bar y)$ (i.e. the rows $k\in\N[1,m]$ such that $B_{(k,:)}x^\star(\bar y)< b_k$). Set $A:=[B^0 ; H]$ and $a:=[b^0,h]$. 
Suppose that the following properties hold:
\begin{enumerate}
\item[\ref{ass:diff}a)] { $\nabla_x F(x^\star(\bar y),\bar y)+\nabla_x F(x^\star(\bar y),\bar y)^\top \succ 0$,}
\item[\ref{ass:diff}b)] $A$ has full row rank,
\item[\ref{ass:diff}c)] the strict complementarity slackness condition
$\lambda_k(\bar y)>0 \mbox{ when } B_{(k,:)}x^\star(\bar y)= b_k $ 
is satisfied.\footnote{Under Assumptions \ref{cost_y}, \ref{constraints},    \ref{ass:diff}a and \ref{ass:diff}b, $\lambda(\bar y)$ is unique, see e.g.  \cite{friesz1990sensitivity}.   Assumption \ref{ass:diff} is used to apply an implicit function theorem type of argument to the KKT system, see \cite{friesz2016foundations}. For example,   \ref{ass:diff}a) is a second-order   condition implying that the equilibrium is  locally unique. }  
\end{enumerate}
\end{assumption}

Under Assumptions \ref{cost_y}, \ref{constraints} and \ref{ass:diff}, the  Nash equilibrium is locally unique and differentiable, as shown in \cite{facchinei201012} and \cite{friesz2016foundations}.
It is important to note that while  Lipschitz continuity of the Nash equilibrium can be easily obtained in network games (see Theorem \ref{cor:Lipschitz}), to prove differentiability one needs to rely on Assumption \ref{ass:diff}a), which is implied by strong monotonicity, and on Assumptions  \ref{ass:diff}b) and  \ref{ass:diff}c) which, on the other hand, depend on the constraint structure  and might be  difficult to verify a priori. In some special cases however they can be easily checked.  For example, if there are only equality constraints both conditions immediately hold.  In fact, in the case of equality constraints, if $A$ is not full row rank then there are redundant constraints that can be removed without loss of generality, hence \ref{ass:diff}b) always hold. Condition \ref{ass:diff}c) on the other hand involves only  inequality constraints and is trivially met if $B$ is the empty matrix.   

For the cases when Assumption \ref{ass:diff} holds,   an explicit formula for the Jacobian of $x^\star(y)$ at $\bar y$ is given in \cite[Theorem 1]{friesz1990sensitivity}. Such formula  involves both  the primal and dual variables associated with the solution to VI$(\mathcal{X},F(\cdot,\bar y))$ and is therefore difficult to interpret   in terms of simple primitives of the network game.\footnote{ The formula for the derivative of the  vector $w(y)$  provided therein is
$\nabla_y w(\bar y)=-[\nabla_w \Gamma(w,y)]^{-1}[\nabla_y \Gamma(w,y)]\mid _{\{w=w(\bar y), y=\bar y\}},$
where $\Gamma(w,y)$ is the KKT system obtained by considering only the equalities in \eqref{KKT1} and \eqref{KKT2}, see  \cite[Eq. (7.126)]{friesz2016foundations} for more details. The formula for $\nabla_y x^\star(\bar y)$ obtained from the first block component of $\nabla_y w(\bar y)$ depends on the dual variables $\lambda(\bar y),\mu(\bar y)$. 
}  
By considering only the active constraints instead of the whole KKT system, we  present  here an equivalent formula for $\nabla_y x^\star(\bar y)$ that does not depend explicitly on the dual variables. This reformulation allows us to have an immediate understanding  of the behavior of the Nash equilibrium for small perturbation of the parameter.

\begin{proposition}[Nash equilibrium sensitivity formula]\label{thm:diff}
If Assumptions \ref{cost_y}, \ref{constraints} and  \ref{ass:diff} hold then the Jacobian of $x^\star(y)$ at $\bar y$ can be expressed as \begin{equation}\label{sensitivity}
\begin{aligned}
\nabla_y x^\star(\bar y) &= - M [\nabla_y F(x,y) ]_{\{x=x^\star(\bar y),y=\bar y\}}\\
\end{aligned}
\end{equation}
where 
\begin{equation}\label{L}
\begin{aligned}
M&:=[L-LA^\top [ALA^\top]^{-1}AL],\\ L&:=[\nabla_x F(x,y) ]_{\{x=x^\star(\bar y),y=\bar y\}}^{-1}
\end{aligned}
\end{equation}
and $A$ is as defined in Assumption \ref{ass:diff}. 
\end{proposition}

 Proposition \ref{thm:diff} is proven in our companion paper  \cite{parise2017} where we additionally show how it  
 can be used to study optimal interventions in networks, along the lines of \cite{ballester2006s, acemoglu2015networks}.  We show in  Appendix \ref{appendix:comp} how Assumption \ref{ass:diff}  and Proposition \ref{thm:diff} simplify in  the case of scalar strategies and non-negativity constraints.

To illustrate our  findings we apply the previous theory to the race and tournaments model introduced in Example \ref{competition}. 

\begin{example}[continues=competition] 
We already showed that this game has a symmetric equilibrium $\bar x(\gamma)=\bar x_\gamma \mathbbm{1}_N$ for any value of $\gamma$ and that this is  the unique equilibrium  when $\gamma < \frac{1}{b}$ (since Assumption \ref{ass:gen}b is met). Theorem~\ref{cor:br} immediately allows us to conclude that for $\gamma < \frac{1}{b}$,   the BR dynamics (both continuous and discrete) are globally converging, hence the unique equilibrium is globally asymptotically stable with respect to both continuous and discrete BR dynamics.

In the following we aim at studying how the symmetric equilibrium changes if $\gamma$ changes. To this end, note that 
 the equilibrium is interior since $\bar x_\gamma\in(a,b)$. Hence the matrix $A$, as defined in Assumption \ref{ass:diff}, is  empty. 
Assumption \ref{ass:diff}b) and \ref{ass:diff}c) are therefore trivially met. We  need to verify Assumption \ref{ass:diff}a).  Recall that $F^i(x)=x^i-(a+\gamma z^i(b-z^i))$ hence
$$\frac{\partial F^i(x)}{\partial x^i}=1 \quad \mbox{and} \quad \frac{\partial F^i(x)}{\partial x^j}=-\gamma G_{ij}(b-2z^i).$$
Overall,
$$ \nabla_x F(x,\gamma)_{\{x=\bar x(\gamma)\}}=I-\epsilon_\gamma G  \quad \mbox{and} \quad \nabla_\gamma F(x,\gamma)_{\{x=\bar x(\gamma)\}}= - \bar x_\gamma(b-\bar x_\gamma)\mathbbm{1}_N,$$
where $\epsilon_\gamma:= \gamma(b-2\bar x_\gamma)$.
 Consequently, Assumption \ref{ass:diff}a) is met if 
 $ I-\epsilon_\gamma \frac{G+G^\top}{2}\succ 0.$
 For $\gamma < \frac{1}{b}$ it holds
 $-1< -\gamma b =\gamma (b-2b)< \epsilon_\gamma < \gamma b <1$ or equivalently $|\epsilon_\gamma|<1$.
Hence a sufficient condition for Assumption~\ref{ass:diff}a) to be met is that $G$ is not only row stochastic but also column stochastic, in which  case $\rho\left(\frac{G+G^\top}{2}\right)=1$.
Formula \eqref{sensitivity} then  leads to 
\begin{align*}
\textstyle \nabla_\gamma \bar x(\gamma)&\textstyle =\bar x_\gamma(b-\bar x_\gamma) [I-\epsilon_\gamma G]^{-1}\mathbbm{1}_N=\bar x_\gamma(b-\bar x_\gamma) \left[\sum_{k=0}^\infty \epsilon_\gamma^k G^k\right ]\mathbbm{1}_N=\bar x_\gamma(b-\bar x_\gamma) \left[\sum_{k=0}^\infty \epsilon_\gamma^k G^k\mathbbm{1}_N\right ]\\&\textstyle =\bar x_\gamma(b-\bar x_\gamma) \left[\sum_{k=0}^\infty \epsilon_\gamma^k \mathbbm{1}_N \right ] = \frac{\bar x_\gamma(b-\bar x_\gamma) }{1-\epsilon_\gamma} \mathbbm{1}_N >_e0, \end{align*}
hence the equilibrium effort is increasing in the network weight $\gamma$ (note that  all the components have the same derivative consistently with the fact that this is a symmetric equilibrium).

When $\gamma >\frac1b$ we have already shown that multiple equilibria might appear. It is interesting to note that for large values of $\gamma$ 
 Assumption \ref{ass:diff}a) is not satisfied (hence formula \eqref{sensitivity} cannot be applied) and $\nabla_x F(x,\gamma)_{\{x=\bar x(\gamma)\}}$ has a negative eigenvalue. Recalling that $\bar x(\gamma)$ is an interior equilibrium and using formula \eqref{eq:br:closed}, the continuous BR dynamics can be linearized around the symmetric equilibrium leading to the error dynamics $\dot e(t)=-\nabla_x F(\bar x(\gamma),\gamma) e(t)$, $e(t):=x(t)-\bar x(\gamma)$ (see also the proof of Theorem \ref{lemma:br_sub}).  For large $\gamma$ the  matrix $\nabla_x F(\bar x(\gamma),\gamma) $ has a negative eigenvalue and, in this range, $\bar x(\gamma)$ is not asymptotically stable (by Lyapunov's method).\footnote{This result is qualitatively consistent with Proposition 5 in \cite{bramoulle2014strategic}, where it is shown that  for the game of strategic substitutes considered therein and for large network effects all stable  equilibria involve at least one inactive agent.} 
  This result suggests a bifurcation type of behavior where the symmetric equilibrium is asymptotically stable when unique, but is not asymptotically stable  for values of $\gamma$ corresponding to multiple equilibria. An interesting consequence  is that the total effort \textup{in the stable equilibrium} might not be monotone in $\gamma$. For instance, the following plot  shows that the sum of the aggregate effort in the game with $N=2$ agents is increasing in $\gamma$ in the symmetric equilibrium, but is  decreasing  in the asymmetric equilibrium. It can be shown by the linearization argument above that for this example the symmetric equilibrium is  asymptotically stable until the bifurcation point and after that it is not, while the asymmetric equilibrium is asymptotically stable whenever it exists. Therefore, the total effort in the  asymptotically stable equilibrium is first increasing and then decreasing. This is in sharp contrast with the results in \cite{belhaj2014network} and \cite{bramoulle2014strategic} that show that  the total effort at the maximum equilibrium\footnote{We note that the symmetric equilibrium  of the game in Figure \ref{fig:example2_tot} is not a maximum equilibrium because, as shown in Figure~\ref{fig:example2_eq}, in the asymmetric equilibrium one agent plays an higher strategy. In fact, the three equilibria of this game cannot be ordered.} is   increasing for games of strategic complements and  decreasing  for  games of strategic substitutes, respectively.
\end{example}

\begin{figure}[H]
\begin{center}
 \includegraphics[height=0.2\textwidth]{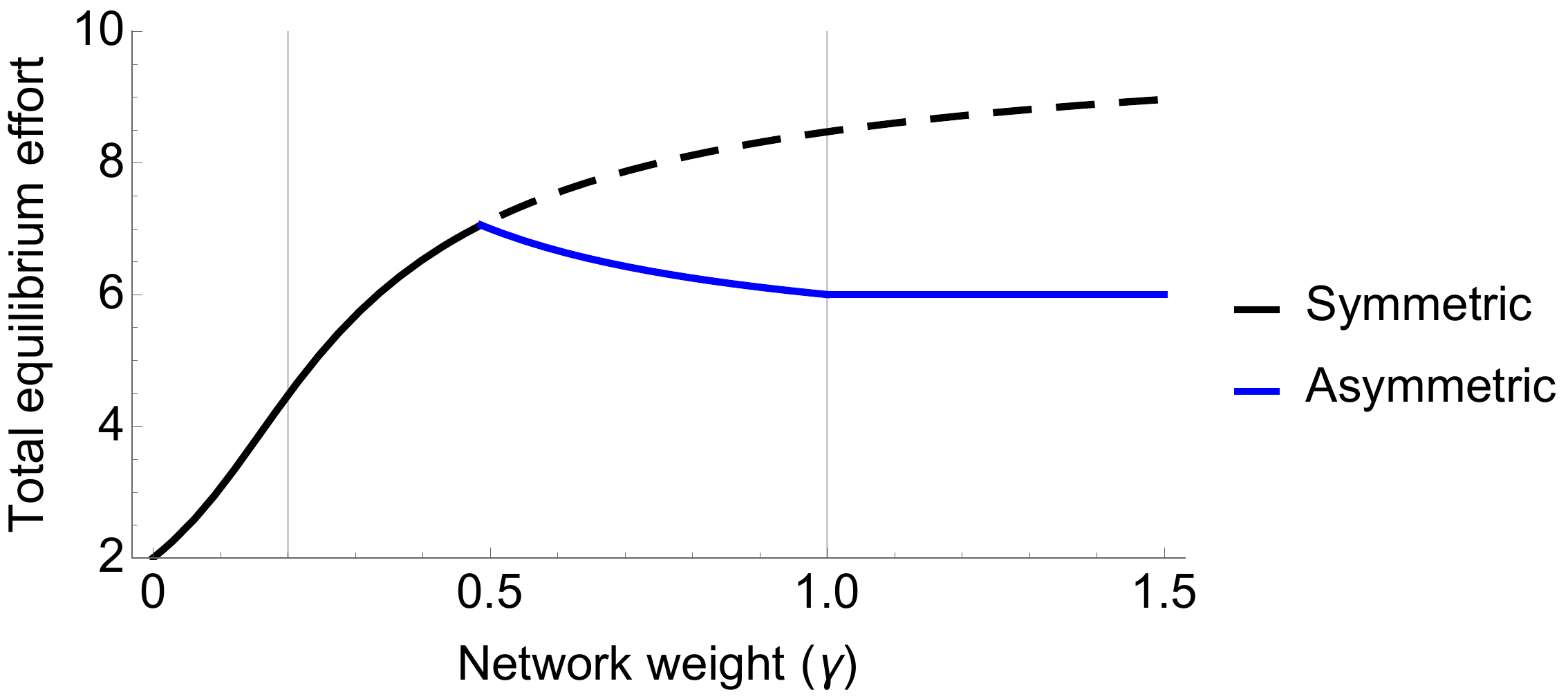}  
\end{center}
\caption{ Total equilibrium effort in the  equilibria of Example \ref{competition} with $N=2$, $G=G_c$, $a=1$, $b=5$. The solid line denotes  asymptotically stable equilibria for the continuous BR dynamics.}
\label{fig:example2_tot}
\end{figure}

\section{Conclusions}
\label{sec:conc}
In this paper, we provide a systematic characterization of the properties of the Nash equilibrium of network games. Our analysis encompass games of strategic complements or substitutes or that feature none of these effects. It applies to games with  multidimensional  and possibly unbounded strategy sets. Our approach exploits the equivalent characterization of Nash equilibria as solutions of variational inequalities and  proceeds in two steps. In the first part of the paper (Section \ref{sec:overview}), we study how the cost function of the agents and different network properties can be related to properties of the game Jacobian (see summary in Table \ref{fig:summary1}). In the second part of the paper (Sections \ref{sec:ex} to \ref{sec:comparative}), we study how the properties of the game Jacobian affect existence and uniqueness of equilibria, convergence of discrete and continuous best response dynamics and comparative statics in network games (see summary in Table \ref{fig:summary2}). 
Our analysis allows a more systematic understanding of the relation between assumptions that are typically used in the literature of network games and provides novel conditions that cover new classes of network games. 

% We next present a set of open questions that emerged from our work: First, in Proposition \ref{prop:suff}c) we show that our new uniform P-matrix condition guarantees that the operator $F$ is a uniform P-function. Verifying whether the same can be obtained by  imposing a uniform bound on $w^\top H_{w,x} A(x) w$ instead of the more strict uniform P-matrix condition is an interesting problem that could have implications for variational inequality theory outside the game theoretic domain. Second, in this work we derived sufficient conditions for uniqueness of the Nash equilibrium. Example \ref{competition} shows that such conditions are not necessary. Finding necessary conditions for special classes of network games is object of future research. Along the same line, the extension of Theorem \ref{thm:scalar} to the multidimensional case  and the study of comparative statics for games with multiple equilibria (as in Example \ref{competition} for large network effects) are  interesting future directions.
%Finally, in Section \ref{sec:br_dynamics} we provided an example to show that strong monotonicity of the game Jacobian does not imply convergence of the discrete best response dynamics. It would be interesting to check convergence properties of other learning dynamics (i.e. replicator dynamics, fictitious play, etc.) under different properties of the game Jacobian. 

\appendix
\renewcommand{\thesection}{\Alph{section}}
\section{ Tools from operator theory, variational inequalities, matrix and spectral theory}
\label{appendix:def}
\subsection{Operator properties, existence and uniqueness for VIs and connection with \cite{rosen1965existence}}
We report in the following definition a slightly weaker version of the properties in Definition \ref{def:mon}.

\begin{definition}\label{def:mon2}
An operator $F:\mathcal{X}\subseteq \R^{Nn}\rightarrow \R^{Nn}$ is
\begin{itemize}
\item[a)] \textup{Strictly monotone:} if $(F(x)-F(y))^\top(x-y)>0$ for all $x,y\in\mathcal{X}, x \neq y$. 
\item[b)] \textup{A  block P-function with respect to the partition $\mathcal{X}=\mathcal{X}^1\times\ldots\times\mathcal{X}^N$:} if $\max_{i\in\N[1,N]} [F^i(x)-F^i(y)]^\top[x^i-y^i]> 0$
for all $x,y\in\mathcal{X}, x \neq y$.
\item[c)] \textup{A  P-function:} if $\max_{h\in\N[1,Nn]} [F(x)-F(y)]_h[x-y]_h> 0$
for all $x,y\in\mathcal{X}, x \neq y$.
\end{itemize}
\end{definition}\

Figure \ref{fig:full} illustrates the relations between the properties in Definition \ref{def:mon} (used in the main text) and the weaker properties in Definition \ref{def:mon2}. 

\begin{figure}[H]
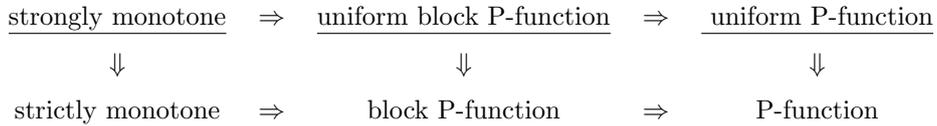

\begin{center}
\begin{tabular}{ccccc}
\underline{strongly monotone}&  $\Rightarrow$&\underline{\vphantom{g}uniform block P-function}& $\Rightarrow$& \underline{\vphantom{g} uniform P-function}\\[0.2cm]
$\Downarrow$ && $\Downarrow$ && $\Downarrow$ \\[0.2cm]
strictly monotone&  $\Rightarrow$&  block P-function& $\Rightarrow$& P-function
\end{tabular}
\end{center}
\caption{Relation between properties of $F$ commonly used in the theory of VIs. The underlined properties  are the ones used in the main text.  }
\label{fig:full}
\end{figure}%

Consider a VI$(\mathcal{X},F)$, where $\mathcal{X}$ is nonempty, closed and convex and $F$ is continuous. If $F$ satisfies any of the properties in the \textit{first line} of Figure \ref{fig:full} then VI$(\mathcal{X},F)$ \textit{admits a unique} solution (for any $\mathcal{X}$, $\mathcal{X}$ a cartesian product and rectangular $\mathcal{X}$, respectively). If $F$ satisfies any of the properties in the \textit{second line} of Figure~\ref{fig:full} then VI$(\mathcal{X},F)$ \textit{admits at most} a solution (for the same options of $\mathcal{X}$), see \cite[Theorem 2.3.3, Proposition 3.5.10]{facchinei2007finite}. Note that in the latter case, existence of a solution is not guaranteed. Existence can however be guaranteed in general if $\mathcal{X}$ is additionally bounded \cite[Corollary 2.2.5]{facchinei2007finite}. 

 In \cite{rosen1965existence} existence and uniqueness of the (Nash equilibrium seen as) solution of the VI$(\mathcal{X},F)$ is shown for $\mathcal{X}$ \textit{compact} and convex under a diagonal strict concavity condition which (for $r=1$ as defined therein) coincides with strict monotonicity. This result is consistent with the known results for existence and uniqueness of the VI solution summarised above. Note that, in this sense, our sufficient conditions for strong monotonicity derived in Theorems \ref{thm:norm} and \ref{thm:min} can be seen as sufficient conditions for the diagonal strict concavity condition  employed in \cite{rosen1965existence}. By proving strong instead of strict monotonicity we are however able to prove existence and uniqueness of the Nash equilibrium for sets $\mathcal{X}$ that are not necessarily bounded (e.g., $\mathcal{X}^i=\R^n_{\ge0}$).  Finally, the case of uniform block P-functions and uniform P-functions are not considered in  \cite{rosen1965existence} and allow for the analysis of network games whose operator is not strictly monotone (see Theorems \ref{thm:inf} and \ref{thm:scalar}). 
  
 \subsection{Auxiliary lemmas and definitions}
\begin{definition}[Vector and matrix norms]
Given a vector $v\in\R^d$ and a matrix $A\in\R^{d\times d}$ the following definitions hold
\begin{enumerate} 
\item  $\|v\|^2_Q=v^\top Q v$ for $Q=Q^\top\succ 0$ and $\|v\|_2=\|v\|_I$ 
\item $\|A\|_2:=\sqrt{\lambda_{\textup{max}}(A^\top A)}$ (spectral norm)
\item $\|A\|_1:=\max_j \sum_{i=1}^d |[A]_{ij}|$ (max column sum)
\item $\|A\|_\infty:=\max_i \sum_{j=1}^d |[A]_{ij}|$ (max row sum)
\end{enumerate}
\end{definition}
\begin{lemma}[Matrix norms]
Given a matrix $A\in\R^{d\times d}$ the following relations between matrix norms hold
\begin{enumerate}
\item $\|A\|_2\le \sqrt{\|A\|_1 \|A\|_\infty}$
\item $\rho(A)\le \|A\|_2$
\item $\frac{1}{\sqrt{d}} \|A\|_\infty\le  \|A\|_2 \le \sqrt{d}  \|A\|_\infty$
\item $\|A\|_2=\|A^\top\|_2$
\end{enumerate}
\end{lemma}

\begin{lemma}[Properties of the Kronecker product]\label{lem:k} Consider any matrices $A,B,C,D$ of suitable dimensions so that the following statements are well-defined. The following properties hold.
\begin{enumerate}
\item $\|A \otimes B \|_2=\|A\|_2\|B\|_2$
\item $\Lambda(A \otimes B )=\Lambda(A) \Lambda(B):=\{\lambda_A\lambda_B \mid \lambda_A\in\Lambda(A), \lambda_B\in\Lambda(B)\}$
\item $(A\otimes B) (C\otimes D)= (AC\otimes BD)$
\item $(A\otimes B)+ (A\otimes C)=(A\otimes (B+C))$
\item $(A\otimes B)^\top =(A^\top\otimes B^\top)$
\item $\alpha(A\otimes B) =(\alpha A\otimes B)=(A\otimes \alpha B)$
\end{enumerate}
\end{lemma}

\begin{lemma}[Properties symmetric matrices]\label{lem:s}
Consider a symmetric matrix $A\in\R^{d\times d}$, $A=A^\top$. The following statements hold
\begin{enumerate}
\item $A\succeq \lambda_{\textup{min}}(A) I_d$
\item $|\lambda_{\textup{min}}(A)|\le \rho(A)=\|A\|_2$
\item $\lambda_{\textup{min}}(\alpha I+A)=\alpha +\lambda_{\textup{min}}(A)$
\end{enumerate}
\end{lemma}

\begin{lemma}\label{lem:d_reg}
Consider a non-negative matrix $A\in\R^{N\times N}$ such that $A\mathbbm{1}_N=d\mathbbm{1}_N$. Then $\rho(A)=d$.
\end{lemma}
\begin{proof}
This fact is known, we report the proof for completeness.  Note that clearly $d$ is an eigenvalue of $A$, hence $\rho(A)\ge d$. We conclude the proof by showing  $\rho(A)\le d$. Let $\lambda$ be the eigenvalue that achieves the spectral radius, that is, $|\lambda|=\rho(A)$ and let $v$ be the corresponding eigenvector. Moreover, let $i\in\N[1,N]$ be the index such that $|v_i|=\max_{j=1}^N |v_j|$, note that $|v_i|>0$. From the $i$th component of $Av=\lambda v$ we get
$|\lambda|| v_i|=|\lambda v_i|= |\sum_{j=1}^N A_{ij}v_j| \le \sum_{j=1}^N |A_{ij}||v_j| \le \sum_{j=1}^N A_{ij} |v_i| = d  |v_i| .$
Dividing by $| v_i|> 0$, we get $\rho(A)=|\lambda|\le d$.
\end{proof}

 \begin{lemma}\label{lemma:pot}
Consider the linear quadratic game of Example \ref{ex:lq}. If there exists $\{\beta^i>0\}_{i=1}^N$ such that for all $i,j\in\{1,N\}$
\begin{equation}\label{bra_ext}
\frac{ K^i}{\beta_i} G_{ij}= \frac{K^j}{\beta_j}  G_{ji}
\end{equation}
then the game with rescaled coordinates $\tilde x_i=\frac{1}{\sqrt{\beta_i}}  x_i$ has parameters $\tilde K^i=\frac{K^i}{\beta_i}$, network $\tilde G_{ij}={G_{ij}}{\sqrt{\beta_i\beta_j}}$
and admits an exact potential.
\end{lemma}
\begin{proof}
By making the  change of coordinates  $x_i=\sqrt{\beta^i} \tilde x_i$
the cost function can be reformulated as
$
J^i(\tilde x)=\frac12 (\sqrt{\beta_i}\tilde x_i)^2-a^i \sqrt{\beta_i}\tilde x_i + K^i\sum_{j}G_{ij} \sqrt{\beta_i}\tilde x_i \sqrt{\beta_j}\tilde x_j 
=\beta_i\frac12 (\tilde x_i)^2-a^i \sqrt{\beta_i}\tilde x_i + K^i\sum_{j}\tilde G_{ij} \tilde x_i \tilde x_j
$
where $\tilde G_{ij}:= G_{ij} \sqrt{\beta_i}\sqrt{\beta_j}$.  Minimizing $J^i(\tilde x)$ is equivalent to minimizing
$
\tilde J^i(\tilde x)=\frac{J^i(\tilde x)}{\beta_i}=\frac12 (\tilde x_i)^2- \frac{a^i}{\sqrt{\beta_i}}\tilde x_i + \tilde K^i\sum_{j}\tilde G_{ij} \tilde x_i \tilde x_j
$
where $\tilde K^i:=\frac{K^i}{\beta_i}$. This game has an exact potential if and only if $\tilde K^i \tilde{G}_{ij}=\tilde K^j \tilde{G}_{ji} $
which is true by assumption.
\end{proof}
 \begin{remark}
The following special cases of Lemma \ref{lemma:pot} are discussed in \cite{bramoulle2014strategic}.
\begin{enumerate}
\item If $G=G^\top$ the change of coordinates corresponding to $\beta^i={K^i}$ leads to a  game  that is homogeneous ($\tilde K^i=1$) and has a symmetric network.
\item If  there exists $\{\alpha_i>0\}_{i=1}^N$ such that 
$\alpha_i G_{ij}= \alpha_j   G_{ji} $ for all $i,j\in\{1,N\}$, 
the change of coordinates corresponding to  $\beta_i=\frac{K_i}{\alpha_i}$ leads to a game with $\tilde K^i=\alpha_i$ and  $\tilde G_{ij}=G_{ij} \frac{\sqrt{K^iK^j} }{\sqrt{\alpha_i\alpha_j}}$. This is equivalent to  a game with $\bar K^i=1$ and symmetric network $\bar G_{ij}= \alpha_i \tilde G_{ij}$ (since $\bar G_{ij}= \alpha_i \tilde G_{ij}=  \alpha_i G_{ij} \frac{\sqrt{K^iK^j} }{\sqrt{\alpha_i\alpha_j}} =\alpha_j G_{ji} \frac{\sqrt{K^iK^j} }{\sqrt{\alpha_i\alpha_j}}=\bar G_{ji}$ by assumption). 
\end{enumerate}
\end{remark}

\section{Table 2: Technical statements and counter examples}
\label{sec:game_jacobian}

 In this section  we detail the results of Table \ref{fig:summary1}. To help with the understanding of the following theorems we consider as benchmark the simple linear quadratic game presented in Example \ref{ex:lq}. This model belongs to the class of network games that are typically considered in the literature and allow us to easily compare our results. Moreover, we   show in the next subsections that, despite its simplicity,  a vast range of different strategic interactions can be captured with it 
by simply tuning the values of the parameters $K^i$ (modeling the weight of the neighbor aggregate on each agent $i$).

\subsection{\textbf{A condition in terms of $\|\PP\|_2$}}

As discussed in Section \ref{sec:overview} the strongest result is obtained  under Assumption \ref{ass:gen}a, which depends on $\|\PP\|_2$. In this case,    both strong monotonicity and the $P_\Upsilon$ condition hold.

\begin{theorem}[Strong monotonicity and $\boldsymbol{P_\Upsilon}$ condition under Assumption \ref{ass:gen}a]\label{thm:norm}
Suppose that Assumptions \ref{cost} and \ref{ass:gen}a hold. Then the game Jacobian $F$ as defined in \eqref{eq:F}
\begin{enumerate}
\item is strongly monotone with constant $\alpha_2$,
\item satisfies the $P_\Upsilon$ condition.
\end{enumerate}
\end{theorem}

For cases when Assumption \ref{ass:gen}a is not satisfied (see Section \ref{sec:overview}), we develop in the following subsections alternative guarantees in terms of  Assumptions \ref{ass:gen}b and  \ref{ass:gen}c.

\subsection{\textbf{A condition in terms of $\|\PP\|_\infty$}}

 As detailed in Section \ref{sec:overview1},  Assumption \ref{ass:gen}b  is useful for asymmetric networks where each agent is influenced by a relatively small number of agents (i.e.,  $\|\PP\|_\infty$ is small). The next theorem shows that this  is sufficient to prove the 
$P_\Upsilon$ condition.

\begin{theorem}[$\boldsymbol{P_\Upsilon}$ condition  under Assumption \ref{ass:gen}b]\label{thm:inf}
Suppose that  Assumptions \ref{cost} and \ref{ass:gen}b hold. Then the game Jacobian $F$ as defined in \eqref{eq:F} satisfies the $P_\Upsilon$ condition.
\end{theorem}

On the other hand, under Assumption \ref{ass:gen}b, it is not possible to guarantee strong monotonicity. An illustrative example is given next. 

{ \begin{example}[Assumption \ref{ass:gen}b: $\boldsymbol{P_\Upsilon}$ but not strongly monotone]\label{ex:PnotS}
Consider the linear quadratic game in Example \ref{ex:lq} with $N=6$ players 
over an asymmetric star as in  Figure \ref{fig:net}d) and $K^i=0.9$ for all $i\in\N[1,6]$. Then  
$$\alpha_\infty:= 1 - 0.9\cdot 1 >0,$$
hence Theorem \ref{thm:inf} holds and $F$ satisfies the $P_\Upsilon$ condition. On the other hand, once can see that $\nabla_x F(x)=I_6+0.9 \left[\begin{smallmatrix}0 & 0  \\ \mathbbm{1}_5 & 0\end{smallmatrix}\right]$ and $\frac{\nabla_x F(x)+\nabla_x F(x)^\top}{2}=I_6+\frac{0.9}{2} \left[\begin{smallmatrix}0 & \mathbbm{1}_5^\top  \\ \mathbbm{1}_5 & 0\end{smallmatrix}\right]$ has a negative eigenvalue. Hence $F$ is not monotone. This is consistent with the fact that $\alpha_2=1 - 0.9\cdot \sqrt{5} <0$, hence Assumption \ref{ass:gen}a is violated. 
 Following the discussion in Section \ref{diss_smon}, we note that Assumption~\ref{ass:gen}b does not guarantee strong monotonicity because it provides a bound on $\|\PP\|_\infty$ but not on $\|\PP\|_1$. In this example for instance $\|\PP\|_1=N-1 \gg \|\PP\|_\infty=1$.
\end{example} }

\subsection{\textbf{Conditions in terms of $|\lambda_{\textup{min}}(G)|$}}
Finally, we consider Assumption \ref{ass:gen}c which is formulated in  terms of $|\lambda_{\textup{min}}(G)|$ and therefore can be applied only to symmetric networks, for which $|\lambda_{\textup{min}}(G)|$  is well defined.\footnote{Even if $|\lambda_{\textup{min}}(G)|$ was well defined for an asymmetric network, Theorem \ref{thm:min} and Theorem \ref{thm:scalar} would still not apply.} We recall that for symmetric  networks Assumption \ref{ass:gen}c is the least restrictive condition (see Section \ref{sec:overview2} and Figure \ref{fig:sets}). Because of such generality, Assumption \ref{ass:gen}c alone is  not sufficient to guarantee any of the properties in Definition~\ref{def:mon}, as illustrated in the next example. In this subsection, we derive additional conditions to complement Assumption~\ref{ass:gen}c and guarantee that at least one of the properties in Definition \ref{def:mon} is satisfied. 

{ \begin{example}[Assumption \ref{ass:gen}c:  multiple equilibria]\label{ex2}
Consider the  linear quadratic game in Example \ref{ex:lq} with $K^i=-\frac{1}{N-1}$ and $a^i=0$ for all $i\in\N[1,N]$ played over the complete network $G_c:=\mathbbm{1}_N\mathbbm{1}_N^\top-I_N$. Then

$$\textstyle \alpha_{\textup{min}}:= \kappa_1-\kappa_2 |\lambda_{\textup{min}}(\PP_c)|=1-\frac{1}{N-1}>0,$$
hence Assumption \ref{ass:gen}c is met for $N>2$.
From Proposition \ref{prop:kkt}, a vector $x$ is a Nash equilibrium if and only if it solves the VI$(\mathcal{X}, F)$. This is equivalent to the requirement

\begin{equation}
\label{eq:fp}x=\Pi_{\mathcal{X}}[x-F(x)],
\end{equation}
see \cite[Proposition 1.5.8]{facchinei2007finite}. In this example $\mathcal{X}=\mathbb{R}^N_{\ge 0}$ and $F(x)=[I_N -\frac{1}{N-1} G_c] x$. For any $\beta\in\R$ we get 

$$\textstyle F(\beta\mathbbm{1}_N)=[I_N -\frac{1}{N-1} G_c] \beta\mathbbm{1}_N= \beta\mathbbm{1}_N - \beta \frac{1}{N-1} G_c \mathbbm{1}_N= \beta\mathbbm{1}_N - \beta \frac{N-1}{N-1} \mathbbm{1}_N=0.$$
Consequently, any vector $x=\beta \mathbbm{1}_N$ with $\beta\ge0$ solves  the fixed point equation \eqref{eq:fp} and this game has infinitely many Nash equilibria. None of the properties in Definition~\ref{def:mon} can therefore hold (recall that  they all imply uniqueness by Proposition \ref{prof:ex}). This example does not contradict Theorem \ref{thm:norm}, since 

$$\alpha_2:=\textstyle \kappa_1-\kappa_2 \|\PP_c\|=1-\frac{1}{N-1}(N-1)=0$$
hence Assumption \ref{ass:gen}a is violated. We finally note that, since $K^i<0$ for all $i$,  this is a game of strategic complements. \end{example}}

The previous example  proves that Assumption \ref{ass:gen}c is not enough to guarantee uniqueness of the Nash equilibrium in games of strategic complements. We  show in the following  that Assumption~\ref{ass:gen}c is instead sufficient for games of strategic substitutes. The intuition for this discrepancy is that for games with the simple structure in Example \ref{ex:lq}, the properties in Definition \ref{def:mon} are related to the minimum eigenvalue of  $I+K\PP$ (i.e., $\nabla_x F(x)$) being positive.  For simplicity consider the homogeneous weight case where $K^i=\kappa$ for all $i$ so that $K=\kappa I_N$. It is then clear that if $\kappa>0$ (substitutes case), the minimum eigenvalue of $I+\kappa\PP$  is related to the $\lambda_{\textup{min}}(\PP)$ hence Assumption \ref{ass:gen}c suffices, while if $\kappa<0$ (complements case) the minimum eigenvalue of $I+\kappa\PP$  is related to  $\lambda_{\textup{max}}(\PP)$ hence the stronger  Assumption \ref{ass:gen}a is needed. This argument is made rigorous and generalized to nonlinear and multidimensional network games in the following theorems.

\begin{theorem}[Strong monotonicity  under Assumption \ref{ass:gen}c]\label{thm:min}
Suppose that  Assumptions~\ref{cost} and \ref{ass:gen}c hold. Moreover, assume that
\begin{enumerate}
\item $K^i(x)=\tilde K(x)$ for all $i\in\N[1,N]$ and all $x\in\mathcal{X}$;
\item $\tilde K(x)+\tilde K(x)^\top \succeq 0$ for all $x\in\mathcal{X}$.
\end{enumerate}
Then the game Jacobian $F$ as defined in \eqref{eq:F} is strongly monotone with constant $\alpha_{\textup{min}}$.
\end{theorem}

In Theorem \ref{thm:min}, we guarantee strong monotonicity under Assumption \ref{ass:gen}c instead of Assumption~\ref{ass:gen}a by additionally assuming that 1) the neighbor aggregate $z^i$ affects the cost of each agent in the same way (i.e., $\frac{\partial J^i}{\partial x^iz^i}=\tilde K(x)$ for all $i$) and 2) the symmetrized version of $\tilde K(x)$  is positive semidefinite.
Such conditions are unfortunately not  sufficient to  guarantee  the $P_\Upsilon$ condition,  as illustrated next.

{ \begin{example}[Assumption \ref{ass:gen}c: Strongly monotone not $\boldsymbol{P_\Upsilon}$]\label{ex4}
Consider the  linear quadratic game in Example \ref{ex:lq}, with $N=10$, $K^i=0.5>0$ for all $i\in\N[1,N]$ and $\PP=\PP_c$. Then
$$\alpha_{\textup{min}}:=\kappa_1-\kappa_2 |\lambda_{\textup{min}}(\PP_c)|=1-0.5=0.5>0,$$
hence all the assumptions of Theorem \ref{thm:min}  are satisfied and $F$ is strongly monotone.
Nonetheless, the matrix $\Upsilon$ in \eqref{upsi} is equal to $I_{10} - 0.5 G_c$ and has minimum eigenvalue $\lambda_{\textup{min}}(\Upsilon)=1  - 0.5 \lambda_{\textup{max}}(G_c) = 1 - 0.5 \cdot 9 <0$. Hence $\Upsilon$ is not a P-matrix (all the real eigenvalues of a P-matrix are positive as shown in \cite[Theorem 3.3]{fiedler1962matrices}). Condition $P_\Upsilon$ is thus violated. Intuitively this happens because the condition $P_\Upsilon$ is agnostic to the sign of the $K^i$s (since $\kappa_{i,j}=|K^i|G_{ij}$) and is therefore overly conservative for cases when the $K^i$ are positive (substitute). Note that (for scalar games) this is not the case for the uniform P matrix condition in Definition \ref{def:P_matrix_uniform}. By Lemma \ref{lemma:Paffine} this condition is satisfied if there exists $H$ diagonal and positive definite such that $H(I_{10}+0.5 \PP_c)+(I_{10}+0.5 \PP_c)H\succ 0$. This is clearly the case for $H=I_{10}$ since $I_{10}+0.5 \PP_c\succ 0$. Note that the uniform P matrix condition accounts for the fact that the $K^i$ are positive (i.e., depends on $I_{10}+0.5 \PP_c$ instead of $I_{10}-0.5 \PP_c$).
\end{example}}

Both conditions in Theorem \ref{thm:min} are needed to guarantee strong monotonicity. In fact, 
Example \ref{ex2} shows that if condition 1) holds but condition 2) does not, then $F$ might not be strongly monotone. The next example shows that the same is true if condition 2) holds but condition 1) does not, that is,  if we consider a game of strategic substitutes with  heterogeneous weights.

\begin{example}[Assumption \ref{ass:gen}c: Uniform P-function not strongly monotone, not $\boldsymbol{P_\Upsilon}$]\label{ex3}
Consider the  linear quadratic game in Example \ref{ex:lq}, with $N=10$, $K^i=0.1>0$ for all $i\in\Z[1,9]$, $K^{10}=0.9>0$ and $\PP=\PP_c$. Then

$$\alpha_{\textup{min}}:=\kappa_1-\kappa_2 |\lambda_{\textup{min}}(\PP_c)|=1-0.9=0.1>0,$$
hence Assumption \ref{ass:gen}c is satisfied, $K^i>0$ for all $i$ but $K^{10}\neq K^i$ for $i\in\N[1,9]$.
One can verify that the following quantity

\begin{equation}\label{setp_ex1}
\lambda_{\textup{min}}\left(\frac{\nabla_x F(x)+\nabla_x F(x)^\top}{2}\right)=\lambda_{\textup{min}}\left(I+\frac{K\PP_c+\PP_cK}{2}\right)
\end{equation} 
is negative, hence $F$ is not strongly monotone (it is not even monotone). 
The reason why with heterogeneous weights, strong monotonicity may fail  is that in this case the matrix $K$ is diagonal but not scalar.  Consequently, even though $K^i>0$ for all $i$ it might happen that

\begin{equation}\label{setp_ex2}
\lambda_{\textup{min}}\left(I+\frac{K\PP_c+\PP_cK}{2}\right) < 1+ \lambda_{\textup{max}}(K)  \lambda_{\textup{min}}(\PP_c) = 1- \kappa_2 |\lambda_{\textup{min}}(\PP_c)|=: \alpha_{\textup{min}}.
\end{equation}
Consequently, one \textup{cannot} use Assumption \ref{ass:gen}c to lower bound the left hand side of \eqref{setp_ex2} by  a positive quantity (i.e., $\alpha_{\textup{min}}$), which is what would be needed to prove that $F$ is strongly monotone.
The $P_\Upsilon$ condition is  not satisfied either because in this example $\Upsilon=I-K\PP_c$ and $\lambda_{\textup{min}}(\Upsilon)<0$, hence $\Upsilon$ is not a P-matrix. By choosing $H=K^{-1}$ it is instead easy to show that $H(I+KG_c)+(I+G_cK)H= 2K^{-1}+ 2G_c \succ 0$. It follows from Lemma \ref{lemma:Paffine} and Proposition \ref{prop:suff}c) that the uniform P-matrix condition holds and $F$ is a uniform P-function. Again note  that the difference between the $P_\Upsilon$ and the uniform P-matrix condition reduces in this simple case to the difference between $I-K\PP_c$ and $I+K\PP_c$.
\end{example}

In the previous example, the operator is not strongly monotone nor satisfies the $P_\Upsilon$ condition. Nonetheless,  uniqueness of the Nash equilibrium can still be guaranteed since $F$ is a uniform P-function. We next prove that for scalar games (i.e., when $n=1$) this is true in general, as long as the $K^i(x)$ are positive.

\begin{theorem}[Uniform P-matrix condition for scalar games of strategic substitutes]\label{thm:scalar}
 Suppose that  Assumptions \ref{cost} and \ref{ass:gen}c hold and that $n=1$. Then 
\begin{enumerate}
\item if there exists $\nu>0$ such that $K^i(x)\ge \nu$ for all $i\in\N[1,N]$ and all $x\in\mathcal{X}$ then $\nabla_x F(x)$ satisfies the uniform P-matrix condition and the game Jacobian $F$ as defined in \eqref{eq:F} is a uniform P-function;
\item  if  $K^i(x)>0$ for all $i\in\N[1,N]$ and all $x\in\mathcal{X}$ then $\nabla_x F(x)$ is a P-matrix for all $x\in\mathcal{X}$ and  $F(x)$ is a P-function.
\end{enumerate}
\end{theorem}

\section{Table 2: Proofs of technical statements}
\label{appendix:proof}

\subsection*{\textbf{Proof of Theorem \ref{thm:norm}: Strong monotonicity and $\boldsymbol{P_\Upsilon}$ condition under Assumption \ref{ass:gen}a}}

\begin{enumerate}
 \item To prove strong monotonicity note that $D(x)$ in \eqref{eq:DK} is a symmetric block diagonal matrix and, by definition \eqref{k1}, $\lambda_{\textup{min}}(D(x))\ge \kappa_1$. It follows from Lemma \ref{lem:s} in Appendix \ref{appendix:def} that  $D(x)\succeq \kappa_1 I$ for all $x\in\mathcal{X}$.  By \eqref{eq:grad_F}  it then holds

\begin{equation}\label{eq:start}
 \frac{\nabla_x F(x)+\nabla_x F(x)^\top}{2}\succeq   \kappa_1 I +\frac{K(x) \W + \W ^\top K(x)^\top}{2}.
\end{equation}

Note that 

\begin{align*}
\lambda_{\textup{min}}\left(\frac{K(x) \W + \W ^\top K(x)^\top}{2}\right)&\ge -\rho\left(\frac{K(x) \W + \W ^\top K(x)^\top}{2}\right) = - \Big\|\frac{K(x) \W + \W ^\top K(x)^\top}{2}\Big\|_2\\
&\ge - \frac{\|K(x) \W \|_2+ \|\W ^\top K(x)^\top\|_2}{2} = - \|K(x) \W \|_2 \ge  -\|K(x)\|_2  \| \PP \|_2  ,\end{align*}
where we used the fact that for symmetric matrices all eigenvalues are real, the spectral radius equals the $2$-norm (see Lemma~\ref{lem:s}), the triangular inequality, the fact that for any matrix $A$ it holds $\|A\|_2=\|A^\top\|_2$ and $\|\W \|_2=\|\PP  \otimes I_n\|_2=\|\PP \|_2$ (see Lemma \ref{lem:k}).
Moreover, note that 

\begin{align}\label{k2_bound}
\|K(x)\|_2 &\le \max_{x} \|K(x)\|_2 =  \max_{x} \max_{i} \|K^i(x)\|_2  =  \max_{i} \max_{x} \|K^i(x)\|_2  =  \kappa_2
\end{align}
since the norm of a block diagonal matrix equals the largest norm of its blocks.
Using the last result we can immediately see that 

 \begin{align*}
& \lambda_{\textup{min}}\left(\kappa_1 I + \frac{K(x) \W + \W ^\top K(x)^\top}{2}\right)  = \kappa_1 + \lambda_{\textup{min}}\left(\frac{K(x) \W + \W ^\top K(x)^\top}{2}\right) \ge \kappa_1 - \kappa_2 \|\PP\|_2 = \alpha_2.
\end{align*}
Hence \begin{equation}\label{eq:intera} \begin{aligned}
 \frac{\nabla_x F(x)+\nabla_x F(x)^\top}{2} & \succeq \lambda_{\textup{min}}\left(\kappa_1 I + \frac{K(x) \W + \W ^\top K(x)^\top}{2}\right)  I_{Nn} \succeq \alpha_2 I_{Nn}.
\end{aligned} 
\end{equation}
Since by Assumption \ref{ass:gen}a $\alpha_2>0$, equation \eqref{eq:intera} implies that $F(x)$ is strongly monotone with constant $\alpha_2$.
\item For network games the matrix $\Upsilon$ in Definition \ref{def:P_up} can be rewritten equivalently with the notation in \eqref{k1}  as
 
 $$\Upsilon:=\left[\begin{array}{cccc}\kappa_1^1 & -\kappa_2^1 \PP _{12} & \hdots &- \kappa_2^1 \PP _{1N} \\-\kappa_2^2 \PP _{21} & \ddots & 0 & -\kappa_2^2 \PP _{2N} \\\vdots & 0 & \ddots & \vdots \\-\kappa_2^N \PP _{N1} & \hdots & - \kappa_2^N \PP _{N(N-1)} & \kappa_1^N\end{array}\right] \in\R^{N\times N}. $$
Let us define
$M:= \kappa_1 I_N - \kappa_2\PP \in\R^{N\times N}$ and note that 
under  Assumption \ref{ass:gen}a 

\begin{align*}
\lambda_{\textup{min}}\left(\frac{M+M^\top}{2}\right)&=\lambda_{\textup{min}}\left(\kappa_1 I_N - \kappa_2 \frac{\PP+\PP^\top}{2}\right) =\kappa_1  - \kappa_2\lambda_{\textup{max}}\left( \frac{\PP+\PP^\top}{2}\right)\\&=\kappa_1  - \kappa_2\Big\| \frac{\PP+\PP^\top}{2}\Big\|_2 \ge \kappa_1  - \kappa_2\| \PP\|_2 =\alpha_2 >0.
\end{align*}
Hence $\frac{M+M^\top}{2} \succ 0$ and for any $w\neq 0$ it holds $w^\top \frac{M+M^\top}{2}w =w^\top M w = \sum_{h=1}^N w_h [Mw]_h>0$. Consequently, there exists $h\in\N[1,N]$ such that $w_h [Mw]_h>0$ and $M$ is a P-matrix. Since both $\Upsilon$ and $M$  are Z-matrices (i.e., all the elements outside the diagonal are non-positive) and $\Upsilon$ is greater or equal than $M$ element-wise, the fact that $M$ is a P-matrix implies that $\Upsilon$ is a P-matrix \cite[Theorem 4.2]{fiedler1962matrices}. Therefore the $P_\Upsilon$ condition holds.
\end{enumerate}

\subsection*{\textbf{Proof of Theorem \ref{thm:inf}:  $\boldsymbol{P_\Upsilon}$ condition  under Assumption \ref{ass:gen}b}}

Consider the matrix $M:= \kappa_1 I_N - \kappa_2\PP \in\R^{N\times N}$ as defined in the proof of Theorem \ref{thm:norm} part 2. Note that for each $i$ it holds $\sum_{j\neq i} |M_{ij}| \le \kappa_2 \|G\|_\infty$. By Gershgorin theorem all the eigenvalues of $M$ are therefore contained in a circle of radius $\kappa_2 \|G\|_\infty$ centered in $(\kappa_1,0)$. Since $\kappa_1-\kappa_2 \|G\|_\infty>0$ all the  eigenvalues of $M$ have positive real part. The same argument can be applied to any principal sub-matrix of $M$ to conclude that all its eigenvalues have positive real part. 
This is a sufficient condition for $M$ to be a P-matrix. The conclusion follows as in Theorem \ref{thm:norm} part 2.

\subsection*{\textbf{Proof of Theorem \ref{thm:min}: Strong monotonicity under Assumption \ref{ass:gen}c}}
The proof of this statement is similar to the proof of Theorem \ref{thm:norm} part 1. The only difference is the bound on $\lambda_{\textup{min}}\left(\frac{K(x) \W + \W ^\top K(x)^\top}{2}\right)$. Note that under Assumption \ref{ass:gen}c the matrix $\PP $ is   symmetric. It follows from Lemma \ref{lemma:P} that $\lambda_{\textup{min}}(\PP )<0$. 
Moreover, since $K^i(x)=\tilde K(x)$ for all $x$ we can write  $K(x)=(I_N\otimes \tilde K(x))$. Consequently, $K(x)\W =(I_N\otimes \tilde K(x))(\PP \otimes I_n)=(\PP \otimes \tilde K(x))$. By the properties of the Kronecker product, 

\begin{equation*}
\begin{aligned}
 \lambda_{\textup{min}}\!\left(\!\frac{K(x) \W + \W ^\top K(x)^\top}{2}\!\right)\!\!&= \lambda_{\textup{min}}\left(\frac{(\PP \otimes \tilde K(x))+(\PP \otimes \tilde K(x))^\top}{2}\right)=\lambda_{\textup{min}}\left(\frac{(\PP \otimes \tilde K(x))+(\PP^\top \otimes \tilde K(x)^\top)}{2}\!\!\right)\\
&=\lambda_{\textup{min}}\left(\frac{(\PP \otimes \tilde K(x))+(\PP \otimes \tilde K(x)^\top)}{2}\right)
=\lambda_{\textup{min}}\left(\PP \otimes \frac{ \tilde K(x)+\tilde K(x)^\top}{2}\right)\\
& =  \lambda_{\textup{min}}(\PP )  \lambda_{\textup{max}}\left(\frac{ \tilde K(x)+\tilde K(x)^\top}{2}\right)  = -| \lambda_{\textup{min}}(\PP ) | \Big\|\frac{ \tilde K(x)+\tilde K(x)^\top}{2}\Big\|_2\\
&\ge\!-| \lambda_{\textup{min}}(\PP )|  \frac{ \|\tilde K(x)\|_2\!+\!\|\tilde K(x)^\top\|_2}{2}\!=\!-|\lambda_{\textup{min}}(\PP )| \|\tilde K(x)\|_2
\!=\!-|\lambda_{\textup{min}}(\PP )| \|K(x)\|_2,
\end{aligned}
\end{equation*}
where we used the fact that the eigenvalues of the Kronecker product are the product of the eigenvalues,  all the eigenvalues of $\tilde K(x)+\tilde K(x)^\top$ are non-negative,  $ \lambda_{\textup{max}}(\frac{ \tilde K(x)+\tilde K(x)^\top}{2})=\|\frac{\tilde K(x)+\tilde K(x)^\top}{2}\|_2$ and $\|\tilde K(x)\|_2=\|K(x)\|_2$, since $K(x)= I_N \otimes \tilde K(x)$ (see Lemmas  \ref{lem:k} and \ref{lem:s}).
Combining this new bound with   \eqref{k2_bound} we get

 \begin{align*}
& \lambda_{\textup{min}}\left(\kappa_1 I + \frac{K(x) \W + \W ^\top K(x)^\top}{2}\right)  = \kappa_1 + \lambda_{\textup{min}}\left(\frac{K(x) \W + \W ^\top K(x)^\top}{2}\right) \ge \kappa_1 - \kappa_2 |\lambda_{\textup{min}}(\PP )| = \alpha_{\textup{min}}.
\end{align*}
Hence from \eqref{eq:start} we get

\begin{equation}\label{eq:inter} \begin{aligned}
 \frac{\nabla_x F(x)+\nabla_x F(x)^\top}{2} & \succeq \lambda_{\textup{min}}\left(\kappa_1 I + \frac{K(x) \W + \W ^\top K(x)^\top}{2}\right)  I_{Nn} \succeq \alpha_{\textup{min}} I_{Nn}.
\end{aligned} 
\end{equation}
Since by Assumption \ref{ass:gen}c $\alpha_{\textup{min}}>0$, equation \eqref{eq:inter} implies that $F(x)$ is strongly monotone with constant $\alpha_{\textup{min}}$.

\subsection*{\textbf{Proof of Theorem \ref{thm:scalar}: Uniform P-matrix condition for scalar games of strategic substitutes}}

 To  prove the first statement  we prove that $\nabla_x F(x)\in\R^{N\times N}$ satisfies the uniform P-matrix condition. The fact that $F$ is a uniform P-function then follows from Proposition \ref{prop:suff}c). Let us consider any set of $N$ vectors $\{x_{[h]}\in\mathcal{X}\}_{h=1}^N$,
and the corresponding matrix

$$A_x:=\left[\begin{smallmatrix} [ \nabla_x F(x_{[1]})]_{(1,:)}\\ \vdots \\   [ \nabla_x F(x_{[N]})]_{(N,:)} \end{smallmatrix}\right]\in\R^{N\times N}.$$
Because of formula \eqref{eq:grad_F} and since $n=1$ we get

$$A_x:=\left[\begin{smallmatrix} D^1(x_{[1]}) & &\\ & \ddots  &\\ & &   D^N(x_{[N]})\end{smallmatrix}\right] + \left[\begin{smallmatrix} K^1(x_{[1]}) & &\\ & \ddots  &\\ & &   K^N(x_{[N]}) \end{smallmatrix}\right] \PP.$$
Note that since $n=1$ for all   $\{x_{[h]}\in\mathcal{X}\}_{h=1}^N$ and $h\in\N[1,N]$,   $D^h(x_{[h]})$ and $K^h(x_{[h]})$ are positive numbers.
In particular, 

$$H_x:= \left[\begin{smallmatrix} K^1(x_{[1]}) & &\\ & \ddots  &\\ & &   K^N(x_{[N]}) \end{smallmatrix}\right]^{-1}$$
is well defined, diagonal and positive definite.
Overall, we get

$$H_xA_x=\left[\begin{smallmatrix} \frac{D^1(x_{[1]})}{K^1(x_{[1]})} & &\\ & \ddots  &\\ & &   \frac{D^N(x_{[N]})}{K^N(x_{[N]})} \end{smallmatrix}\right] + \PP.$$
Note that $H_xA_x$ is symmetric, therefore we can use Lemma \ref{lem:s} and conclude that 

$$H_xA_x \succeq \left(\min_h \frac{D^h(x_{[h]})}{K^h(x_{[h]})} +\lambda_{\textup{min}}(\PP)\right) I_N\succeq \left( \frac{\min_h D^h(x_{[h]})}{\max_h K^h(x_{[h]})} +\lambda_{\textup{min}}(\PP)\right) I_N \succeq \left(\frac{\kappa_1}{\kappa_2}+ \lambda_{\textup{min}}(\PP )\right) I_N=:\eta I_N.$$
It follows that if for any $w\in\R^N$ we set $H_{x,w}=H_x$, it holds $w^\top H_{x,w}A_x  w=w^\top H_xA_x  w \ge \eta\|w\|_2^2$. Moreover, by Assumption~\ref{ass:gen}c we have that $\eta=\frac{\alpha_{\textup{min}}}{\kappa_2}>0$ (recall $\lambda_{\textup{min}}(\PP )<0$ as by Lemma \ref{lemma:P}).
Finally note that for any fixed set of  vectors $\{x_{[h]}\in\mathcal{X}\}_{h=1}^N$,

$$\max_w \lambda_{\textup{max}} (H_{x,w})= \lambda_{\textup{max}} (H_x)= \max_h \frac{1}{K^h(x_{[h]})} = \frac{1}{\min_h K^h(x_{[h]})}\le \frac{1}{\nu}.$$
Hence $\max_{\{x_{[h]}\in\mathcal{X}\}_{h=1}^N} \max_w \lambda_{\textup{max}} (H_{x,w}) \le \frac{1}{\nu} < \infty $.

 To  prove the second statement  we  note that following similar arguments it is possible to prove that $\nabla_x F(x)$ is a P-matrix for all $x\in\mathcal{X}$. Since $n=1$, $\mathcal{X}$ is a rectangle and the conclusion that $F(x)$ is a P-function follows by \cite[Proposition 3.5.9 (a)]{facchinei2007finite}.

\section{A novel sufficient condition for uniform P-functions}
\label{appendix:uniform}

\begin{definition}[Uniform P-matrix condition] \label{def:P_matrix_uniform}
Given a  matrix valued function $x\mapsto A(x)\in\R^{d\times d}$,  we say that  $A(x)$ satisfies the uniform P-matrix condition over $\mathcal{X}$ if there exists $\eta>0$ such that for all sets of $d$ vectors $\{x_{[h]}\in\mathcal{X}\}_{h=1}^d$ and for any $w\in\R^d$ there exists a diagonal matrix $H_{x,w}\succ 0$  such that  it holds $w^\top H_{x,w} A_x w\ge \eta \|w\|_2^2$,
where 
$$A_x:=\left[\begin{smallmatrix} [A(x_{[1]})]_{(1,:)}\\ \vdots \\  [A(x_{[d]})]_{(d,:)} \end{smallmatrix}\right]$$ 
and $\max_{\{x_{[h]}\in\mathcal{X}\}_{h=1}^d} \max_{w\in\R^d} \lambda_{\textup{max}} (H_{x,w})<\infty$.
\end{definition}

\begin{remark}\label{remark}
The uniform P-matrix condition in Definition \ref{def:P_matrix_uniform} might seem complicated at first sight. This is a slightly stronger condition than requiring that: ``$A(x)$ is a P-matrix for all $x$'' (we refer to this condition as  6').  
If $A(x)=\nabla_x F(x)$ satisfies 6' instead of 6, one can  guarantee that $F$ is a P-function, as by Definition~\ref{def:mon2} in Appendix \ref{appendix:def}, but  not necessarily that $F$ is a \textit{uniform} P-function. Consequently, condition 6'  is not sufficient to guarantee existence of a VI solution \cite[Proposition 3.5.10]{facchinei2007finite}. 
 To obtain Definition \ref{def:P_matrix_uniform} from 6', we strengthen the requirement that $A(x)$ is a P-matrix for all $x$ (or equivalently that for any $x,w\neq 0$ there exists $H_{x,w}$ such that $w^\top H_{x,w} A(x) w>0$) in two ways. First, we require a uniform curvature bound (i.e., we require $w^\top H_{x,w} A(x) w\ge \eta \|w\|_2^2$ instead of $w^\top H_{x,w} A(x) w>0$). Second, we require that this condition holds  not only for all possible $A(x)$ matrices but also for all matrices $A_x$ obtained by selecting  each row from a different matrix in the set $\{A(x_{[h]})\}_{h=1}^d$ obtained by evaluating $A(x)$ at $d$ different points $\{x_{[h]}\}_{h=1}^d$. The technical reason for this is that in the proof of  Proposition~\ref{prop:suff}c)  we need to apply the mean value theorem to each row separately.\footnote{Whether a statement as in Proposition \ref{prop:suff}c) could be proven by only adding the first condition to 6'  is an open problem. We did not investigate this point further because for the network games studied in this paper,   Definition \ref{def:P_matrix_uniform} is satisfied.}
We also note that for the case $n=1$, Definition~\ref{def:P_matrix_uniform} is different from Definition \ref{def:P_up} (see Example~\ref{ex3}). Consequently, parts b) and c) of  Proposition~\ref{prop:suff} provide different sufficient conditions even for the case $n=1$.
\end{remark}

For gradients  of affine operators, an equivalent  but simpler condition for Definition \ref{def:P_matrix_uniform} is given next. 
\begin{lemma}\label{lemma:Paffine}
Consider an operator $F(x)=Ax+a$ where $ A\in\R^{d\times d}, a\in\R^d$. Then $\nabla_x F(x)$ satisfies the uniform P-matrix condition if and only if there exists $\eta>0$ such that for any $w\in\R^d$ there exists a diagonal matrix $H_w\succ 0$ such that $w^\top H_w A w \ge \eta \|w\|^2$ and $\max_w \lambda_{\textup{max}}(H_w)<\infty$. More simply, $\nabla_x F(x)$ satisfies the uniform P-matrix condition if there exists a diagonal matrix $H\succ 0$ such that $HA+A^\top H\succ 0$.
\end{lemma}
\begin{proof}
For affine operators $\nabla_x F(x)=A$ is a matrix independent of $x$. Consequently, $A_x=A$ for all $\{x_{[h]}\}_{h=1}^d$. Hence the uniform P-matrix condition is satisfied if there exists $\eta>0$ such that for any $w\in\R^d$ there exists a diagonal matrix $H_w\succ 0$ such that $w^\top H_w A w \ge \eta \|w\|^2$. Here $\max_{\{x_{[h]}\}_{h=1}^d} \max_w \lambda_{\textup{max}}(H_{w,x})=\max_w \lambda_{\textup{max}}(H_w)$.
Note that if there exists a diagonal matrix $H\succ 0$ such that $HA+A^\top H\succ 0$, then the previous statement holds with $H_w=H$ for all $w\in\R^d$, $\eta:=\frac{1}{2}\lambda_{\textup{min}}(HA+A^\top H)>0$ and $\lambda_{\textup{max}}(H)<\infty$.
\end{proof}

For the affine case, condition 6'  in Remark \ref{remark}  requires $A$ to be a P-matrix. Equivalently,   $F(x)=Ax+a$ is a P-function if for any $w\neq 0$ there exists a diagonal matrix $H_w\succ 0$ such that $w^\top H_w A w>0$. The only difference in Lemma \ref{lemma:Paffine} is that  it requires a uniformity condition on the eigenvalues of the $H_w$.

We next report  the proof of Proposition \ref{prop:suff}c) and some comparison with previous literature results. We start with the auxiliary Lemma \ref{lemma:part2}. Theorem \ref{lem:proof} proves Proposition \ref{prop:suff}c).

\begin{lemma}\label{lemma:part2}
Consider a matrix $A_x\in\R^{d\times d}$ and suppose that there exists $\eta>0$ such that for all $w\in\R^d$ there exists a diagonal matrix $H_{x,w}\succ 0$ such that  it holds 
$w^\top H_{x,w} A_x w \ge \eta \|w\|_2^2$. Then for all $w\in\R^d$ there exists $h\in\N[1,d]$ such that 
$$ w_h [A_x w]_h \ge \frac{\eta}{d M_x} \|w\|_2^2,$$
where $M_x:=\max_{w} \lambda_{max} (H_{x,w})$.
\end{lemma}
\begin{proof}
By contradiction suppose that it is not true. That is, there exists $w\in \R^d$ such that for all $h\in\N[1,d]$

\begin{equation} w_h [A_x w]_h < \frac{\eta}{d M_x} \|w\|_2^2.\label{step1P}
\end{equation}
Consider now the product $w^\top H_{x,w} A_x w$, recall that $H_{x,w}$ is diagonal positive definite, from \eqref{step1P} we get

\begin{align*}
w^\top H_{x,w} A_x w = \sum_{h=1}^d [H_{x,w}]_h  w_h [A_x w]_h < \sum_{h=1}^d [H_{x,w}]_h  \frac{\eta}{d M_x} \|w\|_2^2 \le \sum_{h=1}^d \frac{\eta}{d } \|w\|_2^2=\eta \|w\|_2^2,
\end{align*}
where we used $ [H_{x,w}]_h \le \lambda_{max} (H_{x,w})\le M_x$ for all $h\in\N[1,d]$.
This is absurd.
\end{proof}

\begin{theorem}\label{lem:proof}
Consider a continuously differentiable operator $F(x)$ defined in $\mathcal{X}\subseteq \R^d$. If $\nabla_x F(x)$ satisfies the uniform P-matrix condition then $F(x)$ is a uniform P-function.
\end{theorem}
\begin{proof}
Let $[F(x)]_h=F^h(x)$.
Take any two vectors $x,y\in\mathcal{X}$ and for any $h\in\N[1,d]$ consider the function $\Phi_h(t):[0,1] \mapsto \R$ defined as

$$\Phi_h(t)=F^h(y+t(x-y))[x-y]_h.$$
Note that $\Phi_h(1)=F^h(x)[x-y]_h$ and $\Phi_h(0)=F^h(y)[x-y]_h$, hence

\begin{equation}\Phi_h(1)-\Phi_h(0)=(F^h(x)- F^h(y))[x-y]_h=[F(x)- F(y)]_h[x-y]_h.\label{thm7_1}\end{equation}
By the mean value theorem, there exists $\theta_h\in[0,1]$ such that

\begin{equation}\Phi_h(1)-\Phi_h(0)= \frac{d \Phi_h(\theta_h)}{dt} (1-0)= \frac{d \Phi_h(\theta_h)}{dt}.\label{thm7_2}\end{equation}
However, if we let $x_{[h]}:= y+\theta_h(x-y)\in\mathcal{X}$ we get 

\begin{equation}\frac{d \Phi_h(\theta_h)}{dt} = [x-y]_h \sum_{k=1}^d \frac{\partial F^h(x_{[h]})}{\partial x_k} [x-y]_k = [x-y]_h \sum_{k=1}^d [ \nabla_x F(x_{[h]})]_{h,k} [x-y]_k.\label{thm7_3}\end{equation}
Combining \eqref{thm7_1}, \eqref{thm7_2} and \eqref{thm7_3} we get

$$ [F(x)- F(y)]_h[x-y]_h=[x-y]_h \sum_{k=1}^d [ \nabla_x F(x_{[h]})]_{h,k} [x-y]_k.$$
Let us define the matrix

$$A_x:=\left[\begin{smallmatrix} [ \nabla_x F(x_{[1]})]_{(1,:)}\\ \vdots \\   [ \nabla_x F(x_{[d]})]_{(d,:)} \end{smallmatrix}\right]$$
then

$$ [F(x)- F(y)]_h[x-y]_h=[x-y]_h [A_x (x-y) ]_h.$$

Since $\nabla_x F(x)$ satisfies the uniform P-matrix condition, we have that for any  $w\in\R^d$ there exists a diagonal matrix $H_{x,w}\succ 0$ such that   it holds $w^\top H_{x,w} A_x w \ge \eta \|w\|_2^2$. From Lemma \ref{lemma:part2}, for any  $w\in\R^d$  there exists $h\in\N[1,d]$ such that 

$$ w_h [A_x w]_h \ge \frac{\eta}{d \max_w \lambda_{max} (H_{x,w})} \|w\|_2^2 \ge \frac{\eta}{d M} \|w\|_2^2, $$
where $M:=\max_w \lambda_{max} (H_{x,w})<\infty$. By selecting $w=x-y$ we get that there exists $h\in\N[1,d]$ such that

$$[F(x)- F(y)]_h[x-y]_h= [x-y]_h [A_x (x-y)]_h \ge  \frac{\eta}{d M} \|x-y\|_2^2. $$
Since this holds for all $x,y$ we get that $F$ is a uniform P-function with constant $\frac{\eta}{dM}>0$.

\end{proof}

Note that in \cite[Proposition 3.5.9(a)]{facchinei2007finite} it is shown that if $\nabla_x F(x)$ is a P-matrix for all $x\in\mathcal{X}$ then $F$ is a P-function (according to Definition \ref{def:mon2}).  Lemma  \ref{lem:proof} above shows that by substituting the condition ``$\nabla_x F(x)$ is a P-matrix for all $x\in\mathcal{X}$'' with ``$\nabla_x F(x)$ satisfies the uniform P-matrix condition'' one can get that $F$ is a uniform P-function which is stronger than $F$ being a P-function (see Figure \ref{fig:full}). In this work, we need the uniformity property  to prove existence of the VI solution without imposing that the set $\mathcal{X}$ is bounded.

\section{ Best response dynamics: Additional statements and proofs}
\label{appendix:br}

\begin{lemma} \label{lemma:descent}
Suppose that Assumption~\ref{cost}  holds. Let $d^i(x):=B^i(z^i(x))-x^i$ then
\begin{equation}
F^i(x)^\top d^i(x) \le - \kappa_1 \|d^i(x)\|^2_2.
\end{equation}
\end{lemma}
\begin{proof}
For simplicity set $\beta^i:=B^i(z^i(x)):=\arg\min_{y^i\in\mathcal{X}^i} J^i(y^i,z^i(x))$. Then it holds
$
\nabla_{x^i} J^i(\beta^i,z^i)^\top(y^i-\beta^i) =   F^i(\beta^i,x^{-i})^\top(y^i-\beta^i)\ge 0,  \forall y^i\in\mathcal{X}^i.
$
Select $y^i=x^i$ then
\begin{align*}
0\le& F^i(\beta^i,x^{-i})^\top (x^i-\beta^i) = [ F^i(\beta^i,x^{-i})-F^i(x^i,x^{-i})]^\top (x^i-\beta^i) + F^i(x)^\top (x^i-\beta^i) \\
&=(\beta^i-x^i)^\top D^i(\tilde x^i,x^{-i}) (x^i-\beta^i) -F^i(x)^\top d^i(x),
\end{align*} 
where $\tilde x^i=(1-t)x^i+t\beta^i$ for some $t\in[0,1]$ by the mean value theorem. Hence
$$ F^i(x)^\top d^i(x) \le  - d^i(x)^\top D^i(\tilde x^i,x^{-i}) d^i(x) \le - \kappa_1 \|d^i(x)\|^2_2.$$
\end{proof}

 \begin{lemma}[Convergence in potential games]\label{lemma:potential}
Consider a potential  game satisfying Assumption~\ref{cost}   and with strongly convex potential $U(x)$. Then the continuous BR dynamics (Algorithm 1) and the discrete sequential BR dynamics (Algorithm 2 for $\mathcal{T}_i=N(\N-1)+i$) globally converge to the unique Nash equilibrium. 
\end{lemma}
\begin{proof}
This is a well known result, see e.g. \cite{monderer:shapley:96} and \cite[Theorem 7.1.3]{sandholm2010population}. The proof is reported for completeness.
For continuous BR dynamics consider $L(x):=U(x)-U(x^\star)$ as Lyapunov function. Note that $L(x)\ge 0$ for all $x$ and $L(x)=0$ if and only if $x=x^\star$, since the Nash equilibrium is the unique minimizer of  $U(x)$. Moreover,  $\dot{{L}}(x(t))=\nabla_x U(x(t))^\top\dot x(t) = F(x(t))^\top d(x(t)) \le -\kappa_1 \|d(x(t))\|_2^2\le 0$ by Lemma \ref{lemma:descent} and $\dot{{L}}(x)=0$ if and only if $d(x)=0$. Since $d(x)=0$ if and only if $x=x^\star$ it follows by Lyapunov theorem that $x(t)\rightarrow x^\star$.\\
For discrete sequential BR dynamics note that since for any agent $i$ we have
$$ B^i(z^i(x)):=\arg\min_{y^i\in\mathcal{X}^i} J^i(y^i,z^i(x))=\arg\min_{y^i\in\mathcal{X}^i} U(y^i,x^{-i})$$
then, if $i$ is the agent updating at time $k$, it holds  
$U(x_{k+1})=U(x^i_{k+1}, x^{-i}_k)=\min_{y^i\in\mathcal{X}^i} U(y^i,x^{-i}_k) \le U(x^i_k,x^{-i}_k)=U(x_{k})$. Following the same arguments of \cite[Proposition 3.9]{tsitsiklis1989parallel} it is possible to show that every limit point of $\{x_k\}_{k=1}^\infty$ minimizes $U(x)$. Since $U(x)$ is strongly convex there is a unique minimizer which is the Nash equilibrium. Hence all limit points of $\{x_k\}_{k=1}^\infty$ coincide and equal the Nash equilibrium. We can conclude that $x_k\rightarrow  x^\star$.
\end{proof}

\begin{lemma}\label{br_closed}
If the cost function is as in \eqref{br:cost} then
$$B^i(z^i(x))=\Pi_{\mathcal{X}^i}^{Q^i}[x^i-(Q^i)^{-1}F^i(x)]:=\arg\min_{y^i\in\mathcal{X}^i}\|y^i-x^i+Q_{i}^{-1}F^i(x) \|^2_{Q^i}=\arg\min_{y^i\in\mathcal{X}^i}\frac 12 \|y^i-x^i \|^2_{Q^i}+  F^i(x)^\top (y^i-x^i)$$
\end{lemma}
\begin{proof}
If the cost function is as in \eqref{br:cost} then $F^i(x)=Q^i x^i+f^i(z^i)$ and 
\begin{align*}
B^i(z^i(x))&= \arg\min_{y^i\in\mathcal{X}^i} \frac 12 \|y^i\|^2_{Q^i} + f^i(z^i)^\top y^i = \arg\min_{y^i\in\mathcal{X}^i}  \|y^i\|^2_{Q^i} + 2[F^i(x)-Q^i x^i]^\top y^i\\
&= \arg\min_{y^i\in\mathcal{X}^i}  \|y^i\|^2_{Q^i} + 2[(Q^i)^{-1}F^i(x)- x^i]^\top Q^i y^i\\
&= \arg\min_{y^i\in\mathcal{X}^i}  \|y^i\|^2_{Q^i} + 2[(Q^i)^{-1}F^i(x)- x^i]^\top Q^i y^i + \|(Q^i)^{-1}F^i(x)- x^i\|_{Q^i}^2\\&= \arg\min_{y^i\in\mathcal{X}^i}  \|y^i-x^i+(Q^i)^{-1}F^i(x) \|^2_{Q^i}.
\end{align*}
\end{proof}

\subsection*{\textbf{Proof of Lemma \ref{lemma:descent_smon}}}
This  fact coincides with \cite[Proposition 4.1]{fukushima1992equivalent} upon noticing that $B(x)=\arg\min_{y\in\mathcal{X}} F(x)^\top (y-x)+\frac12 \|x-y\|^2_Q$ as shown in Lemma \ref{br_closed}.  The following proof is reported for completeness.
Let $h(x,y):=F(x)^\top (y-x)+\frac12 \|x-y\|^2_Q$.  It follows by Lemma \ref{br_closed} that $B(x)=\arg\min_{y\in\mathcal{X}} h(x,y)$ and thus 
$$\tilde U(x)= -h(x,B(x))= - \min_{y\in\mathcal{X}} h(x,y).$$
Consequently, 
$$\nabla_x \tilde U(x)= - \nabla_x h(x,B(x))= F(x)-[\nabla_x F(x)-Q] d(x) $$
and 
\begin{align}
\nabla_x \tilde U(x)^\top d(x)& = F(x)^\top d(x)-d(x)^\top \nabla_x F(x)^\top d(x) + d(x)^\top Q d(x) 
\\&\le  -d(x)^\top Q d(x)  -d(x)^\top \frac{\nabla_x F(x)+\nabla_x F(x)^\top}{2} d(x) + d(x)^\top Q d(x) \le- \alpha \|d(x)\|_2^2,
\end{align}
where the first inequality follows from the proof of Lemma \ref{lemma:descent}  and the second from the fact that $F$ is strongly monotone.

\subsection*{\textbf{Proof of Theorem \ref{thm:br_smon}}}
Note that if the cost functions are as in \eqref{br:cost} then
$\tilde U(x)= -F(x)^\top d(x) - \frac12\|d(x)\|^2_Q \ge  \frac12 \|d(x)\|_Q^2$
by the proof of Lemma \ref{lemma:descent} (in fact $F(x)^\top d(x)\le -  \|d(x)\|_Q^2$).  Hence $\tilde U(x)\ge 0$ and $\tilde U(x)= 0$ if and only if $d(x)=0$.  Recall that $d(x)=0$ if and only if $x$ is a Nash equilibrium. This theorem is then an immediate consequence of Lemma \ref{lemma:descent_smon} using $\tilde U(x)$ as Lyapunov function, since 
$$\dot{\tilde U}(x(t))=\nabla_x \tilde U(x(t))^\top d(x(t))   \le -\alpha \|d(x(t))\|_2^2.$$
\subsection*{\textbf{Proof of Theorem \ref{cor:br}}}
For the discrete dynamics, convergence is a consequence of Lemma \ref{lemma:br_block} and \cite[Proposition 2.1]{tsitsiklis1989parallel}, as shown in \cite[Theorem 10]{scutari2014real}. 
For the continuous dynamics note that since $B(x)$ is a block contraction, then $(1-\tau)x+\tau B(x)$ is a block contraction  for any $\tau\in(0,1]$. In fact 
let $\|v\|_{c,\textup{block}}:=\max_i \frac{\|v^i\|_2}{c_i}$ then
$$\| (1-\tau)x+\tau B(x)\|_{c,\textup{block}}\le  (1-\tau)\|x\|_{c,\textup{block}}+\tau \|B(x) \|_{c,\textup{block}} \le [(1-\tau)+\tau \delta_c]\|x\|_{c,\textup{block}} $$
and $(1-\tau)+\tau \delta_c<1$ for any $\tau\in(0,1]$ since $\delta_c<1$.
Consequently the Euler discretization of the continuous BR dynamics  which is given by 
\begin{equation}
\label{euler}
x^i_{k+1}=(1-\tau) x^i_{k}+\tau B^i(z^i_k) 
\end{equation}
converges for any $\tau\in(0,1]$. Intuitively, in the limit for $\tau\rightarrow 0$ this argument yields convergence of the  continuous BR dynamics. { A more rigorous proof is provided next.
Let $\bar x$ be the unique Nash equilibrium (recall that the  $P_\Upsilon$ condition implies uniqueness). Recall that by the proof of \cite[Proposition 42]{scutari2014real} it holds
$$\|B^i(z^i)-B^i(\bar z^i)\|_2 \le \frac{1}{\kappa_{i,i}}\sum_{j\neq i} \kappa_{i,j} \|x^j-\bar x^j\|_2.$$
Let us define   $v(t)\in\R^N_{\ge0}$ component-wise such that 
$$v_i(t)=\frac{1}{2}\|x^i(t)-\bar x^i\|^2_2\ge 0.$$
Note that each component $v_i(t)$ is differentiable and it holds
\begin{align*}
\dot v_i &=  (x^i-\bar x^i)^\top \dot x^i=  (x^i-\bar x^i)^\top  (B^i(z^i)-x^i) =(x^i-\bar x^i)^\top  (B^i(z^i)-B^i(\bar z^i)+\bar x^i-x^i)\\
&=  (x^i-\bar x^i)^\top  (B^i(z^i)-B^i(\bar z^i)) - \|x^i-\bar x^i\|_2^2  \le     \|x^i-\bar x^i\|_2  \|B^i(z^i)-B^i(\bar z^i)\|_2 - \|x^i-\bar x^i\|_2^2  \\
&\le \frac{1}{\kappa_{i,i}} \|x^i-\bar x^i\|_2 \left[  \sum_{j\neq i} \kappa_{i,j} \|x^j-\bar x^j\|_2 - \kappa_{i,i} \|x^i-\bar x^i\|_2 \right] = - \frac{2}{\kappa_{i,i}}\sqrt{v_i} \left[  -\sum_{j\neq i} \kappa_{i,j} \sqrt{v_j} + \kappa_{i,i} \sqrt{v_i} \right]  \\
&= -\frac{2}{\kappa_{i,i}} \sqrt{v_i} \left[  \Upsilon \sqrt{v} \right]_i,
\end{align*}
where $ [\sqrt{v}]_i:= \sqrt{v_i}$. Therefore  $v(t)$ satisfies the differential inequality
$$\dot v_i  \le  -\frac{2}{\kappa_{i,i}} \sqrt{v_i} \left[  \Upsilon \sqrt{v} \right]_i,\quad  v_i(0) = \frac{1}{2}\|x^i(0)-\bar x^i\|^2_2\ge 0 .$$

The matrix $\Upsilon^\top$ is a P-matrix (by assumption) and a Z-matrix (by construction). It follows by \cite[Theorem 1]{plemmons1977m} that there exists a positive diagonal matrix $\Delta$ such that 
$$\Upsilon^\top \Delta+\Delta \Upsilon \succeq \delta I$$
for some $\delta>0$. Consider the Lyapunov function $L(v)=\sum_i \frac{\kappa_{i,i}}{2} \Delta_{ii} v_i\ge 0$ and note that $L(v)=0$ if and only if $v=0$, since $v\ge 0$ by definition. Moreover
\begin{align*}
\dot{L}(v(t))&=\sum_i \frac{\kappa_{i,i}}{2} \Delta_{ii} \dot v_i(t) \le - \sum_i \Delta_{ii}  \sqrt{v_i(t)} \left[  \Upsilon \sqrt{v(t)} \right]_i  = - \sqrt{v(t)}^\top \Delta \Upsilon \sqrt{v(t)}\\& = - \sqrt{v(t)}^\top \frac{\Delta \Upsilon +\Upsilon^\top \Delta}{2}\sqrt{v(t)} \le -\frac{\delta}{2} \|\sqrt{v(t)}\|_2^2 = -\frac{\delta}{2} \|v(t)\|_1.
\end{align*}
Again, since $v\ge 0$, it holds $\|v\|_1=0$ if and only if $v=0$. This guarantees that $\dot{L}(v)<0$ for all $v\neq 0$.
Note that instead $v(t)=0$ implies $x^i(t)=\bar x^i$ hence when $v(t)=0$
\begin{align*}
\dot{L}(0)&=\sum_i \frac{\kappa_{i,i}}{2} \Delta_{ii} \dot v_i(t) = \sum_i \frac{\kappa_{i,i}}{2} \Delta_{ii} (x^i(t)-\bar x^i)^\top \dot x^i(t) =0.
\end{align*}
 Hence it follows by the Lyapunov theorem that $v(t)\rightarrow 0$. Hence $x(t)\rightarrow \bar x$.
}

\subsection*{\textbf{Proof of Theorem \ref{lemma:br_sub}}}

Consider first the unconstrained case $\mathcal{X}^i=\R$. By Lemma \ref{br_closed} if the cost functions are as in \eqref{br:cost} and $Q=I$ the continuous BR dynamics are equivalent to 
$\dot x(t)=\Pi_{\mathcal{X}}[x- F(x)]-x=x-F(x)-x=-F(x).$
The error dynamics is thus
$
\dot e(t)=- F(e(t)+x^{\star} )
$
where $e(t)=x(t)-x^{\star}$. If we linearize around $ e^\star=0$ we get
$\dot e(t)=- \nabla_x F(x^{\star} ) e(t).$
{ Since $\nabla_x F(x^{\star} ) $ is a P-matrix,} all its real eigenvalues  are positive by \cite{fiedler1962matrices}. We next show that all the eigenvalues of $\nabla_x F(x^{\star} )$ are real. Note that
$ \nabla_x F(x^{\star} )= I+K(x)G.$
Thus it suffices to show that all eigenvalues of $K(x)G$ are real. Note that $K(x)$ is a diagonal positive definite matrix and recall that $G=G^\top$. The matrix $K(x)G$ is similar to $K(x)^{\frac12} G K(x)^{\frac12}$ which is symmetric, hence all its eigenvalues are real.
Overall, all the eigenvalues of $- \nabla_x F(x^{\star} )$ are real and negative, which  is a sufficient condition for local asymptotic stability of the error dynamics.  

If constraints are present (i.e. $\mathcal{X}^i=[a^i,b^i]$ for some $a^i,b^i\in\R$) but the equilibrium is non-degenerate (i.e.  strict complementarity  is met) then for the agents whose equilibrium strategies are at the boundary  it must be either  $x^{\star,i}-F^i(x^{\star}) < a^i$ or $ x^{\star,i}- F^i(x^{\star}) > b^i$. 
In fact, consider the case $x^{\star,i}=a^i$. From the fact that this is an  equilibrium we get $x^{\star,i}=B^i(z^i(x^\star))=\Pi_{[a^i,b^i]}[x^{\star,i}- F^i(x^\star)]=a^i$ and thus $x^{\star,i}- F^i(x^\star)\le a^i$. The fact that the equilibrium  is non degenerate then leads to $x^{\star,i}- F^i(x^\star)< a^i$.
For any $x$ sufficiently close to $x^\star$, by continuity of $F$ it  holds $x^{i}- F^i(x)< a^i$ and thus $B^i(z^i(x))=\Pi_{[a^i,b^i]}[x^{i}- F^i(x)]=a^i$. Similar statements apply to the case $x^{\star,i}=b^i$.
Overall, for local perturbations the dynamics  is either $\dot x^i(t)=a^i-x^i$ or $\dot x^i(t)=b^i-x^i$ hence the strategies of the agents whose equilibrium strategy is at the boundary converge back to equilibrium  independently of $x^{-i}$. For the rest of the agents the equilibrium is interior hence locally $\Pi_{\mathcal{X}^i}[x^i- F^i(x)]=x^i- F^i(x)$. The dynamics of these agents (which we collectively index by $I$) are
$\dot x_I(t) = - F_I(x_I,x_{-I})$ and the linearized error dynamics are 
$\dot e_I(t)=- [\nabla_x F(x^{\star} )]_{I,I} e_I(t).$
Similar arguments can be applied to prove that all eigenvalues of $- [\nabla_x F(x^{\star} )]_{I,I} $ are negative since $[\nabla_x F(x^{\star} )]_{I,I}$ is a principal submatrix of a P-matrix.

\subsection{Additional examples}

{ \begin{example}[Discrete BR dynamics under strong monotonicity]\label{ex:br}
A) Consider the game in Example \ref{ex:lq} with $N=4$, $K^i=0.5$, $G=G_c$ and  $a^i=1$. We have already shown in Example \ref{ex4}  that the game Jacobian is strongly monotone and Assumption \ref{ass:gen}c is met. It is also easy to see that this game is potential since
$$K^i G_{ij}=0.5=K^j G_{ji}.$$
Figure \ref{fig:A}A) shows that the simultaneous discrete BR dynamics do not converge. The fact that the game is potential  instead guarantees that both the continuous and discrete sequential BR dynamics converge. 
 B) We next show that if we remove the potential structure this result is not guaranteed anymore. To this end, consider the exact same example but assume $K_1=-1.5$ instead of $K_1=0.5$. Note that the game is not potential anymore because
$$K^1 G_{1j}=-1.5\neq 0.5=K^j G_{j1}.$$
Nonetheless, by numerical computation of the eigenvalues of $\nabla_xF(x)+\nabla_xF(x)^\top$ one can verify that $F(x)$ is strongly monotone. Figure \ref{fig:A}B) illustrates that both simultaneous and sequential  discrete BR dynamics oscillate.
\end{example}
\begin{figure}[H]

\begin{center}
\begin{minipage}{0.05\textwidth}
A) 
\end{minipage} 
\begin{minipage}{0.9\textwidth}
\includegraphics[height=0.16\textwidth]{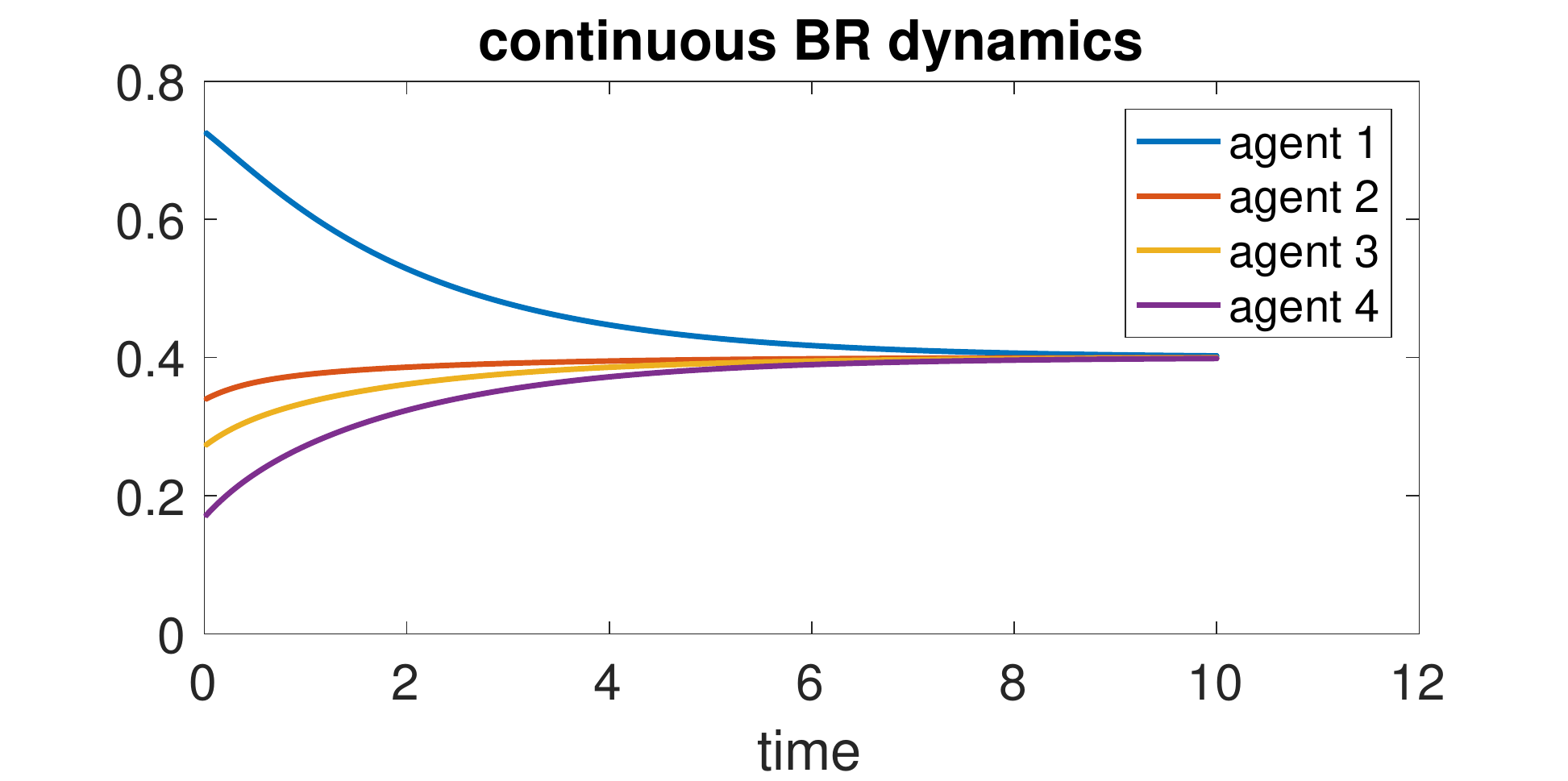}  
\includegraphics[height=0.16\textwidth]{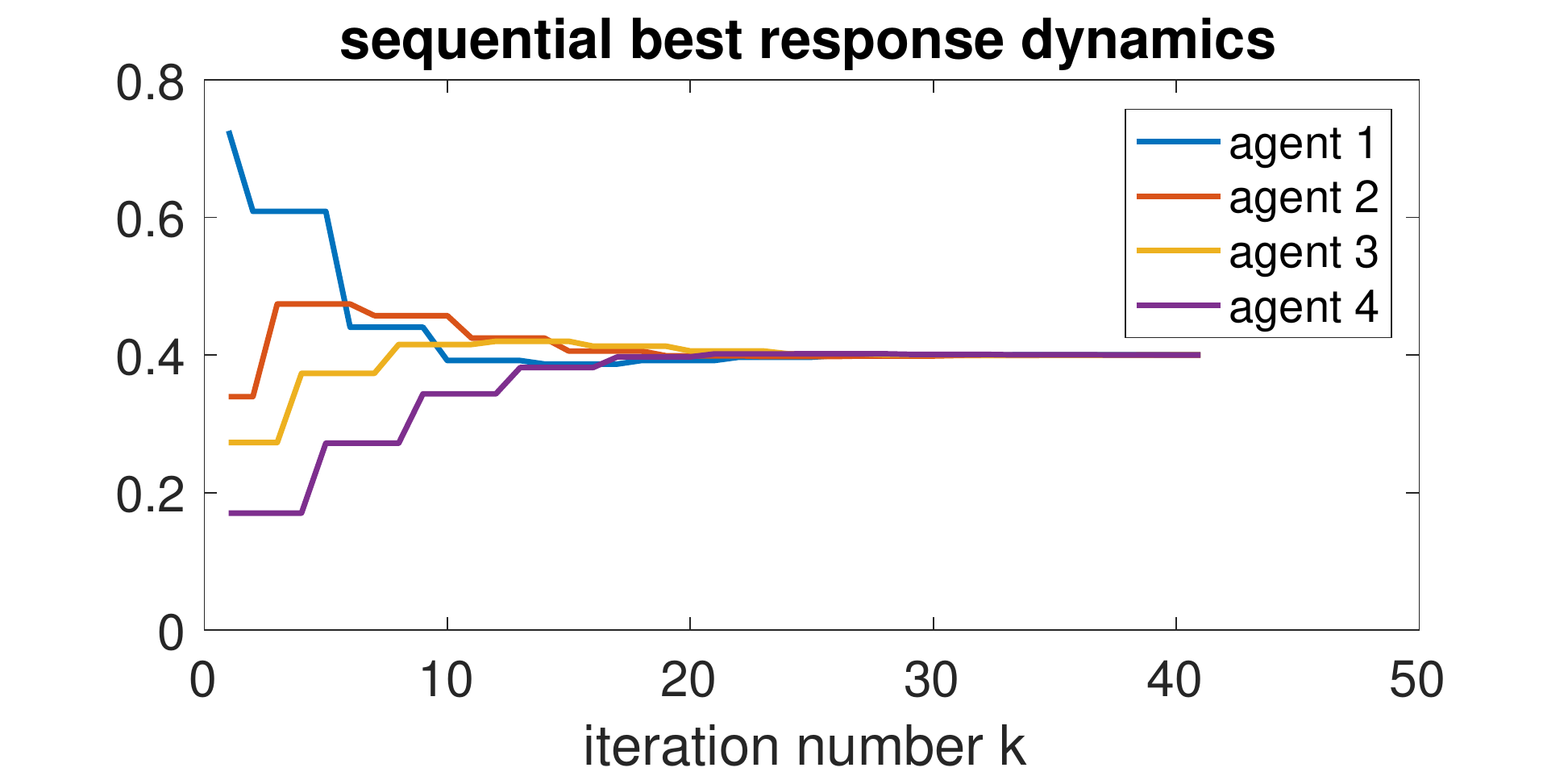}
\includegraphics[height=0.16\textwidth]{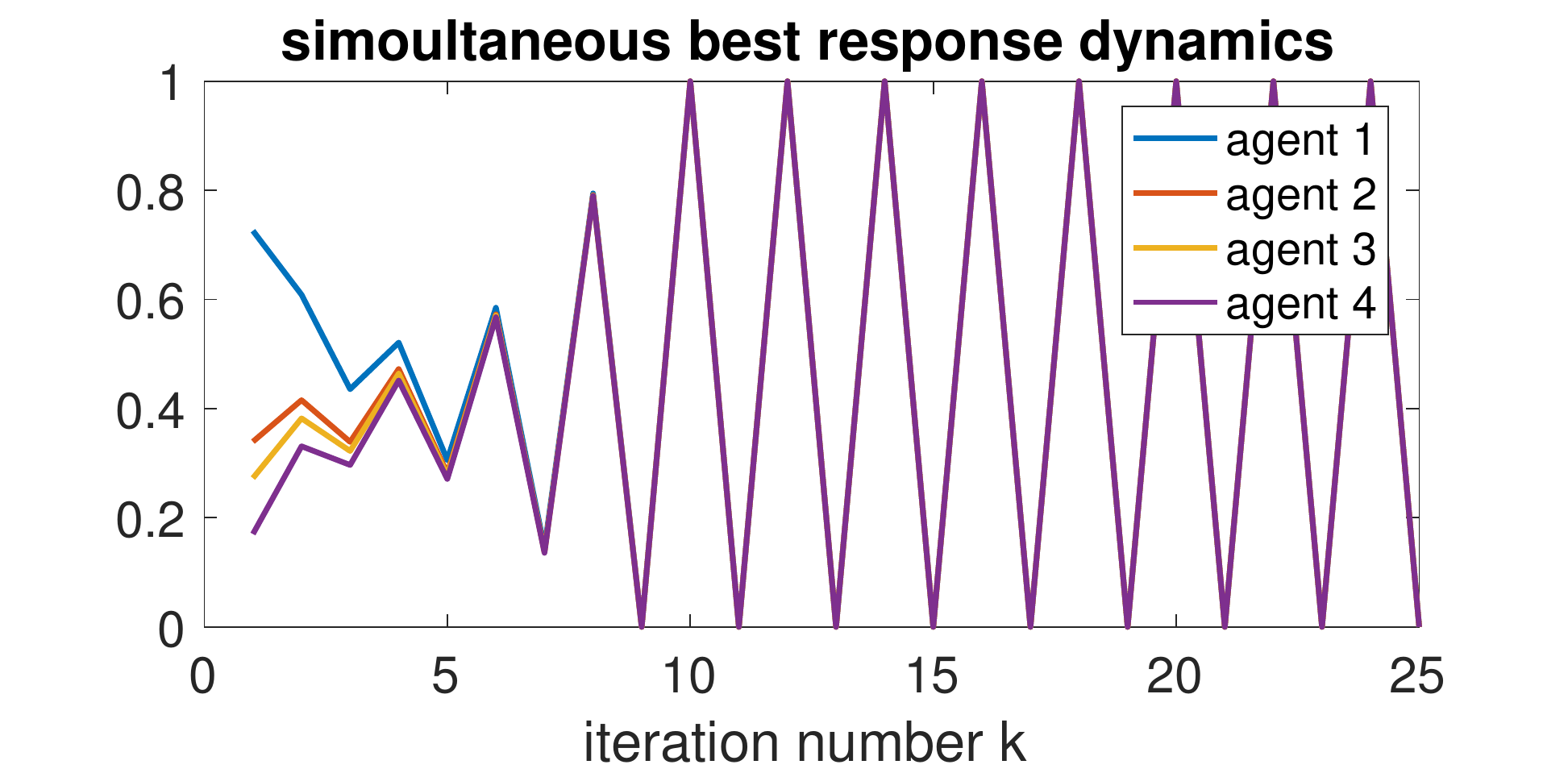}
\end{minipage} \\[0.2cm]

\begin{minipage}{0.05\textwidth}
B) 
\end{minipage} 
\begin{minipage}{0.9\textwidth}
 \includegraphics[height=0.16\textwidth]{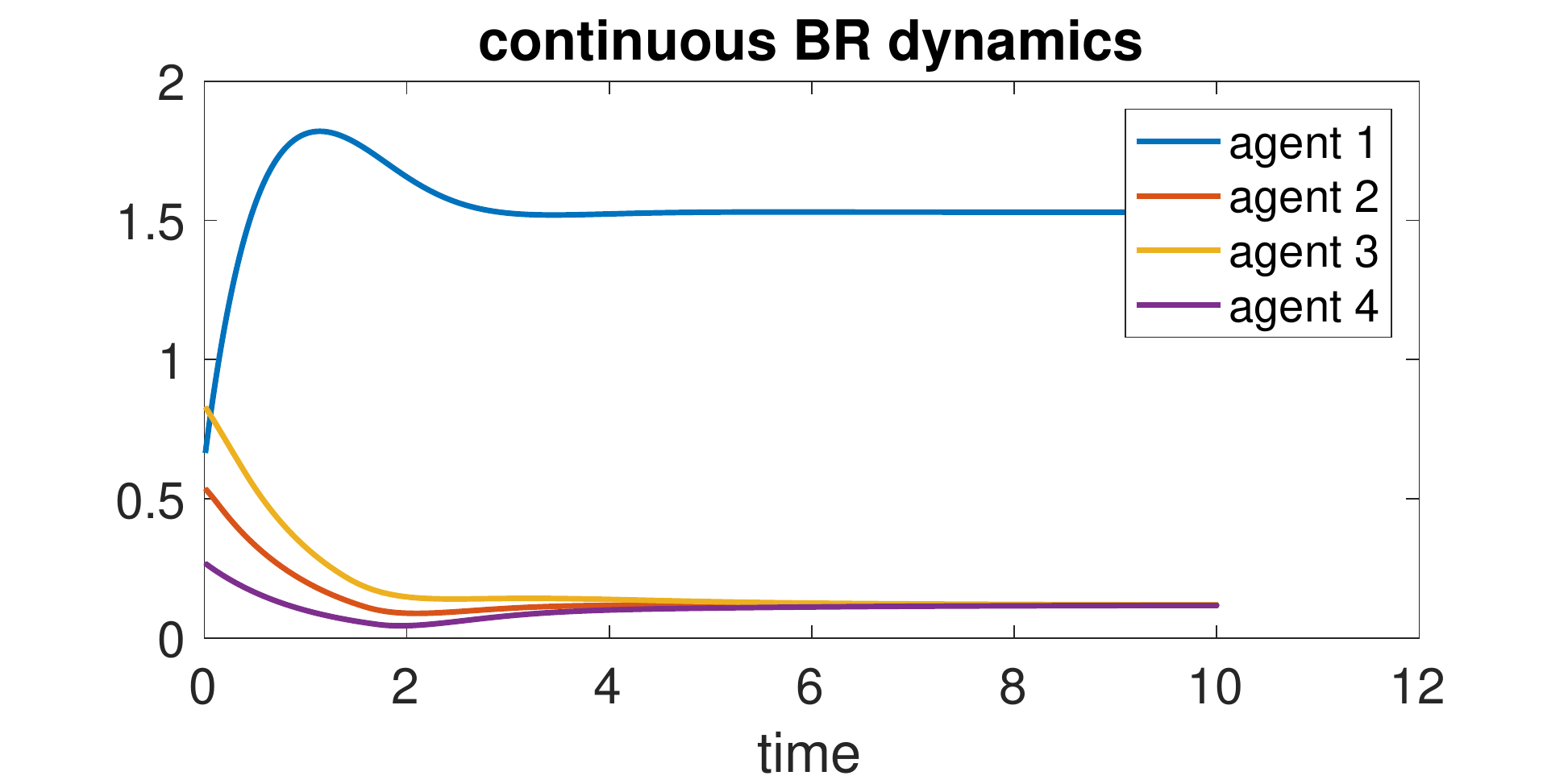}  
\includegraphics[height=0.16\textwidth]{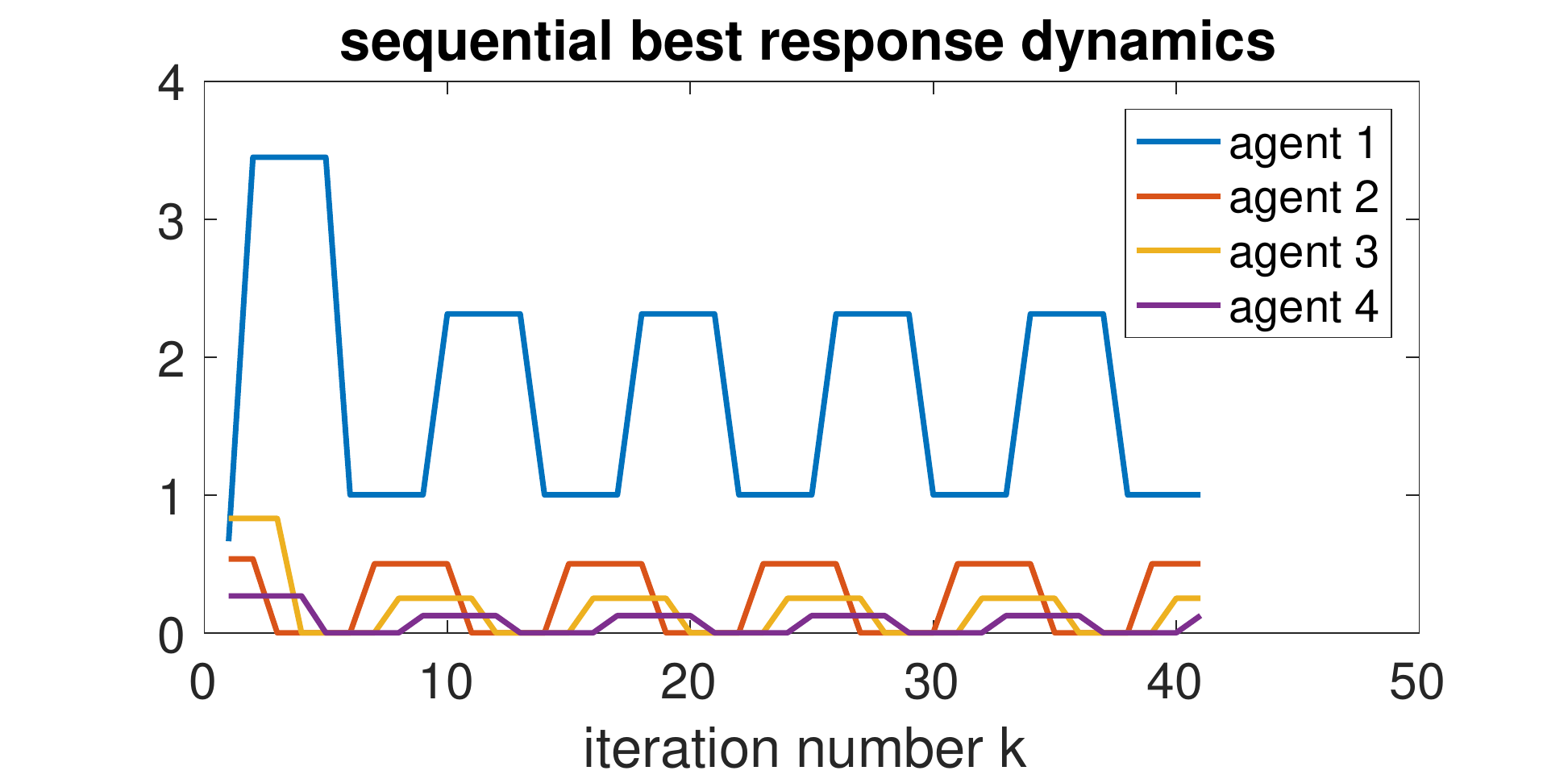}
\includegraphics[height=0.16\textwidth]{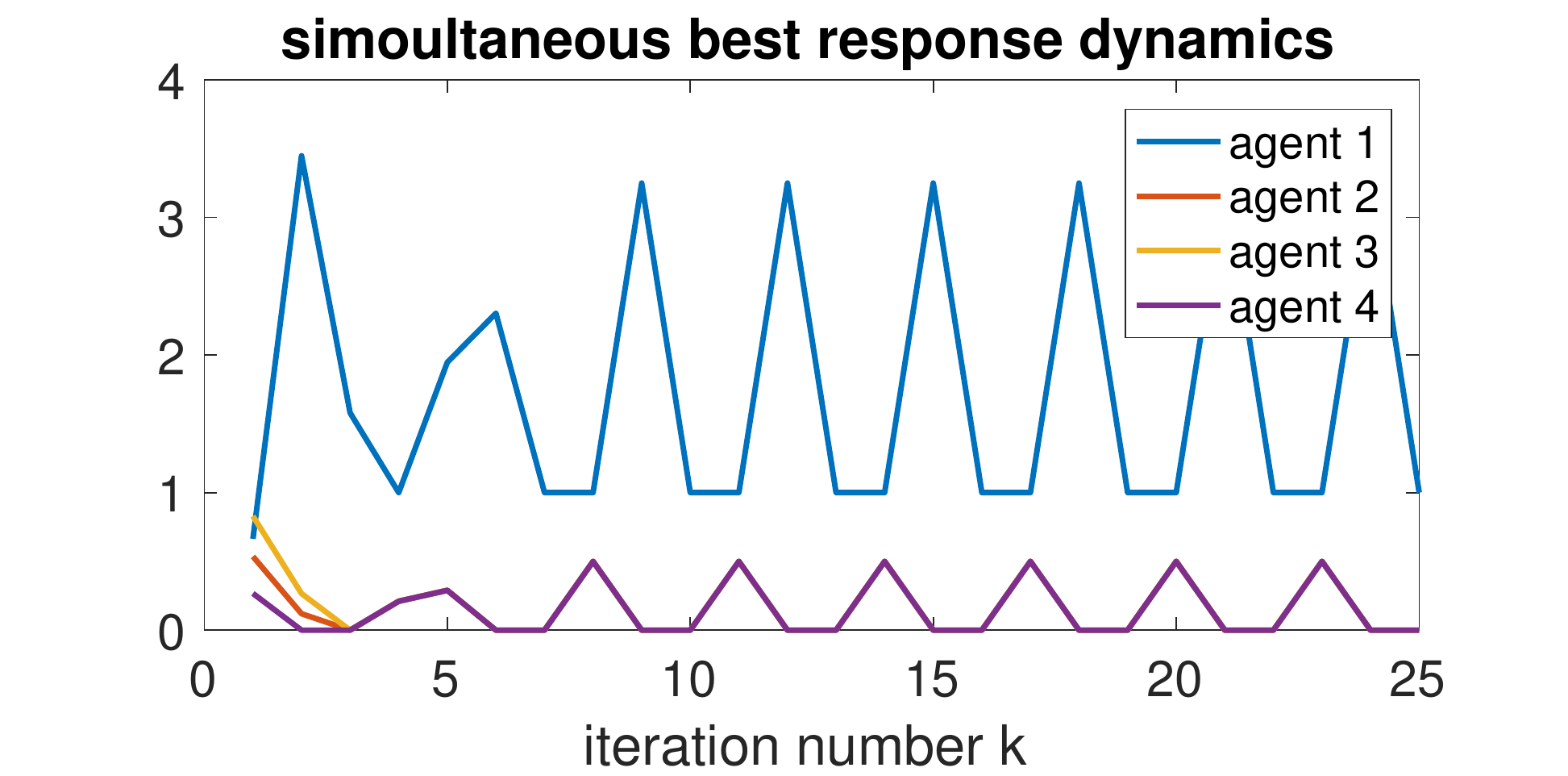}
\end{minipage} 
\end{center}
\caption{ Best response dynamics for the network games in Example \ref{ex:br}.}
\label{fig:A}
\end{figure}
}

{ \begin{example}[Strong monotonicity does not imply block-contraction]\label{ex:br2}
Consider the\\ game in Example \ref{ex:br}A) with  $\mathcal{X}^i=\mathbb{R}$ for simplicity. We have already shown therein that the game Jacobian is strongly monotone, however
the best response mapping is not a block contraction. To see this take any two vectors $x,y\in\R^N$ such that $x-y=\mathbbm{1}_N$. Then we get $\| B(x)- B(y)\|_\infty=\| K\PP_c(x- y)\|_\infty=0.5\|\PP_c\mathbbm{1}_N\|_\infty =0.5(N-1)\|\mathbbm{1}_N\|_\infty=0.5(N-1)\|x-y\|_\infty>\|x-y\|_\infty$. 
\end{example}}

\section{Comparative statics: Additional statements and proofs }
\label{appendix:comp}

\subsection*{\textbf{Proof of Theorem \ref{cor:Lipschitz}: Lipschitz continuity}}
We prove only the second statement, since the first one is similar. For simplicity let $\bar F(\cdot)=F(\cdot, \bar y,\bar \PP)$ and $ F(\cdot)=F(\cdot,  y, \PP)$. Also let $\bar x:=x^\star(\bar y, \bar G)$ be the unique solution of 
VI$(\mathcal{X}, \bar F)$ and $x:=x^\star( y, G)$ any  solution  of VI$(\mathcal{X}, F)$. 
Note that for each agent $i$,  $\bar x^i$ and $x^i$ are the best responses to $\bar x$ and $x$, respectively. Consequently, by the minimum principle for convex optimization   it holds

\begin{align*}
\bar F^i(\bar x)^\top(y^i-\bar x^i) \ge 0, \quad \forall y^i\in\mathcal{X}^i, \\
 F^i( x)^\top(y^i- x^i) \ge 0, \quad \forall y^i\in\mathcal{X}^i.
\end{align*}
By setting $y^i=x^i$ in the first line and $y^i=\bar x^i$ in the second line, we get 
\begin{align}\label{nag_step}
0\ge (\bar F^i(\bar x)-F^i( x))^\top( \bar x^i- x^i) =  (\bar F^i(\bar x) - \bar F^i( x))^\top( \bar x^i- x^i) + (\bar F^i( x) -F^i( x))^\top( \bar x^i- x^i).
\end{align}
Since $\bar F$ is a uniform block P function there exists at least one agent $i$ such that $(\bar F^i(\bar x) - \bar F^i( x))^\top( \bar x^i- x^i)\ge \bar \eta \|\bar x-x\|_2^2$. 
From \eqref{nag_step} for such agent we can write

\begin{align*}
\bar \eta \|\bar x-x\|_2^2 &\le  (\bar F^i(\bar x) - \bar F^i( x))^\top( \bar x^i- x^i) \le  
|(\bar F^i( x) -F^i( x))^\top( \bar x^i- x^i)|\\
& \le  
\|\bar F^i( x) -F^i( x)\|_2\| \bar x^i- x^i\|_2 \le \|\bar F^i( x) -F^i( x)\|_2\| \bar x- x\|_2.
\end{align*}
Dividing by $\| \bar x- x\|_2$ we get $\bar\eta \|\bar x-x\|_2 \le \|\bar F^i( x) -F^i( x)\|_2$.\footnote{Note that this is a result about the ``global'' stability of the solution $\bar x$ of a VI($\mathcal{X},\bar F)$ where $\mathcal{X}$ has a cartesian structure and $\bar F$ is a \textit{uniform block  P-function} with respect to the same partition. Stability of  VI solutions has been extensively studied, see e.g.  \cite[Chapter 5]{facchinei2007finite}.  Our proof follows the same lines as  the proof of \cite[Theorem~1.14]{nagurney2013network},  \cite[Theorem 2.1]{dafermos1988sensitivity}  where the same result is proven   under the more restrictive assumption of \textit{strong monotonicity}.} Hence

\begin{align*}
\bar \eta^2 \|\bar x - x\|^2_2& \le  \|\bar F^i( x) -F^i( x)\|_2^2 = \|\nabla_{x^i} J^i( x^i,\bar z^i( x),\bar y^i) - \nabla_{x^i} J^i( x^i, z^i( x),y^i) \|_2^2 \\
& \le  L^2 (\|\bar z^i( x) -  z^i(x) \|_2^2 + \|\bar y^i-y^i\|_2^2) \le L^2( \|\bar z( x) -  z( x) \|_2^2 + \|\bar y-y\|_2^2) \\
&\le L^2 ( \|\bar\PP- \PP\|_2^2 \|x\|_2^2 + \|\bar y-y\|_2^2)\le L^2 ( \|\bar\PP- \PP\|_2^2 \Delta^2 + \|\bar y-y\|_2^2),
\end{align*}
where for any $x\in\mathcal{X}$ we set $\bar z^i(x):=\sum_{j=1}^N \bar \PP _{ij} x^j$ and $ z^i(x):=\sum_{j=1}^N  \PP _{ij} x^j$.

{ \subsection*{\textbf{The case of  scalar non-negativity constraints}}\label{non_neg}

In the case of scalar strategies and non-negativity constraints, Assumption~\ref{ass:diff} and Proposition~\ref{thm:diff} can be  simplified. First of all, for any game $\mathcal{G}(y)$ let us call \textit{active} any agent $i\in\N[1,N]$ for which $x^{\star i}(y)>0$ and \textit{inactive}  any agent for which $x^{\star i}(y)=0$. Let $\mathcal{A}$ be the set of active players at $\bar y$.  Since only non-negativity constraints
are present $B=-I_N,b=0$ and $H,h$ are empty. This simplifies Assumption~\ref{ass:diff}. Specifically,  since $A$ is obtained by selecting only  rows from  $B=-I_N$  Assumption \ref{ass:diff}b) always holds.  If we write the KKT conditions for these constraints we have

\begin{subequations}
\begin{align}
&F(x^\star(y),y) - \lambda(y) =0 \label{KKT1nn} \\
&   \lambda(y)^\top x^\star(y)=0, \ x^\star(y)\ge 0,\   \lambda(y) \ge 0 \label{KKT2nn}.
\end{align}
\end{subequations}
Hence from \eqref{KKT1nn} we get that the dual variable $\lambda^i(y)$ associated with the constraint that the strategy of agent $i$ should be non-negative is 
$\lambda^i(y)= \nabla_{x^i} J^i(x^{\star i}(y),z^i(x^\star(y)),y^i).$
The strict complementarity condition in Assumption \ref{ass:diff}c) is then equivalent to the requirement that 

\begin{equation}\label{eq:scn}
x^{\star i}(\bar y)=0\quad \Rightarrow \quad \nabla_{x^i} J^i(x^{\star i}(\bar y),z^i(x^\star(\bar y)),\bar y^i)>0.
\end{equation}

When  Assumption \ref{ass:diff}   holds,  for any small perturbation of $y$ around $\bar y$ the set of active players does not change. This means that in a neighborhood  of $\bar y$ we can equivalently consider a reduced game $\mathcal{G}_\mathcal{A}(y)$ where we remove all the players that are inactive at $\bar y$. Let $\PP _\mathcal{A}$ be the  adjacency matrix obtained by deleting from $\PP $ all the rows and columns with indices not in $\mathcal{A}$. Any Nash equilibrium of the reduced game $\mathcal{G}_A(y)$, complemented with $x^{\star i}(y)=0$ for all $i\notin \mathcal{A}$   is  a Nash equilibrium of the full game $\mathcal{G}(y)$. Consequently, the vector of sensitivities can be computed applying formula \eqref{sensitivity} to the reduced game (for which there are no active constraints) resulting in 

\begin{align*}
\nabla_y x_\mathcal{A}^\star(\bar y) &= - [(\nabla_x F_\mathcal{A}(x,y) )^{-1} \nabla_y F_\mathcal{A}(x,y) ]_{\{x=x^\star_\mathcal{A}(\bar y),y=\bar y\}}\\
\nabla_y x_{\mathcal{A}^c}^\star(\bar y) &=0,
\end{align*}
where ${\mathcal{A}^c}$ denotes the set of inactive players at $\bar y$.
This result allows for example to derive \cite[Proposition 1]{allouch2015private}.
}

%\section*{References}

\bibliography{mit.bib}

\end{document}